\documentclass[aoas,preprint]{imsart}

\RequirePackage[OT1]{fontenc}
\RequirePackage{amsthm,amsmath,amsfonts}
\RequirePackage[authoryear]{natbib}
\RequirePackage[colorlinks,citecolor=blue,urlcolor=blue]{hyperref}

\usepackage{graphicx}              
\usepackage{subcaption}
\usepackage{float}
\usepackage{caption}
\usepackage[ruled]{algorithm}
\usepackage{soul}  
\usepackage{notoccite}

\usepackage[table]{xcolor}
\usepackage{multirow}
\usepackage{multicol}
\usepackage[font={small,it}]{caption}
\usepackage{algorithm}
\usepackage{enumerate}
\usepackage{bm}
\usepackage{bbm}
\usepackage[parfill]{parskip}

\usepackage{pdflscape}

\newtheorem*{remark*}{Remark}

\usepackage{booktabs}

\usepackage[compact]{titlesec}

\newcommand{\bft}{\fontseries{b}\selectfont}

\DeclareMathOperator{\real}{\mathbb{R}}
\DeclareMathOperator{\X}{\mathbf{X}}

\DeclareMathOperator{\y}{\mathbf{y}}
\DeclareMathOperator*{\minimize}{\mathrm{minimize}}

\newcommand{\1}{\mathbbm{1}}
\newcommand{\bbeta}{\bm{\beta}}
\newcommand{\bW}{\bm{W}}
\newcommand{\bX}{\bm{X}}
\newcommand{\bTheta}{\bm{\Theta}}

\newcommand{\by}{\bm{y}}

\newcommand{\simiid}{\buildrel{\textsc{iid}}\over\sim}
\newcommand{\lmax}{\Lambda_{\max}}
\DeclareMathOperator*{\argmin}{arg\,min}

\newcommand{\Supp}{\text{Supp}}
\newcommand{\bPhi}{\bm{\Phi}}

\startlocaldefs
\numberwithin{equation}{section}
\theoremstyle{plain}

\newtheorem{prop}{Proposition}
\endlocaldefs

\setattribute{journal}{name}{}

\begin{document}

\begin{frontmatter}
\title{Feature Selection for Data Integration with Mixed Multi-view Data}
\runtitle{Feature Selection for Mixed Multi-view Data}

\begin{aug}
\author{\fnms{Yulia} \snm{Baker}\thanksref{m1}\ead[label=e1]{yulia.baker@rice.edu}},
\author{\fnms{Tiffany M.} \snm{Tang}\thanksref{m2}\ead[label=e3]{tiffany.tang@berkeley.edu}},
\author{\fnms{Genevera I.} \snm{Allen}\thanksref{m1}\ead[label=e2]{gallen@rice.edu }}



\affiliation{Rice University\thanksmark{m1} and University of California, Berkeley\thanksmark{m2}}

\address{Department of Statistics\\
Rice University\\
Houston, TX 77005 \\
\printead{e1}}

\address{Department of Statistics\\
University of California, Berkeley\\
Berkeley, CA 94720\\
\printead{e3}}

\address{Department of Electrical and Computer Engineering\\
Rice University\\
Houston, TX 77005 \\
\printead{e2}}
\end{aug}

\begin{abstract}
Data integration methods that analyze multiple sources of data simultaneously can often provide more holistic insights than can separate inquiries of each data source. Motivated by the advantages of data integration in the era of ``big data'', we investigate feature selection for high-dimensional multi-view data with mixed data types (e.g. continuous, binary, count-valued). This heterogeneity of multi-view data poses numerous challenges for existing feature selection methods. However, after critically examining these issues through empirical and theoretically-guided lenses, we develop a practical solution, the Block Randomized Adaptive Iterative Lasso (B-RAIL), which combines the strengths of the randomized Lasso, adaptive weighting schemes, and stability selection. B-RAIL serves as a versatile data integration method for sparse regression and graph selection, and we demonstrate the effectiveness of B-RAIL through extensive simulations and a case study to infer the ovarian cancer gene regulatory network. In this case study, B-RAIL successfully identifies well-known biomarkers associated with ovarian cancer and hints at novel candidates for future ovarian cancer research.
\end{abstract}

%

\end{frontmatter}

\section{Introduction}

As the amount of data grows in volume and variety, data integration, or the analysis of multiple sources of data simultaneously, is becoming increasingly necessary in numerous disciplines. For example, in genomics, scientists can gather data from many related, yet distinct sources including gene expression, miRNA expression, point mutations, and DNA methylation. Since all of these genomic sources interact within the same biological system, it can be advantageous to analyze them together via data integration. Ultimately, the abundance and diversity of information captured by integrated data offers an invaluable opportunity to gain a better and more holistic understanding of the phenomena at hand.


In this work, we aim to perform feature selection for a common family of integrated data sets called high-dimensional \textit{multi-view} data. Multi-view data refers to data collected on the same set of samples, but with features from multiple sources of potentially mixed types (e.g. categorical, binary, count, proportion, continuous, and skewed continuous values). More formally, suppose we observe multi-view data with $K$ high-dimensional views (or sources), $\X_1 \in \real^{n \times p_1}, \dots, \X_K \in \real^{n \times p_K}$, which are measured on the same $n$ samples but with features of mixed types. We seek to recover a sparse set of features from each $\X_k$ associated with the response $\y \in \real^n$ by considering:
\begin{align} \label{eq:prob}
\minimize_{\alpha, \bbeta_1, \ldots \bbeta_K}  \ - \frac{1}{n} \ell \left(
\y ; \ \ \alpha \mathbf{1}_n + \sum_{k=1}^{K} \X_k \bbeta_k
\right) \ \ \textrm{subject to }  \ \sum_{k=1}^{K} || \bbeta_k ||_0 \leq \nu.
\end{align}
Here, $\bbeta_k \in \real^{p_k}$ are the coefficients associated with view $k$, $\nu > 0$ is a tuning parameter which regulates the sparsity level, and $\ell()$ is the generalized linear model (GLM) log-likelihood associated with $\y$. Note we not only consider continuous (Gaussian) responses, but also the broader class of GLMs including the Poisson (log-linear) and Bernoulli (logistic) families.

While there are many applications for multi-view feature selection in genomics, imaging, national security, economics, and other fields, major difficulties, stemming from the heterogeneity of features and how to appropriately integrate such differences, have prevented the successful use of multi-view feature selection in practice. To our knowledge, no one has proposed an effective practical solution to perform feature selection with multi-view data. A plethora of works have studied feature selection in the high-dimensional setting via the Lasso or GLM Lasso \citep{tibshirani1996regression, yuan2007model, tibshirani2013lasso, zhao2006model}, and others have studied various data integration problems \citep{hall1997introduction, shen2009integrative, acar2011all}. However, there is limited research at the intersection of the two fields. 

The one area that touches on multi-view feature selection is in the context of mixed graphical models, which estimate sparse graphs between features in multi-view data \citep{yang2014general, yang2014mixed, lee2013structure, cheng2013high, haslbeck2015mgm}. Using the node-wise neighborhood estimation approach of \cite{meinshausen2006high}, mixed graphical models estimate the neighborhood of each node (i.e. feature) separately via a penalized regression model (typically based on the Lasso or GLM Lasso) and combine neighborhoods using an ``AND'' or ``OR'' rule. Though mixed graphical models perform well in idealized settings for which theoretical guarantees have been proven, we will demonstrate in Section~\ref{sec:Challenges} that there are severe limitations with these approaches in realistic settings with correlated, heterogeneous features, commonly found in multi-view data.

To facilitate more effective integrative analyses in practice, we investigate the under-studied problem of high-dimensional multi-view feature selection, and we propose a practical solution. Our work is the first to identify and critically examine the fundamental challenges of multi-view feature selection, and we leverage this deep understanding of the challenges to develop a new high-dimensional multi-view selection method, the Block Randomized Adaptive Iterative Lasso (B-RAIL). B-RAIL is a practical tool for multi-view feature selection with its roots grounded in theory, and it builds upon adaptive $\ell_1$ penalties, the randomized Lasso, and stability selection \citep{meinshausen2010stability} to overcome the issues incurred by existing methods. Our method can be used for both regression and mixed graphical selection, thus lending itself to a host of important applications. 

In Section~\ref{sec:Challenges}, we investigate the major challenges of multi-view feature selection and highlight the literature gaps relating to these issues. We also show that the culmination of these challenges lead to poor feature recovery in existing Lasso-type methods and mixed graphical models. In Section~\ref{sec:BRAIL}, we introduce our proposed method, B-RAIL, which takes steps to address the challenges from Section~\ref{sec:Challenges}. In Section~\ref{sec:Sims}, we showcase the strong empirical performance of B-RAIL through simulations and contrast it to existing methods. In Section~\ref{sec:OVCase}, we further demonstrate the effectiveness of B-RAIL in a novel integrative genomics case study for ovarian cancer, and we provide concluding remarks in Section~\ref{sec:Disc}.

\section{Challenges}
\label{sec:Challenges}


Before introducing our proposed method, it is instructive to understand the challenges posed by feature selection in the multi-view setting. These challenges have been over-looked in previous methods and thus contribute to many of their shortcomings. In this section, we focus on the challenges faced by linear models with Lasso-type penalties due to their overwhelming popularity and desirable statistical properties \citep{tibshirani1996regression, yuan2007model, meinshausen2009lasso, tibshirani2013lasso, zhao2006model, zhang2008sparsity}. Given data $\X \in \real^{n \times p}$ and response $\y \in \real^n$, recall that the (GLM) Lasso solves
\begin{align}
\hat{\alpha}, \hat{\bbeta} = \underset{\alpha \in \real, \ \bbeta \in \real^p}{\operatorname{\argmin}}  \ \ - \frac{1}{n} \ell \left(\y ; \ \ \alpha \mathbf{1}_n + \X \bbeta \right) + \lambda || \bbeta ||_1, \label{eq:lasso} 
\end{align} 
where $\lambda > 0$ is a regularization parameter and $\ell()$ is the GLM log-likelihood associated with the response. For clarity, we use the term ``Lasso'' to refer to the $\ell_1$-penalized model with continuous (Gaussian) responses, ``GLM Lasso'' to mean the $\ell_1$-penalized model with non-Gaussian GLM responses (e.g. binary, Poisson), and ``Lasso-type'' methods to mean either the Lasso or GLM Lasso with some form of $\ell_1$ penalty (e.g. a global penalty, separate penalties, adaptive penalties).

Our focus here is not on deriving new theoretical guarantees for the Lasso in multi-view settings. Rather, we highlight deep practical concerns, which are rooted in theory and commonly arise in feature selection for data integration. By identifying these practical challenges, we open up numerous avenues for future theoretical research and set the stage for the construction of a new method, which overcomes the identified issues.

\subsection{Motivating Example}
To first illustrate the current challenges and motivate the need for a solution, we present in Figure~\ref{fig:MotivPlot} the estimated graphs from common Lasso-type methods and our proposed method when applied to real ovarian cancer genomics data. Here, there are $n = 293$ samples and $p = 836$ features from three views: count-valued RNASeq data ($p_{\text{RNASeq}}= 408$), continuous miRNA data ($p_{\text{miRNA}}= 307$), and proportion-valued methylation data ($p_{\text{Methyl}}= 301$) (refer to Section~\ref{sec:OVCase} for data collection and pre-processing details). As in several previous graphical models and mixed graphical models \citep{meinshausen2006high, ravikumar2010high, jalali2011learning}, we estimated the graphs using node-wise neighborhood selection. We then combined neighborhoods using the ``AND'' rule and applied stability selection \citep{meinshausen2010stability, liu2010stability} with the threshold $0.98$ to select stable edges.

Figure~\ref{fig:MotivPlot} specifically compares three types of estimation schemas at each node: (a) GLM Lasso with one global penalty, (b) GLM Lasso with separate penalties for each view, and (c) our proposed B-RAIL algorithm (introduced in Section~\ref{sec:BRAIL}). The first two methods have been proposed in several mixed graphical models \citep{yang2014general, chen2014selection, haslbeck2015mgm} and satisfy strong theoretical guarantees in idealized settings. In the real data example however, the Lasso-type methods are unstable (illustrated by the fewer edges), favor feature selection within one view, and select only a few edges between views. This overall instability indicates that the Lasso-type methods are not robust to small perturbations of the data and raises serious concerns about the reproducibility and reliability of the results \citep{yu2013stability}. Our proposed B-RAIL algorithm, in contrast, avoids these issues and exhibits greater stability as well as balance, selecting a larger number of within and between block edges under the same thresholding value. We will later see through extensive simulations in Section~\ref{sec:Sims} that the issues with existing Lasso-type methods observed here are recurring problems in very general multi-view scenarios. 

\begin{figure}[h!]
\begin{tabular}[c]{cc}
\captionsetup{type=figure}\addtocounter{figure}{-1}
\begin{subfigure}[b]{0.44\textwidth}
\centering
\includegraphics[width=1\linewidth]{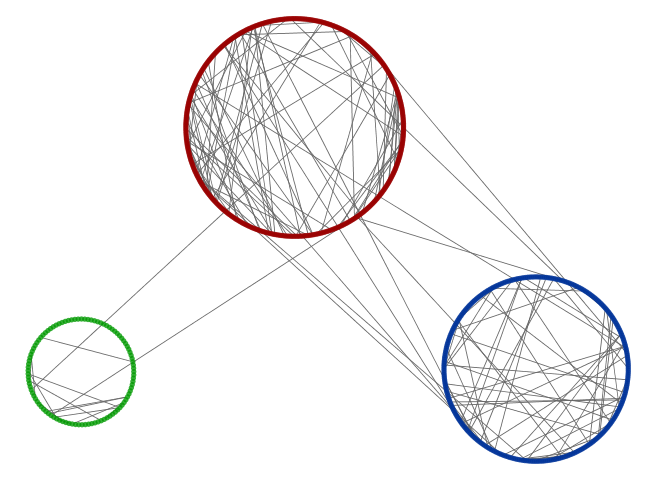}
  \captionof{figure}{Lasso with a single global penalty}
\end{subfigure}& 
\hspace{.5em}
\begin{subfigure}[b]{0.44\textwidth}
\centering
\includegraphics[width=1\linewidth]{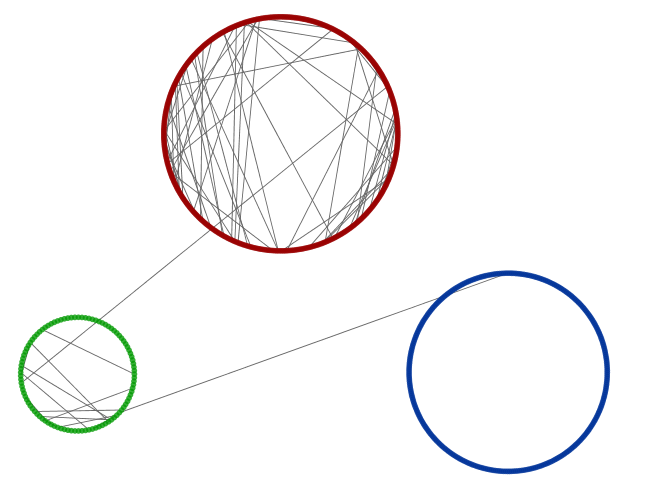}
  \captionof{figure}{Lasso with separate penalties}
\end{subfigure}\\
\multicolumn{2}{c}{
\hspace{.5em}
\begin{subfigure}[b]{0.44\textwidth}
\centering
\includegraphics[width=1\linewidth]{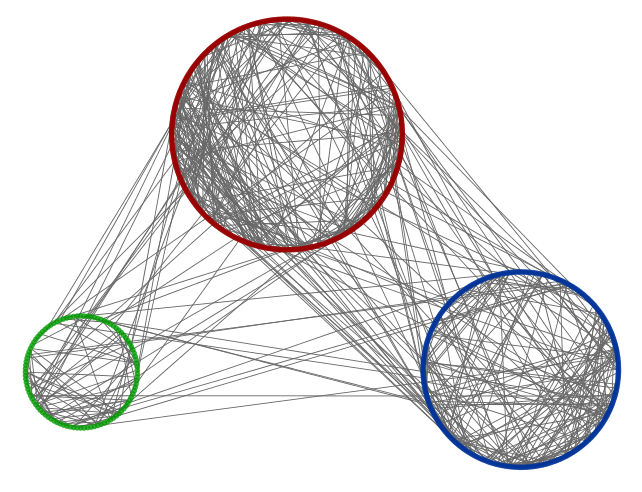}
  \captionof{figure}{The B-RAIL algorithm}
\end{subfigure} }\\
  \end{tabular}   
  \captionof{figure}{We compare three different graph selection methods when applied to real ovarian cancer genomics data. The data is comprised of three blocks: RNASeq (red), miRNA (blue), and methylation (green), with $n=293$ and $p=836$. For all three methods, we use stability selection with the threshold $0.98$ to select stable edges, and hence, we can directly compare the number of selected edges across the various methods. The Lasso with a single global penalty and the Lasso with separate penalties select few edges within blocks and almost no edges between blocks, indicating that these methods are highly unstable to small perturbations of the data.}
    \label{fig:MotivPlot}
\end{figure}

To begin understanding why existing Lasso-type methods struggle in practice, we identify and study four major challenges of feature selection for high-dimensional multi-view data: 1) scaling, 2) ultra-high-dimensionality, 3) signal interference, and 4) domain-specific beta-min. These issues stem from a combination of domain differences, signal differences, and the high-dimensionality of each view. Together, these challenges can have a significant adverse effect on feature recovery for data integration. We next examine each of these challenges in greater detail.

\subsection{Scaling}

The first and most obvious challenge with integrative analyses revolves around scaling. That is, each view in a multi-view data set is often measured on a different scale, and it is unclear how to most effectively integrate such differences. Many believe that normalizing all features to mean $0$ and variance $1$ remedies the scaling differences, but this is not always the case. Even after centering and scaling, data views remain distinct if they differ in ways beyond the first and second moments. This issue is especially problematic with binary and count-valued data blocks, two common types in multi-view data, since they are defined by much higher moments. We thus highly discourage using the ordinary (GLM) Lasso with a single penalty \eqref{eq:lasso} on normalized multi-view data.

Now while one can use different regularization parameters for each view to help alleviate the scaling differences, this generates another set of issues that are complicated by the following challenges. We will revisit the scaling issue in light of these complications later in this section.


\subsection{Ultra-High-Dimensionality}

In addition to the scaling issue, performing exact feature selection with the Lasso is already difficult in the ordinary high-dimensional setting. For exact feature selection, the number of samples $n$ must be above a theoretical minimum known as the \emph{sample complexity}. In the highly idealized scenario of an iid standard Gaussian design and a Gaussian response, \cite{wainwright2009sharp} showed that the sample complexity scales at approximately $2s \log(p - s)$, where $p$ is the number of features, and $s$ is the number of non-zero features. This idealized lower bound can be difficult to attain in many applications including genomics, where typical values of $p = 1000$ and $s = 30$ demand $n\approx 400$ patients - a large and highly expensive study. We informally refer to the regime where $p \gg n \geq 2s \log(p - s)$ as ``high-dimensional" and $n < 2 s \log(p - s)$ as ``ultra-high-dimensional." Roughly, the Lasso can never perform exact feature selection in the ultra-high-dimensional regime. 

For non-Gaussian responses and more realistic designs such as correlated, heterogeneous views in multi-view data, the sample complexity is significantly higher than the idealized Gaussian bound \citep{chen2014selection, ravikumar2010high}. As an example, the Poisson GLM's sample complexity scales at approximately $s^2 \log(p (\log p)^2)$ \citep{yang2015graphical}, so if $p = 1000$ and $s = 30$, we require $n \approx 10,000$ samples. This problem is further exacerbated in multi-view settings since combining multiple high-dimensional views for data integration almost always results in an ultra-high-dimensional problem.


\subsection{Signal Interference}

The third challenge we identify stems from a problem with the Lasso known as shrinkage noise. \cite{su2015false} showed that with high probability, no matter how strong the effect sizes, false discoveries appear early on the Lasso path due to pseudo-noise introduced by shrinkage in the high- and ultra-high-dimensional regimes. When the Lasso selects its first few features using large regularization parameters, the residuals still contain much of the signal associated with the selected features, and it is this extra noise which \cite{su2015false} calls shrinkage noise. 

In the multi-view context, shrinkage noise becomes a very complex and serious issue due to the different signals across blocks. Since the Lasso naturally selects features from the block with the highest signal first, the resulting shrinkage noise will mask the weaker signals from other blocks and compromise our ability to select from these weaker blocks. We refer to this adverse consequence of shrinkage noise as \textit{signal interference}.

\begin{figure}
\begin{tabular}[c]{cc}
\captionsetup{type=figure}\addtocounter{figure}{-1}
\begin{subfigure}[b]{0.46\textwidth}
\includegraphics[width=\linewidth]{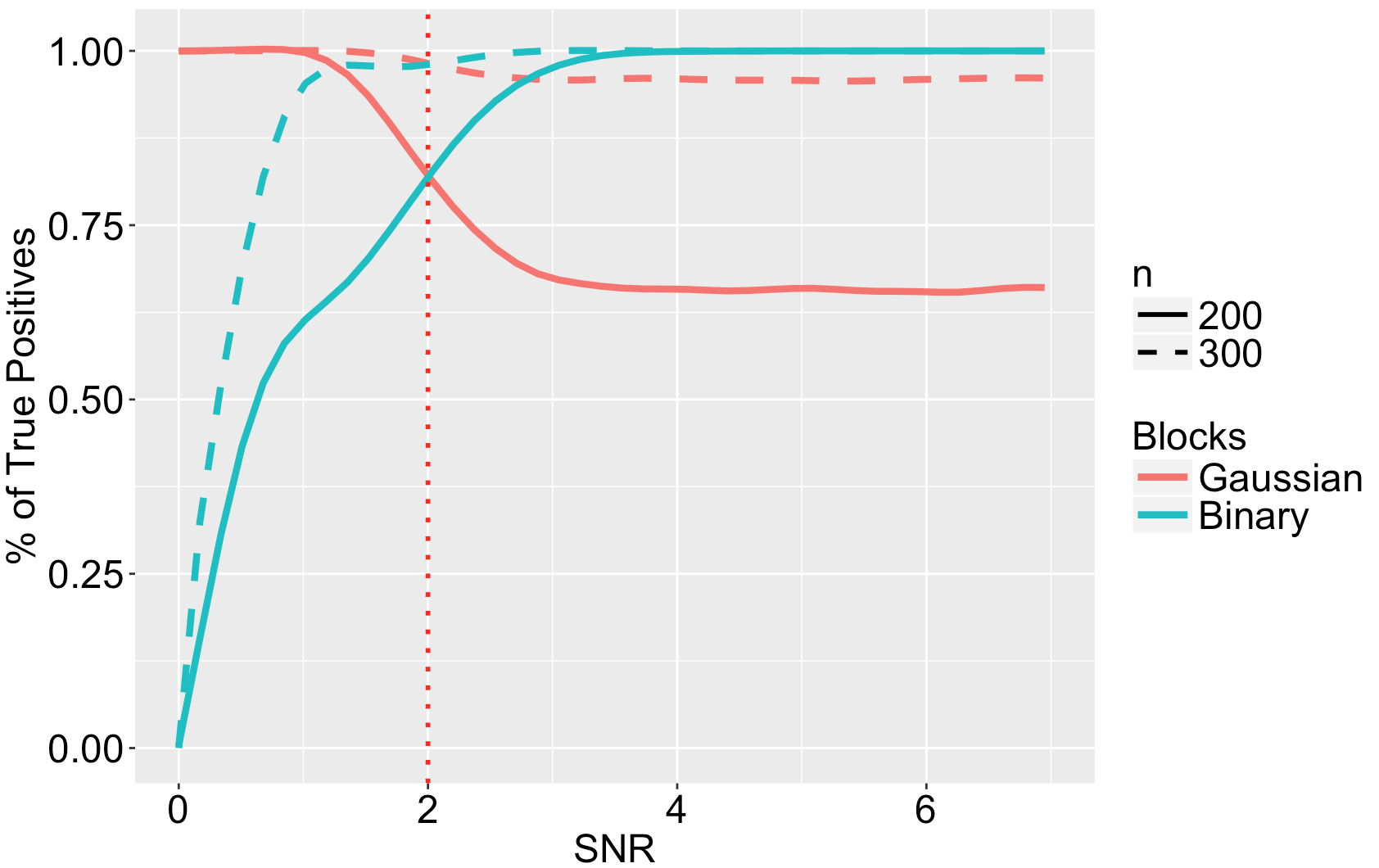}
  \captionof{figure}{ Gaussian response, Gaussian and Binary predictors}
\end{subfigure}& 
\hspace{.5em}
\begin{subfigure}[b]{0.46\textwidth}
\includegraphics[width=\linewidth]{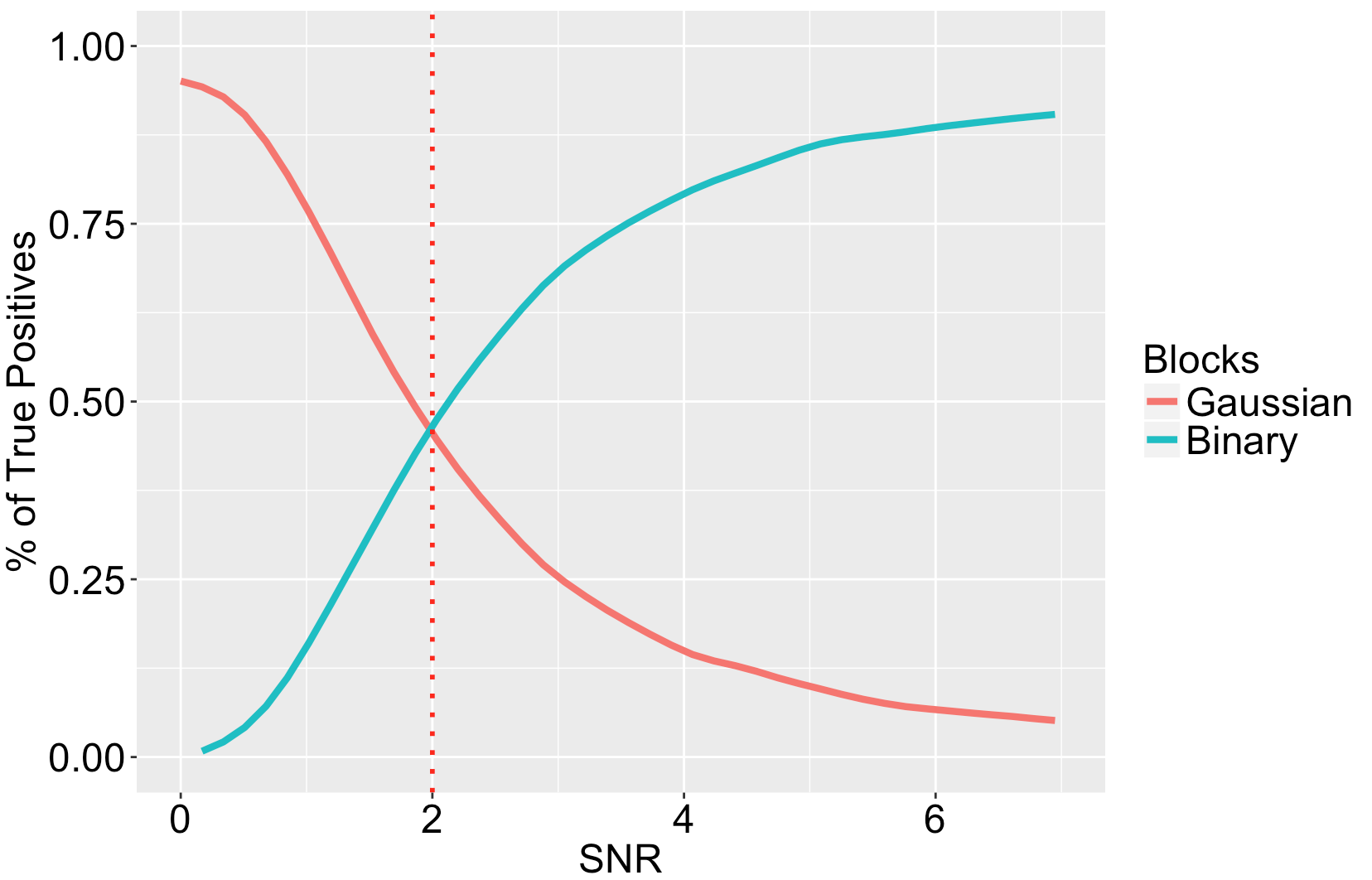}
  \captionof{figure}{ Binary response, Gaussian and Binary predictors}
\end{subfigure}\\
  \end{tabular}   
  \captionof{figure}{We illustrate signal interference for both Gaussian and binary responses given iid Gaussian $\X_1$ and iid binary $\X_2$ predictors. We simulate $n = 200$, $p_1 = p_2 = 1000$, and 10 true features in each block. We fix the SNR for the Gaussian block at $2$ and let the SNR of the binary block vary between $0$ and $7$. The vertical red line highlights the point at which $SNR_1 = SNR_2 = 2$. As the SNR of the binary block increases, it interferes with the ability to recover the true Gaussian features. This signal interference is even more severe for binary responses.}
    \label{fig:SigInterf}
\end{figure}

The problem of shrinkage noise has not been widely studied beyond the iid Gaussian design in \cite{su2015false}, but we provide strong empirical evidence in Figure~\ref{fig:SigInterf} that confirms the existence of shrinkage noise and signal interference in non-Gaussian multi-view settings. In the case of an iid Gaussian and an iid binary block, shown in Figure~\ref{fig:SigInterf}, the Lasso achieves perfect recovery in the Gaussian block when the signal-to-noise ratio (SNR) in the binary block is $0$, but as the SNR of the binary block increases, it interferes with our ability to recover the Gaussian features in the small sample scenario of $n = 200$. This signal interference is especially disastrous in the GLM Lasso with binary responses, where support recovery in the Gaussian block tends to $0$. When we increase the sample size to $n=300$ however, there is no decline in the recovery of the Gaussian block in Figure~\ref{fig:SigInterf}(a). This agrees with the known result from \cite{su2015false} that shrinkage noise occurs when the Lasso's theoretical conditions are violated and in particular, when $n$ is not sufficiently large.



\subsection{Domain-Specific Beta-min Condition}

\begin{figure}
\centering
\begin{minipage}{1\textwidth}
  \centering
  \includegraphics[width=.9\linewidth]{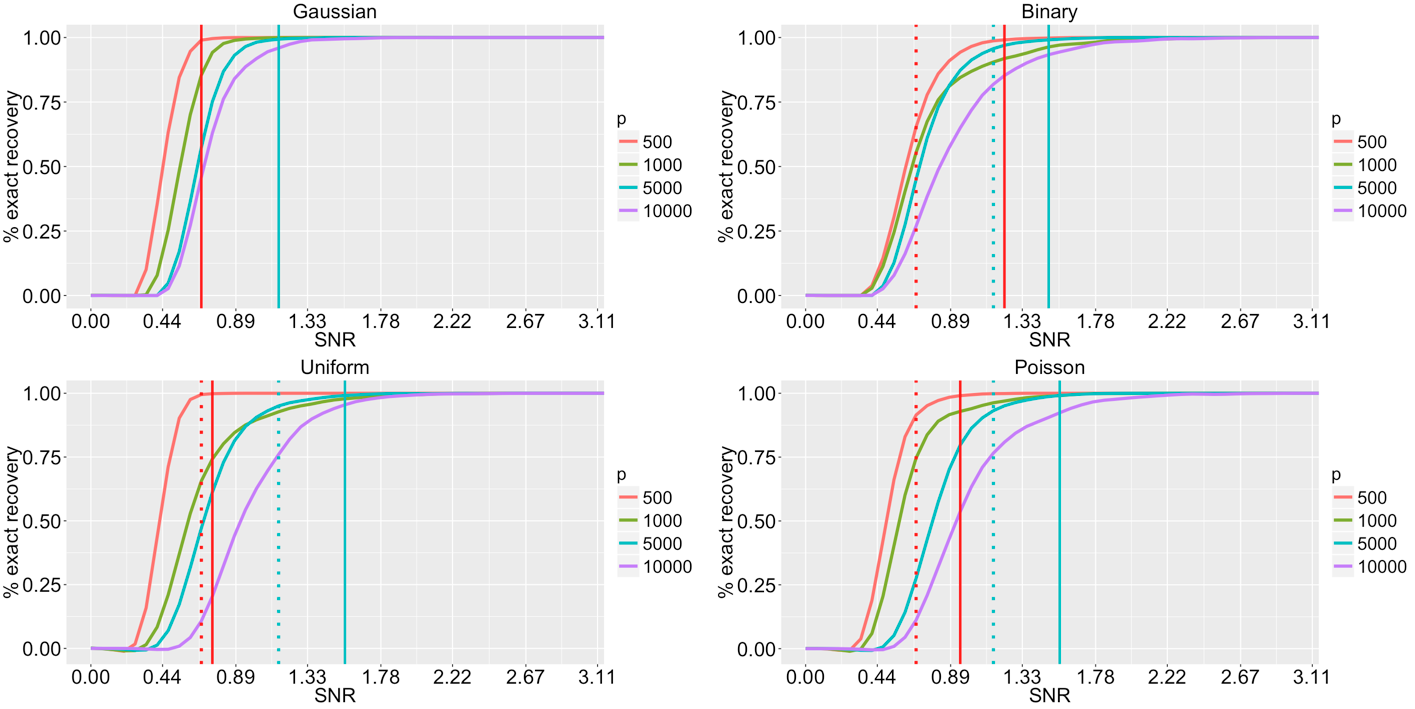}
  \captionof{figure}{We simulate Gaussian responses given four types of predictors (Gaussian, binary, uniform, Poisson) and compare our ability to recover the true features under the four designs. There are $n = 200$ samples, $p$ features, and $10$ true features. The red and blue solid vertical lines indicate the minimum SNR required to achieve $99\%$ recovery for $p=500$ and $p=5000$, respectively. In the case of non-Gaussian predictors, dashed vertical lines are overlayed to compare the minimum SNR requirements to those of Gaussian predictors. These results show that different data types can tolerate different minimum SNRs.}
  \label{fig:SNR}
\end{minipage}%
\end{figure}

Finally, analogous to how signal differences can exacerbate the Lasso's shrinkage noise issue, domain differences in multi-view problems can complicate the Lasso's \textit{beta-min condition}, which establishes a lower bound for the minimum amount of signal (i.e. SNR) required for feature recovery \citep{zhao2006model, meinshausen2006high, buhlmann2013statistical}.

In Figure~\ref{fig:SNR}, we report our ability to recover the true features for a simple simulation with iid features from four data types (Gaussian, binary, uniform, and Poisson). In each subplot, the solid vertical lines indicate the minimum SNR needed to recover $99\%$ of all true features when $p = 500$ (red) and $p = 5000$ (blue). We observe that the minimum SNR requirement varies based upon the domain of the features and that the Gaussian predictors can tolerate the lowest SNR. These empirical results reveal that if two blocks have the same amount of signal but are from different domains, we can only recover the features that pass the minimum signal threshold dictated by the domains. Put concretely, if we were to perform feature selection on our simulated multi-view data with $p = 500$ and SNR $= 0.66$, we would be able to recover $99\%$ of the true features in the Gaussian block but only about $3/4$ of the true features in the binary block. This observed phenomenon agrees with previous work, which has shown that an increase in the sparsity in $\X$ effectively reduces the SNR in the high-dimensional setting \citep{wang2010information}. Beyond this however, the beta-min condition has been relatively unexplored for the GLM Lasso and domain differences and remains a ripe area for future theoretical work.


\subsection{Additional Challenges}
It is important to note that the four challenges above do not act independently from one another. In fact, the main source of difficulty with multi-view feature selection is arguably the interactions between challenges. For instance, consider the problem of selecting features from high-dimensional discrete blocks with weak signals. The ultra-high-dimensionality issue can exacerbate the already existing problem of signal interference, which can then worsen scaling issues, increase minimum SNR requirements, and amplify the overall difficulty of the problem.

In conjunction with these complex interactions, the need to select an appropriate amount of regularization $\lambda$ through model selection methods can also increase the difficulty of multi-view feature selection. We will compare three common selection methods, namely, stability selection \citep{meinshausen2010stability, liu2010stability}, cross validation \citep{stone1974cv, allen1974cv, shao1993linear}, and extended BIC \citep{chen2012extended}, and discuss their additional challenges in Section~\ref{sec:Sims}.

We lastly note that the majority of our discussion has been focused on the Lasso. Feature selection is even more challenging for the GLM Lasso. \cite{chen2014selection} investigated this for mixed graphical models and concluded that the predictors associated with Gaussian responses are easier to recover than those with responses from other exponential families. Specifically, Gaussian responses require fewer samples, allow for a wider range of tuning parameters, and generally have a higher probability of success.

\subsection{Challenges with Existing Methods}

Having identified a host of challenges, we return to address why common Lasso-type methods are not well-suited for multi-view feature selection.

To begin with its most simple form, the \textit{Lasso with a single global penalty} \eqref{eq:lasso} uses the same penalty for all views and does not alleviate the scaling issues or signal interference issues in multi-view data. The consequences of these problems are evident, especially in the case of non-Gaussian blocks with weak signals, in Figure~\ref{fig:MotivPlot}(a), where fewer edges are selected within the proportion-valued methylation block.

By employing the \textit{Lasso with separate penalties} for each data view, we can mitigate the issue of scaling. Nevertheless, model selection becomes more challenging with multiple penalties, and signal interference remains a driver of poor recovery. In fact, having a separate penalty for each view exacerbates signal interference and encourages selection from the block with the strongest signal and no selection from the blocks with weak signals. This signal interference is exemplified in Figure~\ref{fig:MotivPlot} by the extreme selection imbalance among views, with almost no selection in the miRNA block and heavy selection in the RNASeq block.

In the \textit{Adaptive Lasso} \citep{zou2006adaptive}, the amount of $\ell_1$-regularization associated with $\beta_j$ is typically $\lambda / |\hat{\beta}_j^{OLS}|^{\gamma}$ for some constants $\gamma, \lambda > 0$. This adaptive penalty mitigates the scaling issue by adjusting for signal differences through $\hat{\beta}_j^{OLS}$, but it also encourages selection of features with higher signals and penalizes features with weaker signals. The Adaptive Lasso hence complicates signal interference by treating weaker signals as noise and results in little to no selection in the blocks with weak signals.

While the previous methods all struggle with signal interference, one simple way to reduce the signal interference between blocks is to perform \textit{separate Lassos} for each data view. Since independently-estimated blocks cannot possibly interfere with one another, this method addresses both scaling and signal interference issues. It also avoids the problem of ultra-high-dimensionality. However, each view by itself usually does not contain sufficient information to explain much of the variability in the response, and we lose the advantages of data integration.

Beyond the Lasso-type methods, there are selection methods with non-convex penalties such as SCAD and MCP \citep{fan2001variable, zhang2010nearly}. These non-convex penalties tend to scale better than the Lasso-type penalties but are still not variable selection consistent in the ultra-high-dimensional regime, especially for non-Gaussian responses and highly correlated data. We investigate MCP/SCAD feature selection in Table \ref{tab:SCAD_MCP} in the Appendix, but our primary focus in this paper is on the more commonly used Lasso-type penalties.

\section{Block Randomized Adaptive Iterative Lasso}
\label{sec:BRAIL}

Driven by the many challenges and the lack of effective tools, we propose a new method for multi-view feature selection, the Block Randomized Adaptive Iterative Lasso (B-RAIL). For the sake of notation, suppose we observe the response vector $\y$ and multi-view data $\X = [\X_1, \dots, \X_K]$ with $K$ views of potentially mixed types, $n$ samples, and $p$ total features. We will assume $p \gg n$ and typically $p_k \gg n$ for each view. Let $S$ denote the indices of the support, and let $[\X]_{S}$ denote the columns of $\X$ indexed by $S$. We will introduce B-RAIL in the context of regression and later discuss its extension to graph selection. 

Under the regression framework, the goal of B-RAIL can be viewed as two-fold: 1) to select features from each view $\X_k$ that are associated with the response $\y$, and 2) to do so while avoiding the challenges discussed in Section~\ref{sec:Challenges}. With this goal in mind, we briefly outline the B-RAIL algorithm in Algorithm~\ref{alg:alg1} and summarize the key steps taken to overcome the current challenges.

\begin{algorithm}[h]
{\fontsize{8}{8}\selectfont
\caption{Outline of Block - Randomized Adaptive Iterative Lasso}
\begin{flushleft}
\textbf{Initialize} $t=0$ and $\hat{\bbeta}^{(0)}_{k}$ to have a fixed proportion of sparsity for $k = 1, \ldots K$. \\
\vspace{1mm}
\textbf{Do} until $\Supp(\hat{\bbeta}^{(t)})$ stops changing:  \\ 
\begin{itemize} 
\setlength{\itemindent}{-1em}
\item Set $t = t+1$.
\item For $k = 1, \dots, K$, estimate $\hat{\bbeta}^{(t)}_k$ blockwise, holding $\hat{\bbeta}^{(t)}_l$ ($l < k$) and $\hat{\bbeta}^{(t-1)}_l$ ($l > k$) fixed: \\
 \vspace{3pt}
\begin{enumerate}
\setlength{\itemindent}{-1.5em}
\item Estimate the support $\hat{S}^{(t)}_k$ of block $k$:
\begin{itemize}
\setlength{\itemindent}{-2em}
\item Use stability selection with the randomized Lasso and adaptive penalties
\end{itemize}
\item Given $\hat{S}^{(t)}_k$, estimate $[\hat{\bbeta}^{(t)}_{k}]_{\hat{S}^{(t)}_k}$, the estimated non-zero coefficient values of block $k$
\end{enumerate}
\end{itemize}
\textbf{Output} $\hat{\bbeta}_1 , \ldots \hat{\bbeta}_{K}$.
\end{flushleft}
\label{alg:alg1}
}
\end{algorithm}

At a high-level, B-RAIL iterates across the data blocks $k = 1, \dots, K$ and estimates $\hat{\bbeta}_k$ for each data block $\X_k$ separately while holding all other blocks fixed. This iterative procedure is motivated by the advantages of performing separate Lassos - namely, that it mitigates the ultra-high-dimensionality and signal interference issues. Then, within each of the individual block estimations, B-RAIL first estimates the block's support and subsequently, the coefficient values given the support. Here, B-RAIL leverages ideas from adaptive weighting schemes, stability selection, and the randomized Lasso in an attempt to reduce the scaling discrepancies and domain-specific beta-min issues.

We next provide the full B-RAIL algorithm in Algorithm~\ref{alg:alg2} and proceed to discuss each step of the B-RAIL algorithm in greater detail.





\begin{algorithm}[h!]
{\fontsize{8}{8}\selectfont
  \caption{Block - Randomized Adaptive Iterative Lasso (B-RAIL)}
  \begin{enumerate}[~]
    \item \textbf{Initialization}: 
      \begin{itemize}
        \item Set $t=0$.
        \item Initialize $\hat{\bbeta}^{(0)} = \begin{bmatrix} \hat{\bbeta}^{(0)}_{1} & \cdots & \hat{\bbeta}^{(0)}_K \end{bmatrix}$, where $\|\hat{\bbeta}^{(0)}_k\|_0 \approx 0.2p_k$ for $k = 1, \ldots K$.
        \item Re-order blocks in the data $\bX$, if necessary.
      \end{itemize}
    \item \textbf{Do}: 
      \begin{itemize}
        \item Set $t = t + 1$.
        \item For $k = 1, \dots, K$, estimate $\hat{\bbeta}^{(t)}_k$ \textbf{blockwise}, holding $\hat{\bbeta}^{(t)}_l$ ($l < k$) and $\hat{\bbeta}^{(t-1)}_l$ ($l > k$) fixed:
          \begin{enumerate}[1.]
            \item Update $\hat{S}^{(t)}_k$, the estimated support for block $k$:
              \begin{enumerate}[(a)]
                \item Set \textbf{adaptive regularization}:
                  \begin{align}
                    \lambda_{k,j}^{(t)} = \begin{cases} \eta^{(t)}_k   & \text{if } \hat{\beta}^{(t-1)}_{k, j} \neq 0, \\ 
                                                                           2\eta^{(t)}_k & \text{otherwise} \end{cases} \label{eq:lambda}
                  \end{align}
                  where 
                  \begin{equation} 
                    \eta^{(t)}_k =\frac{\lmax(\widehat{\bTheta}^{(t-1)})}{\lmax(\bX^T\bX)} \  \frac{1}{\sqrt{n} } \|\hat{\bbeta}^{(t-1)}_{k}\|_{2} \sqrt{\frac{\log( p_{k})}{n}\hspace{3pt} \|\hat{\bbeta}^{(t-1)}_{k} \|_{0} } \label{eq:eta}
                  \end{equation}
                  and $\widehat{\bTheta}^{(t-1)} = \bX^T \bW(\hat{\bbeta}^{(t-1)}) \bX$ is the estimated Fisher information matrix.
                \item Perform \textbf{stability selection}: 
                  \begin{enumerate}[i.] 
                    \item Take $B$ bootstrap samples: $\{ \by^{*b}, \bX^{*b} \}_{b=1}^{B}$.
                    \item Solve the \textbf{randomized Lasso}: For each $b = 1, \dots, B$,
                      \begin{align}
                        \hat{\bbeta}_k^{(t)}(b) = \argmin_{\alpha, \bbeta}  -\frac{1}{n} \ell \left(\by^{*b}; \alpha + \bX^{*b}_k \bbeta + \bPhi^{(t)}_k(b) \right) + \sum_{j=1}^{p_k} \gamma_j\lambda^{(t)}_{k,j} |\beta_{j}| \label{eq:main} 
                      \end{align}
                      where $\gamma_j \simiid \mathcal{U}([0.5, 1.5])$ and $\bPhi^{(t)}_k(b) = \sum_{l <  k} \bX_l^{*b} \hat{\bbeta}^{(t)}_l   + \sum_{l >  k} \bX_l^{*b} \hat{\bbeta}^{(t-1)}_l $.
                    \item Select features at stability level $\tau$: 
                      \begin{align}
                        \hat{S}_{k}^{(t)} =  \left\{ j : \frac{1}{B}\sum_{b=1}^{B} \1 \left( \hat{\beta}_{k,j}^{(t)}(b) \neq 0 \right) \geq \tau \right\}. \label{eq:S}
                      \end{align}
                  \end{enumerate}
            \end{enumerate}
          \item Update $[\hat{\bbeta}^{(t)}_{k}]_{\hat{S}^{(t)}_k}$, the estimated non-zero coefficients for block $k$:
            \begin{align}
              [\hat{\bbeta}_{k}^{(t)}]_{\hat{S}^{(t)}_k} = \argmin_{\alpha, \bbeta} -\frac{1}{n} \ell \left(\by;  \alpha +  \left[\bX_k\right]_{\hat{S}^{(t)}_{k}} \bbeta + \bPhi^{(t)}_k \right) + \epsilon \|\bbeta\|_{2}^{2},  \label{eq:coef}
            \end{align}
            where $\bPhi^{(t)}_k = \sum_{l <  k} \bX_l \hat{\bbeta}^{(t)}_l   + \sum_{l >  k} \bX_l \hat{\bbeta}^{(t-1)}_l$.
        \end{enumerate}
      \end{itemize}
    \item \textbf{Until}: $\Supp(\hat{\bbeta}^{(t)}) = \Supp(\hat{\bbeta}^{(t-1)})$, where $\Supp(\cdot)$ denotes the signed support of a vector. 
    \item \textbf{Output}: $\hat{\bbeta}_{\text{B-RAIL}} = \begin{bmatrix} \hat{\bbeta}_1^{(t)}  & \cdots & \hat{\bbeta}_{K}^{(t)} \end{bmatrix}$.
  \end{enumerate}
  \label{alg:alg2} }
\end{algorithm}

\subsection{Initialization}

In our proposed B-RAIL algorithm, the coefficients are first initialized to a pre-specified sparsity level (e.g. $0.2p_k$ non-zero features per block) by fitting separate Lasso path regressions for each block. As long as the algorithm is initialized to an over-selection of the support of $\bbeta$, we have found in all of our empirical simulations, the B-RAIL algorithm tends to perform well and is very robust to the exact choice of initialization.

After initializing $\bbeta$, we must specify the order of the blocks to iterate over. This ordering can slightly alter the estimation results of B-RAIL as accurate estimation of previous blocks makes subsequent estimations of other blocks easier, but in most cases, we have found that the block ordering is not as important to B-RAIL's performance as initializing the coefficients to an over-selection of the support. Nevertheless, for best practices, since previous Lasso results guarantee a high probability of support recovery when $n$ is sufficiently large compared to $p$, we advise estimating the blocks with the smallest $p$ first, especially if $p \leq n$. If dimensions of all the blocks are of similar sizes or much larger than $n$, we recommend starting with Gaussian blocks, which tend to have better support recovery than non-Gaussian blocks.
 
\subsection{Estimating Support ($\hat{S}^{(t)}_k$)}

After initialization, we repeatedly iterate across the $K$ data blocks and estimate the support of each block separately, holding the estimates of all other blocks fixed. This block-wise estimation avoids the ultra-high-dimensionality issue, and because shrinkage noise is mainly a problem in the ultra-high-dimensional regime, the signal interference issue is also mitigated as a direct bi-product. 

Furthermore, to effectively handle correlated features in practice, we incorporate stability selection with the randomized Lasso \citep{meinshausen2010stability} to estimate each block. As given by step 1(b) in Algorithm~\ref{alg:alg2}, we solve the Lasso $B$ times using the bootstrap and randomized penalty terms, and we threshold the stability score at $\tau$ \eqref{eq:S} to select the most stable features. Though $\tau \in (0, 1)$ is a user-specified hyperparameter, the B-RAIL algorithm is insensitive to choices of $\tau$ within reasonable ranges. This insensitivity to $\tau$ has also been observed in previous work on stability selection \citep{meinshausen2010stability}. Ultimately, by utilizing randomized penalties and stability selection when estimating the support of each block, B-RAIL leverages the key property that the randomized Lasso is feature selection consistent even when the Lasso's irrepresentable condition is violated \citep{meinshausen2010stability} and hence can effectively handle correlated features. 


\subsection{Adaptive Regularization ($\lambda$)}
Now, looking more closely at the penalty term in \eqref{eq:main} of the randomized Lasso, the penalty term includes a random weight $\gamma$ like the original randomized Lasso. However, in order to account for the scaling discrepancies, signal variability, and domain differences between blocks, we introduce a block-specific adaptive penalty $\lambda$ in \eqref{eq:main} as well. For feature $j$ in block $k$, we define the adaptive weight
\begin{align*}
\lambda_{k,j}^{(t)} = \begin{cases}
\eta & \textrm{if }
\hat{\beta}^{(t-1)}_{k,j} \neq 0, \\ 
2  \eta & \textrm{otherwise}
\end{cases}  \hspace{10pt}  \tag{\ref{eq:lambda}}
\end{align*}
where
\begin{align*} \eta=\underbrace{\frac{ \Lambda_{\max}(\widehat{\boldsymbol{\Theta}}^{(t-1)})}{\ \Lambda_{\max}(\X^T\X)}}_{\substack{(a) \text{ domain} \\ \text{correction}}} \  \underbrace{\frac{1}{\sqrt{n} } || \hat{\bbeta}^{(t-1)}_{k} ||_{2}}_{\substack{(b) \text{ signal} \\ \text{correction}}} \ \underbrace{\sqrt{\frac{\log( p_{k})}{n}\hspace{3pt} \| \hat{\bbeta}^{(t-1)}_{k} \|_{0} }}_{\substack{(c) \text{ Lasso} \\ \text{penalty}}}. \tag{\ref{eq:eta}}
\end{align*}
Here, $\widehat{\bTheta}^{(t-1)}$ is the Fisher information matrix corresponding to the GLM of the response $\y$, and $\Lambda_{\max}$ denotes the maximum eigenvalue.

In this definition of $\lambda$, there are two moving parts. First, the multiplicative scheme in \eqref{eq:lambda} encourages previously selected features to remain selected while still allowing all features to freely enter or exit the model. Secondly, $\eta$ accounts for the heterogeneity of multi-view data and helps to mitigate the challenges of Section~\ref{sec:Challenges}.

Though the exact form of $\eta$ was derived experimentally, $\eta$ can be interpreted as the product of three factors, each of which is rooted in solid theoretical foundations. Namely, part (c) of \eqref{eq:eta} is closely related to the theoretical bound on the regularization parameter needed for selection consistency of the Lasso \citep{zhao2006model, meinshausen2006high}. The ratio of eigenvalues in part (a) (i.e. the domain correction term) is motivated by the theoretical conditions imposed on the Fisher information matrix for exponential family distributions \citep{yang2015graphical}, and part (b) of \eqref{eq:eta} (i.e. the signal correction term) can be viewed as the average signal in block $k$ since $\frac{1}{\sqrt{n}}$ is derived from the theoretical sparsity level within each block \citep{bunea2007sparsity}. 

By constructing the adaptive penalty $\eta$ in this way, B-RAIL accounts for different block sizes through the $\frac{\log(p_k)}{n}$ term and automatically penalizes non-Gaussian blocks less heavily than Gaussian blocks since $\Lambda_{\text{max}}(\hat{\Theta})$ is larger for Gaussian blocks. This helps to balance the inherently different beta-min conditions. In addition, because $||\hat{\bbeta}_k||_2$ captures information about both the signal and scale of the $k^{th}$ block, $\eta$ addresses the scaling differences and penalizes the stronger signal blocks more heavily to allow for the possibility of selection from weaker blocks.

While in theory this specific combination of weights should correct for scaling and domain-specific beta-min differences across views, we reinforce our choice of $\eta$ through strong empirical results in Section~\ref{sec:Sims}. We also note that even if the form of $\eta$ is slightly misspecified, stability selection is known to be fairly robust to the exact amount of regularization, as long as the amount of regularization is within reason \citep{meinshausen2010stability}. Incorporating stability selection with the randomized Lasso thus serves as a built-in check within B-RAIL, which is advantageous in practice.

\subsection{Coefficient Estimations} 
After estimating the block-wise support using the randomized Lasso with adaptive weights, we seek to estimate the coefficients of the support as accurately as possible since these values are used in future block estimations and iterations of B-RAIL. We hence refit a penalized regression model with a small ridge penalty (e.g., $\epsilon \approx 10^{-4}$) in \eqref{eq:coef} to avoid the known bias issues with the Lasso. The only reason to include the tiny ridge penalty is to ensure that we can still estimate coefficients when the selected support is greater than $n$. The exact choice of $\epsilon$ has a negligible impact in practice because it is chosen to be so small.

\subsection{Convergence}

We finally declare convergence of B-RAIL's iterative block estimation procedure when the estimated support remains unchanged. Our empirical analysis indicates that B-RAIL has quick support convergence, and we provide one example of this fast convergence in Figure~\ref{fig:Conv} in the Appendix. Using the ovarian cancer simulation (see Section~\ref{sec:OVCase}) for three different responses (Gaussian, binary, and Poisson), we report that the average number of iterations until convergence is between 4 and 5 with the maximum number of iterations reaching 15 (over 100 runs). These ranges are similar for all designs, empirically demonstrating B-RAIL's fast convergence.

Though convergence of the full-blown B-RAIL algorithm is currently limited to empirical analysis, B-RAIL can be viewed as a block coordinate descent algorithm, which can be studied theoretically under some simplifications to gain additional insights into the full B-RAIL algorithm. Namely, if we omit the adaptive regularization parameter and apply the ordinary (GLM) Lasso in each block of the algorithm (as detailed in Algorithm~\ref{alg:block_lasso} in the Appendix), we call the resulting algorithm the blockwise (GLM) Lasso and discuss its convergence below.

\vspace{-4mm}

\begin{prop}\label{prop:converge}
Consider the optimization problem:
\begin{align}
\hat{\alpha}, \hat{\bbeta} = \underset{\substack{\alpha_1, \dots, \alpha_K \in \real, \\ \ \bbeta_1, \dots, \bbeta_K \in \real^p}}{\operatorname{\argmin}} - \frac{1}{n} \sum_{k=1}^{K} \ell \left(\y ; \ \ \alpha_k \mathbf{1}_n + \X_k \bbeta_k \right) + \sum_{k=1}^{K} \sum_{j = 1}^{p_k} \lambda_{k, j} | \bbeta_{k, j} |, \label{eq:block_lasso} 
\end{align}
where $\hat{\alpha} = [\hat{\alpha}_1, \dots, \hat{\alpha}_K]$, $\hat{\bbeta} = [\hat{\bbeta}_1, \dots, \hat{\bbeta}_K]$, and $\lambda_{k, j} \geq 0$, and suppose that the objective function in \eqref{eq:block_lasso} is bounded below. Then the blockwise (GLM) Lasso converges to a global solution of\eqref{eq:block_lasso}.
\end{prop}

\vspace{-4mm}

To prove Proposition~\ref{prop:converge}, we leverage the block coordinate descent view of the blockwise Lasso and apply Theorem 4.1 from \citet{tseng2001convergence} since the objective function is convex and separable with respect to each block $\bbeta_k$. The detailed proof is provided in the Appendix. 

The two main differences between the blockwise Lasso and B-RAIL are the adaptive regularization parameter $\eta$ and the use of stability selection with the randomized Lasso in each block's update. If the estimated support from stability selection converges to a common support across the many block iterations in B-RAIL, then B-RAIL reduces to the blockwise Lasso algorithm, for which we have shown convergence. While there is empirical evidence to believe that stability selection, applied iteratively with the adaptive regularization parameter as in B-RAIL, converges to a common support, proving this theoretically is challenging due to the purely algorithmic and random nature of stability selection. In addition, existing optimization-theoretic frameworks cannot handle adaptive parameters ($\eta$) that are dependent on previous iterates ($\hat{\bbeta}^{(t-1)}$). Developing such a framework to handle both of these issues is  beyond the scope of this work, but we plan to further investigate it in the future. For now, due to these serious difficulties, we rely on our empirical analysis to demonstrate B-RAIL's quick convergence.

\subsection{B-RAIL Summary}
While we have introduced B-RAIL under the regression framework, B-RAIL can be naturally extended to estimate mixed graphical models via a penalized node-wise regression approach \citep{meinshausen2006high}. As in the motivating example in Section~\ref{sec:Challenges}, we can use B-RAIL to estimate the neighborhood of each node separately via penalized regressions and then combine the neighborhoods using an``AND'' or ``OR" rule to obtain the graph.

In either the regression or graph selection setting, our B-RAIL algorithm deliberately takes steps to exploit the practical advantages of existing Lasso-type methods while avoiding the drawbacks described in Section~\ref{sec:Challenges}. For instance, by performing iterative block-by-block estimations, B-RAIL inherits the advantages of performing separate Lassos and avoids the issue of ultra-high-dimensionality. This in turn reduces signal interference between blocks since shrinkage noise is only a concern when $n$ is not sufficiently large relative to $p$. Furthermore, we mitigate the scaling and beta-min problems by engineering adaptive $\ell_1$ penalties in B-RAIL to correct for domain and signal differences between blocks. In this construction, slightly weaker non-Gaussian blocks are penalized less heavily and thus not completely overshadowed by Gaussian blocks. Still, selecting an appropriate amount of $\ell_1$ regularization is challenging in practice, especially due to highly correlated data. B-RAIL thus incorporates randomized stability selection, which is known to be feature selection consistent under stronger and more complex dependencies than can be handled by the Lasso. This boosts the support estimation of correlated features, and together, with the previous components, B-RAIL effectively overcomes the many practical challenges of multi-view feature selection and lends itself to a plethora of data integration applications.

\section{Numerical Studies}
\label{sec:Sims}

We next reinforce the theoretically-guided choices in our B-RAIL construction and demonstrate its effectiveness through extensive simulations. In these simulations, we evaluate B-RAIL against four common Lasso-based parametric methods: (i) Lasso with a global penalty for all blocks, (ii) Lasso with separate penalties for each block, (iii) separate Lassos for each block, and (iv) Adaptive Lasso. For the Adaptive Lasso, we use ridge weights as they are better adapted to handle correlated features. Moreover, to avoid biases from penalty selection methods, we use oracle information to select features in the Lasso-based models. That is, if $k$ is the number of true features in the simulation, we fit the full path of the Lasso and select the first $k$ features. In the case of the Lasso with separate penalties, we select the $k$ features with the largest number of true positives. We do not, however, use oracle information for B-RAIL. Instead, B-RAIL internally selects the number of features using stability selection with the threshold $\tau = 0.8$ as outlined in Algorithm~\ref{alg:alg2}, and we set $\epsilon = 0.001/p$.

To systematically compare these methods, we simulate from various designs of $\X$ with three blocks - namely, a Gaussian $\X_1$, Bernoulli $\X_2$, and Poisson $\X_3$ block - and various types of GLM responses $\y$. Due to the popular use of the Gaussian, Bernoulli, and Poisson GLMs, we run simulations with responses $\y$ from each of these families. For the Gaussian response, we fit the linear model $\y = \X\beta + \epsilon$, where $\epsilon \sim N(0,1)$. For the binary and Poisson responses, we use copula transformations \citep{nelsen1999introduction} to simulate $\y$. 

In addition to these response models for $\y$, we consider four simulation designs for $\X$ to understand model behavior under different assumptions. The four simulations designs are: (i) iid features, (ii) independent features with non-constant variance, (iii) correlated features with covariance structure from a Block Directed Markov Random Field, and (iv) a real data inspired simulation with features from The Cancer Genome Atlas (TCGA) ovarian cancer study. We elaborate on each of these designs below. 

Note in all of the simulations, we set the number of true features in each covariate block to 10, and the magnitudes of the true features are drawn from $Unif(4,10)$ with random sign assignment. However, for the Gaussian block in the non-iid simulations below, we artificially lowered the SNR since we know that recovering continuous features is easier than recovering non-continuous features (see Figure~\ref{fig:SNR}). Unless stated otherwise, we simulate $n = 200$ samples and $p_1 = p_2 = p_3 = 300$ features. We also center and scale the design matrix $\X$ before estimation.

\textit{iid Design.} \hspace{.5mm} For each of the three covariate blocks, we simulate $n = 200$ samples from iid features. Here, $p_1=p_2=p_3=300$ for the high-dimensional design and $p_1=p_2=p_3=100$ for the low-dimensional design. 

\textit{Heteroscedasticity Design.} \hspace{.5mm} In this design, we assume that the features are independent but have non-constant variance. For the Gaussian block, the entries in each column are simulated from the normal distribution $N(0, \sigma^2)$, where $\sigma \sim \Gamma(3, 0.6)$. In the Bernoulli block, each column is simulated independently with entries drawn from $Bern(p)$ with $p \sim Unif(0.2, 0.8)$. Similarly, in the Poisson block, the mean $\lambda$ of each column is drawn from the Gamma distribution $\Gamma(4, 0.6)$ (using the shape/scale parameterization).

\textit{Block Directed Graph Design.} \hspace{.5mm} We next drop the independence assumption and use a Block Directed Markov Random Field (BDMRF) \citep{yang2014general} graph to simulate correlated features. In this case, $\X$ is simulated via Gibbs sampling with the partial ordering of the underlying mixed graph given by $P[X_1, X_2, X_3 ] = P[X_1 |X_2,X_3]P[X_2|X_3]P[X_3]$, where $P[X_1 |X_2,X_3]$ is a pairwise Gaussian conditional random field (CRF), $P[X_2|X_3]$ is a pairwise Ising CRF \citep{ravikumar2010high}, and $P[X_3]$ is a pairwise Poisson Markov Random Field (MRF) \citep{yang2013poisson, yang2012graphical}. We set high correlations for the Gaussian and Poisson blocks and low correlations for the binary block and between block structure.  

\textit{Ovarian Cancer Inspired Simulation Design.} \hspace{.5mm} In an attempt to simulate data closest to real world scenarios, we take the continuous-valued miRNA data, proportion-valued methylation data, and the count-valued RNASeq data from The Cancer Genome Atlas (TCGA) ovarian cancer database \citep{cancer2011integrated} to be our covariates. After merging and preprocessing the TCGA ovarian cancer data (refer to Section~\ref{sec:OVCase} for details), we arrive at $n = 293$ samples and $p_{\text{RNASeq}}= 408$, $p_{\text{miRNA}}= 307$, and $p_{\text{Methyl}}= 301$ features. 

\begin{table*}[h!]
\makebox[\linewidth]{
\scalebox{0.7}{
\begin{tabular}{l  ll  ll  ll  ll}

\addlinespace
\multicolumn{9}{c}{ \begin{normalsize}  iid Case, p=300  \end{normalsize} } \\
\midrule

 & \multicolumn{2}{c}{Total} &  \multicolumn{2}{c}{Continuous} & \multicolumn{2}{c}{ Binary} & \multicolumn{2}{c}{Counts}  \\ 
 \cmidrule(lr){2-3}
\cmidrule(lr){4-5}
\cmidrule(lr){6-7}
\cmidrule(lr){8-9}

 & \multicolumn{1}{c}{TPR} &  \multicolumn{1}{c}{FDP}  &  \multicolumn{1}{c}{TPR}  &  \multicolumn{1}{c}{FDP} &  \multicolumn{1}{c}{TPR}  & \multicolumn{1}{c}{FDP}  &  \multicolumn{1}{c}{TPR}  &  \multicolumn{1}{c}{FDP} \\ 
\midrule
B-RAIL & \bft{1.00 (1.1e-3)} & \bft{0.00 (0.0e-0)} & 1.00 (1.7e-3) & 0.00 (0.0e-0) & 1.00 (1.4e-3) & 0.00 (0.0e-0) & 1.00 (1.4e-3) & 0.00 (0.0e-0) \\ 
  Lasso - $\lambda$ (oracle) & 0.87 (1.2e-3) & 0.12 (1.9e-3) & 0.90 (1.4e-3) & 0.12 (5.3e-3) & 0.81 (2.9e-3) & 0.20 (5.2e-4) & 0.90 (0.0e-0) & 0.03 (4.5e-3) \\ 
  Lasso - $\lambda_k$ (oracle) & 0.92 (1.6e-3) & 0.08 (1.6e-3) & 0.96 (4.8e-3) & 0.00 (0.0e-0) & 0.90 (0.0e-0) & 0.15 (4.5e-3) & 0.90 (0.0e-0) & 0.08 (4.4e-3) \\ 
  Separate Lasso (oracle) & 0.75 (1.6e-3) & 0.24 (4.9e-4) & 0.80 (0.0e-0) & 0.20 (0.0e-0) & 0.70 (1.7e-3) & 0.30 (5.7e-4) & 0.77 (4.8e-3) & 0.21 (1.4e-3) \\ 
  Adaptive Lasso (oracle) & \bft{1.00 (0.0e-0)} & \bft{0.00 (0.0e-0)} & 1.00 (0.0e-0) & 0.00 (0.0e-0) & 1.00 (0.0e-0) & 0.00 (0.0e-0) & 1.00 (0.0e-0) & 0.00 (0.0e-0) \\ 

\addlinespace
\addlinespace
\multicolumn{9}{c}{ \begin{normalsize}  iid Case, p=900  \end{normalsize} } \\
\midrule

  & \multicolumn{2}{c}{Total} &  \multicolumn{2}{c}{Continuous} & \multicolumn{2}{c}{ Binary} & \multicolumn{2}{c}{Counts}  \\ 
 \cmidrule(lr){2-3}
\cmidrule(lr){4-5}
\cmidrule(lr){6-7}
\cmidrule(lr){8-9}

 & \multicolumn{1}{c}{TPR} &  \multicolumn{1}{c}{FDP}  &  \multicolumn{1}{c}{TPR}  &  \multicolumn{1}{c}{FDP} &  \multicolumn{1}{c}{TPR}  & \multicolumn{1}{c}{FDP}  &  \multicolumn{1}{c}{TPR}  &  \multicolumn{1}{c}{FDP} \\ 
\midrule
B-RAIL & \bft{0.94 (7.3e-3)} & \bft{0.12 (1.4e-2)} & 0.98 (5.0e-3) & 0.14 (1.7e-2) & 0.95 (9.9e-3) & 0.08 (1.1e-2) & 0.90 (9.2e-3) & 0.12 (1.5e-2) \\ 
  Lasso - $\lambda$ (oracle) & 0.63 (3.3e-4) & 0.37 (5.0e-4) & 0.80 (0.0e-0) & 0.20 (0.0e-0) & 0.50 (1.0e-3) & 0.50 (1.4e-3) & 0.60 (0.0e-0) & 0.40 (0.0e-0) \\ 
  Lasso - $\lambda_k$ (oracle) & 0.74 (1.7e-3) & 0.26 (1.7e-3) & 0.98 (4.4e-3) & 0.19 (3.3e-3) & 0.63 (7.4e-3) & 0.36 (6.5e-3) & 0.62 (3.9e-3) & 0.20 (7.4e-3) \\ 
  Separate Lasso (oracle) & 0.53 (9.6e-4) & 0.46 (1.3e-3) & 0.60 (0.0e-0) & 0.38 (4.4e-3) & 0.49 (2.9e-3) & 0.51 (1.6e-3) & 0.50 (0.0e-0) & 0.50 (0.0e-0) \\ 
  Adaptive Lasso (oracle) & 0.66 (9.6e-4) & 0.34 (6.9e-4) & 0.79 (3.1e-3) & 0.28 (1.3e-3) & 0.50 (3.5e-3) & 0.44 (1.7e-3) & 0.70 (0.0e-0) & 0.30 (9.0e-4) \\ 
\addlinespace
\addlinespace
\multicolumn{9}{c}{ \begin{normalsize} Non-constant Variance (Heteroscedasticity) \end{normalsize}} \\
\midrule

& \multicolumn{2}{c}{Total} &  \multicolumn{2}{c}{Continuous} & \multicolumn{2}{c}{ Binary} & \multicolumn{2}{c}{Counts}  \\ 
\cmidrule(lr){2-3}
\cmidrule(lr){4-5}
\cmidrule(lr){6-7}
\cmidrule(lr){8-9}
 & \multicolumn{1}{c}{TPR} &  \multicolumn{1}{c}{FDP}  &  \multicolumn{1}{c}{TPR}  &  \multicolumn{1}{c}{FDP} &  \multicolumn{1}{c}{TPR}  & \multicolumn{1}{c}{FDP}  &  \multicolumn{1}{c}{TPR}  &  \multicolumn{1}{c}{FDP} \\ 
\midrule
B-RAIL & \bft{0.97 (3.3e-4)} & \bft{0.01 (1.3e-3)} & 0.90 (1.0e-3) & 0.00 (0.0e-0) & 1.00 (0.0e-0) & 0.00 (0.0e-0) & 1.00 (0.0e-0) & 0.02 (3.7e-3) \\ 
  Lasso - $\lambda$ (oracle) & 0.77 (2.2e-3) & 0.21 (2.2e-3) & 0.88 (3.9e-3) & 0.19 (2.9e-3) & 0.83 (4.7e-3) & 0.06 (5.4e-3) & 0.60 (0.0e-0) & 0.37 (3.6e-3) \\ 
  Lasso - $\lambda_k$ (oracle) & 0.83 (0.0e-0) & 0.17 (0.0e-0) & 0.90 (0.0e-0) & 0.11 (2.8e-3) & 1.00 (0.0e-0) & 0.18 (2.4e-3) & 0.60 (0.0e-0) & 0.22 (4.9e-3) \\ 
  Separate Lasso (oracle) & 0.56 (1.4e-3) & 0.43 (1.1e-3) & 0.58 (4.3e-3) & 0.41 (1.9e-3) & 0.80 (0.0e-0) & 0.19 (2.3e-3) & 0.30 (0.0e-0) & 0.70 (1.4e-3) \\ 
  Adaptive Lasso (oracle) & 0.86 (1.3e-3) & 0.13 (1.2e-3) & 0.90 (1.0e-3) & 0.10 (1.8e-3) & 1.00 (0.0e-0) & 0.11 (3.4e-3) & 0.69 (3.4e-3) & 0.20 (5.6e-3) \\ 
\addlinespace
\addlinespace
\multicolumn{9}{c}{ \begin{normalsize} Block Directed Graph Structure \end{normalsize}} \\
\midrule

& \multicolumn{2}{c}{Total} &  \multicolumn{2}{c}{Continuous} & \multicolumn{2}{c}{ Binary} & \multicolumn{2}{c}{Counts}  \\ 
\cmidrule(lr){2-3}
\cmidrule(lr){4-5}
\cmidrule(lr){6-7}
\cmidrule(lr){8-9}
 & \multicolumn{1}{c}{TPR} &  \multicolumn{1}{c}{FDP}  &  \multicolumn{1}{c}{TPR}  &  \multicolumn{1}{c}{FDP} &  \multicolumn{1}{c}{TPR}  & \multicolumn{1}{c}{FDP}  &  \multicolumn{1}{c}{TPR}  &  \multicolumn{1}{c}{FDP} \\ 
\midrule
B-RAIL & \bft{0.88 (4.8e-3)} & \bft{0.16 (1.1e-2)} & 0.79 (4.8e-3) & 0.12 (1.4e-2) & 0.96 (6.8e-3) & 0.12 (7.0e-3) & 0.89 (5.2e-3) & 0.22 (1.4e-2) \\ 
  Lasso - $\lambda$ (oracle) & 0.72 (1.9e-3) & 0.27 (1.9e-3) & 0.88 (5.8e-3) & 0.17 (3.2e-3) & 1.00 (0.0e-0) & 0.28 (3.5e-3) & 0.30 (0.0e-0) & 0.42 (4.2e-3) \\ 
  Lasso - $\lambda_k$ (oracle) & 0.77 (0.0e-0) & 0.23 (0.0e-0) & 1.00 (0.0e-0) & 0.20 (4.0e-3) & 1.00 (0.0e-0) & 0.24 (5.4e-3) & 0.30 (0.0e-0) & 0.26 (1.5e-2) \\ 
  Separate Lasso (oracle) & 0.51 (2.0e-3) & 0.46 (1.9e-3) & 0.79 (3.1e-3) & 0.18 (4.6e-3) & 0.55 (5.0e-3) & 0.42 (2.8e-3) & 0.20 (0.0e-0) & 0.79 (1.3e-3) \\ 
  Adaptive Lasso (oracle) & 0.70 (0.0e-0) & 0.30 (3.4e-4) & 0.80 (0.0e-0) & 0.20 (0.0e-0) & 1.00 (0.0e-0) & 0.29 (1.3e-3) & 0.30 (0.0e-0) & 0.49 (3.0e-3) \\ 

\addlinespace
\addlinespace
\multicolumn{9}{c}{ \begin{normalsize} OV Data \end{normalsize}} \\
\midrule

& \multicolumn{2}{c}{Total} &  \multicolumn{2}{c}{Continuous} & \multicolumn{2}{c}{Proportion} & \multicolumn{2}{c}{Counts}  \\ 
\cmidrule(lr){2-3}
\cmidrule(lr){4-5}
\cmidrule(lr){6-7}
\cmidrule(lr){8-9}
 & \multicolumn{1}{c}{TPR} &  \multicolumn{1}{c}{FDP}  &  \multicolumn{1}{c}{TPR}  &  \multicolumn{1}{c}{FDP} &  \multicolumn{1}{c}{TPR}  & \multicolumn{1}{c}{FDP}  &  \multicolumn{1}{c}{TPR}  &  \multicolumn{1}{c}{FDP} \\ 
\midrule
B-RAIL &\bft{0.95 (2.9e-3)} & \bft{0.11 (4.7e-3)} & 0.85 (8.7e-3) & 0.23 (4.9e-3) & 1.00 (0.0e-0) & 0.09 (8.2e-3) & 1.00 (1.0e-3) & 0.03 (4.4e-3) \\ 
  Lasso - $\lambda$ (oracle) & 0.57 (8.5e-4) & 0.43 (1.2e-3) & 0.80 (0.0e-0) & 0.38 (0.0e-0) & 0.41 (2.6e-3) & 0.31 (6.2e-3) & 0.50 (0.0e-0) & 0.54 (1.5e-3) \\ 
  Lasso - $\lambda_k$ (oracle) & 0.65 (1.6e-3) & 0.35 (1.6e-3) & 0.80 (1.0e-3) & 0.37 (2.3e-3) & 0.51 (3.6e-3) & 0.00 (1.7e-3) & 0.63 (4.6e-3) & 0.48 (4.6e-3) \\ 
  Separate Lasso (oracle) & 0.40 (1.6e-3) & 0.60 (1.3e-3) & 0.58 (4.1e-3) & 0.41 (1.9e-3) & 0.30 (0.0e-0) & 0.70 (0.0e-0) & 0.30 (2.0e-3) & 0.69 (2.3e-3) \\ 
  Adaptive Lasso (oracle) & 0.93 (1.2e-3) & 0.09 (1.1e-3) & 0.80 (0.0e-0) & 0.12 (2.8e-3) & 0.98 (3.6e-3) & 0.00 (1.0e-3) & 1.00 (0.0e-0) & 0.09 (1.7e-3) \\

\specialrule{1pt}{1pt}{1pt}

\end{tabular}
}}
\vspace{3pt}
\caption{We compare various selection methods under 5 different simulation scenarios. For each scenario, we simulate $\X$ with three blocks (continuous, binary, counts) and a Gaussian response $\y$. We report the true positive rate (TPR) and false discovery proportion (FDP) for overall feature recovery and individual block recoveries, averaged across 200 runs with standard errors in parentheses. We bold the best overall $TPR*(1-FDP)$ values for each simulation scenario. Note that we used oracle information for the Lasso-type methods. 
 }\label{tab:Gaus}
\end{table*}

\begin{table*}[h!]
\makebox[\linewidth]{
\scalebox{0.7}{
\begin{tabular}{l  ll  ll  ll  ll}

\addlinespace
\multicolumn{9}{c}{ \begin{normalsize} Binary Response Block Directed Graph Structure \end{normalsize}} \\
\midrule

& \multicolumn{2}{c}{Total} &  \multicolumn{2}{c}{Continuous} & \multicolumn{2}{c}{ Binary} & \multicolumn{2}{c}{Counts}  \\ 
\cmidrule(lr){2-3}
\cmidrule(lr){4-5}
\cmidrule(lr){6-7}
\cmidrule(lr){8-9}
 & \multicolumn{1}{c}{TPR} &  \multicolumn{1}{c}{FDP}  &  \multicolumn{1}{c}{TPR}  &  \multicolumn{1}{c}{FDP} &  \multicolumn{1}{c}{TPR}  & \multicolumn{1}{c}{FDP}  &  \multicolumn{1}{c}{TPR}  &  \multicolumn{1}{c}{FDP} \\ 
 \midrule
B-RAIL & \bft{0.81 (8.6e-3)} & \bft{0.07 (5.2e-3)} & 0.80 (0.0e-0) & 0.04 (5.4e-3) & 0.97 (5.7e-3) & 0.09 (6.2e-3) & 0.68 (2.1e-2) & 0.07 (1.2e-2) \\ 
  Lasso - $\lambda$ (oracle) & 0.70 (0.0e-0) & 0.29 (1.3e-3) & 0.90 (0.0e-0) & 0.23 (4.1e-3) & 0.90 (0.0e-0) & 0.37 (2.0e-3) & 0.30 (0.0e-0) & 0.18 (1.3e-2) \\ 
  Lasso - $\lambda$ (oracle) & 0.70 (8.5e-4) & 0.30 (8.5e-4) & 0.89 (2.6e-3) & 0.21 (4.3e-3) & 0.90 (0.0e-0) & 0.27 (3.6e-3) & 0.30 (0.0e-0) & 0.52 (7.4e-3) \\ 
  Separate Lasso (oracle) & 0.50 (1.3e-3) & 0.46 (2.3e-3) & 0.80 (0.0e-0) & 0.20 (0.0e-0) & 0.51 (3.9e-3) & 0.41 (7.0e-3) & 0.20 (0.0e-0) & 0.79 (1.7e-3) \\ 
  Adaptive Lasso (oracle) & 0.70 (9.6e-4) & 0.29 (1.4e-3) & 0.81 (2.9e-3) & 0.20 (1.3e-3) & 1.00 (0.0e-0) & 0.36 (2.5e-3) & 0.30 (0.0e-0) & 0.27 (5.7e-3) \\ 
    
\addlinespace
\multicolumn{9}{c}{ \begin{normalsize} Poisson Response Block Directed Graph Structure \end{normalsize}} \\
\midrule

& \multicolumn{2}{c}{Total} &  \multicolumn{2}{c}{Continuous} & \multicolumn{2}{c}{ Binary} & \multicolumn{2}{c}{Counts}  \\ 
\cmidrule(lr){2-3}
\cmidrule(lr){4-5}
\cmidrule(lr){6-7}
\cmidrule(lr){8-9}
 & \multicolumn{1}{c}{TPR} &  \multicolumn{1}{c}{FDP}  &  \multicolumn{1}{c}{TPR}  &  \multicolumn{1}{c}{FDP} &  \multicolumn{1}{c}{TPR}  & \multicolumn{1}{c}{FDP}  &  \multicolumn{1}{c}{TPR}  &  \multicolumn{1}{c}{FDP} \\ 
\midrule
B-RAIL & \bft{0.81 (8.6e-3)} & \bft{0.07 (5.2e-3)} & 0.80 (0.0e-0) & 0.04 (5.4e-3) & 0.97 (5.7e-3) & 0.09 (6.2e-3) & 0.68 (2.1e-2) & 0.07 (1.2e-2) \\ 
  Lasso - $\lambda$ (oracle) & 0.70 (0.0e-0) & 0.29 (1.3e-3) & 0.90 (0.0e-0) & 0.23 (4.1e-3) & 0.90 (0.0e-0) & 0.37 (2.0e-3) & 0.30 (0.0e-0) & 0.18 (1.3e-2) \\ 
  Lasso - $\lambda$ (oracle) & 0.70 (8.5e-4) & 0.30 (8.5e-4) & 0.89 (2.6e-3) & 0.21 (4.3e-3) & 0.90 (0.0e-0) & 0.27 (3.6e-3) & 0.30 (0.0e-0) & 0.52 (7.4e-3) \\ 
  Separate Lasso (oracle) & 0.50 (1.3e-3) & 0.46 (2.3e-3) & 0.80 (0.0e-0) & 0.20 (0.0e-0) & 0.51 (3.9e-3) & 0.41 (7.0e-3) & 0.20 (0.0e-0) & 0.79 (1.7e-3) \\ 
  Adaptive Lasso (oracle) & 0.70 (9.6e-4) & 0.29 (1.4e-3) & 0.81 (2.9e-3) & 0.20 (1.3e-3) & 1.00 (0.0e-0) & 0.36 (2.5e-3) & 0.30 (0.0e-0) & 0.27 (5.7e-3) \\ 

\specialrule{1pt}{1pt}{1pt}
\end{tabular}
 }}
 \vspace{3pt}
\caption{We compare various selection methods under the block directed graph simulation design with binary responses and with Poisson responses. We report the TPR and FDP for feature recovery, averaged across 200 runs with standard errors in parentheses. We bold the best overall $TPR*(1-FDP)$ values for each simulation scenario. 
} \label{tab:BinPois}
\end{table*}

Under each of these simulation scenarios, we evaluate the performance of B-RAIL and the oracle Lasso-type methods by reporting the true positive rate (TPR) and false discovery proportion (FDP) for overall feature recovery and individual block recoveries. Due to the large number of features, we use FDP, defined as the number of false positives divided by total the number of recovered non-zero features, instead of the false discovery rate. 

We summarize the results of our simulations with Gaussian responses in Table~\ref{tab:Gaus} and those with binary and Poisson responses in Table~\ref{tab:BinPois}. Note that for the binary and Poisson responses, we show the block directed graph results here and provide the other simulation results in the Appendix. We also highlight in bold the TPR/FDP combination with the highest TPR*(1-FDP) value for overall recovery. In almost all scenarios, the results in Table~\ref{tab:Gaus} and Table~\ref{tab:BinPois} indicate that B-RAIL (with no oracle information) is able to achieve a higher TPR and lower FDP than its competitive Lasso-type methods with oracle information.

When oracle information is unavailable, model selection techniques can introduce additional errors and further complicate feature selection. Table~\ref{tab:Methods} shows one such case and compares the block directed graph simulation performance of B-RAIL against the Lasso-type methods using 5-fold cross-validation, extended BIC, and stability selection to select the penalty parameters. We also include the oracle estimators for the same set of simulations to emphasize the large decrease in performance when the Lasso-type methods do not have oracle information. These simulations indicate that cross-validation tends to over-select the number of features in the model while extended BIC under-selects, and stability selection performs the best but pales in comparison to oracle selection. In contrast, B-RAIL when initialized to an over-selection using the pre-specified sparsity level of $0.2p_k$ outperforms the Lasso-type methods even when oracle selection is used for these competitive methods.  Additional simulations, confirming the strong empirical performance of B-RAIL, are provided in the Appendix.

\begin{table*}[h!]
\makebox[\linewidth]{
\scalebox{0.7}{
\begin{tabular}{l  ll  ll  ll  ll}
  
  \multicolumn{9}{c}{ \begin{normalsize} Block Directed Graph Structure \end{normalsize}} \\
\midrule

& \multicolumn{2}{c}{Total} &  \multicolumn{2}{c}{Continuous} & \multicolumn{2}{c}{ Binary} & \multicolumn{2}{c}{Counts}  \\ 
\cmidrule(lr){2-3}
\cmidrule(lr){4-5}
\cmidrule(lr){6-7}
\cmidrule(lr){8-9}
 & \multicolumn{1}{c}{TPR} &  \multicolumn{1}{c}{FDP}  &  \multicolumn{1}{c}{TPR}  &  \multicolumn{1}{c}{FDP} &  \multicolumn{1}{c}{TPR}  & \multicolumn{1}{c}{FDP}  &  \multicolumn{1}{c}{TPR}  &  \multicolumn{1}{c}{FDP} \\ 
\midrule
B-RAIL & \bft{0.86 (1.2e-2)} & \bft{0.20 (2.8e-2)} & 0.78 (9.2e-3) & 0.15 (3.6e-2) & 0.93 (1.9e-2) & 0.16 (1.8e-2) & 0.88 (1.2e-2) & 0.29 (4.3e-2) \\ 
  Lasso - $\lambda$ (oracle) & 0.72 (3.5e-3) & 0.27 (4.8e-3) & 0.87 (1.1e-2) & 0.20 (5.0e-3) & 1.00 (0.0e-0) & 0.40 (1.6e-2) & 0.30 (0.0e-0) & 0.21 (5.0e-3) \\ 
  Lasso - $\lambda_k$ (oracle) & 0.77 (0.0e-0) & 0.23 (0.0e-0) & 1.00 (0.0e-0) & 0.24 (1.4e-2) & 1.00 (0.0e-0) & 0.32 (2.7e-2) & 0.30 (0.0e-0) & 0.14 (2.2e-2) \\ 
  Separate Lasso (oracle) & 0.51 (5.0e-3) & 0.44 (5.6e-3) & 0.79 (8.2e-3) & 0.18 (1.2e-2) & 0.55 (1.1e-2) & 0.41 (5.0e-3) & 0.20 (0.0e-0) & 0.73 (1.1e-2) \\ 
  Adaptive Lasso (oracle) & 0.70 (0.0e-0) & 0.30 (2.3e-3) & 0.80 (0.0e-0) & 0.20 (0.0e-0) & 1.00 (0.0e-0) & 0.42 (9.2e-3) & 0.30 (0.0e-0) & 0.27 (1.1e-2) \\ 

\addlinespace
\addlinespace
\multicolumn{9}{c}{ \begin{normalsize} 5 Fold Cross Validation  \end{normalsize} } \\
\midrule

  & \multicolumn{2}{c}{Total} &  \multicolumn{2}{c}{Continuous} & \multicolumn{2}{c}{ Binary} & \multicolumn{2}{c}{Counts}  \\ 
 \cmidrule(lr){2-3}
\cmidrule(lr){4-5}
\cmidrule(lr){6-7}
\cmidrule(lr){8-9}

 & \multicolumn{1}{c}{TPR} &  \multicolumn{1}{c}{FDP}  &  \multicolumn{1}{c}{TPR}  &  \multicolumn{1}{c}{FDP} &  \multicolumn{1}{c}{TPR}  & \multicolumn{1}{c}{FDP}  &  \multicolumn{1}{c}{TPR}  &  \multicolumn{1}{c}{FDP} \\ 
\midrule
Lasso - $\lambda$ & 0.98 (3.0e-3) & 0.62 (4.1e-3) & 1.00 (0.0e-0) & 0.58 (6.3e-3) & 1.00 (0.0e-0) & 0.63 (4.0e-3) & 0.95 (8.9e-3) & 0.63 (3.9e-3) \\ 
  Lasso - $\lambda_k$ & 0.60 (2.3e-3) & 0.17 (2.2e-3) & 0.80 (0.0e-0) & 0.13 (3.6e-3) & 0.69 (6.8e-3) & 0.07 (6.3e-3) & 0.30 (0.0e-0) & 0.40 (1.0e-3) \\ 
  Separate Lasso & 0.37 (8.1e-3) & 0.28 (1.5e-2) & 0.58 (9.9e-3) & 0.05 (9.9e-3) & 0.54 (1.8e-2) & 0.38 (2.0e-2) & 0.00 (0.0e-0) & 0.47 (5.0e-2) \\ 
  Adaptive Lasso & 0.73 (3.0e-3) & 0.37 (7.8e-3) & 0.87 (6.8e-3) & 0.23 (6.4e-3) & 1.00 (2.0e-3) & 0.39 (8.1e-3) & 0.31 (2.9e-3) & 0.55 (9.5e-3) \\ 
\addlinespace
\addlinespace
\multicolumn{9}{c}{ \begin{normalsize} Extended BIC \end{normalsize}} \\
\midrule

& \multicolumn{2}{c}{Total} &  \multicolumn{2}{c}{Continuous} & \multicolumn{2}{c}{ Binary} & \multicolumn{2}{c}{Counts}  \\ 
\cmidrule(lr){2-3}
\cmidrule(lr){4-5}
\cmidrule(lr){6-7}
\cmidrule(lr){8-9}
 & \multicolumn{1}{c}{TPR} &  \multicolumn{1}{c}{FDP}  &  \multicolumn{1}{c}{TPR}  &  \multicolumn{1}{c}{FDP} &  \multicolumn{1}{c}{TPR}  & \multicolumn{1}{c}{FDP}  &  \multicolumn{1}{c}{TPR}  &  \multicolumn{1}{c}{FDP} \\ 
\midrule
Lasso - $\lambda$ & 0.03 (0.0e-0) & 0.00 (0.0e-0) & 0.10 (0.0e-0) & 0.00 (0.0e-0) & 0.00 (0.0e-0) & 0.00 (0.0e-0) & 0.00 (0.0e-0) & 0.00 (0.0e-0) \\ 
  Lasso - $\lambda_k$ & 0.58 (2.0e-3) & 0.18 (3.1e-3) & 0.80 (0.0e-0) & 0.15 (5.3e-3) & 0.64 (5.9e-3) & 0.04 (6.1e-3) & 0.30 (0.0e-0) & 0.42 (4.0e-3) \\ 
  Separate Lasso & 0.18 (1.7e-3) & 0.07 (7.9e-3) & 0.50 (1.0e-3) & 0.00 (0.0e-0) & 0.04 (4.9e-3) & 0.00 (0.0e-0) & 0.00 (0.0e-0) & 0.47 (5.0e-2) \\ 
  Adaptive Lasso & 0.30 (0.0e-0) & 0.00 (0.0e-0) & 0.50 (0.0e-0) & 0.00 (0.0e-0) & 0.40 (0.0e-0) & 0.00 (0.0e-0) & 0.00 (0.0e-0) & 0.00 (0.0e-0) \\ 
\addlinespace
\addlinespace
\multicolumn{9}{c}{ \begin{normalsize} Stability Selection \end{normalsize}} \\
\midrule

& \multicolumn{2}{c}{Total} &  \multicolumn{2}{c}{Continuous} & \multicolumn{2}{c}{ Binary} & \multicolumn{2}{c}{Counts}  \\ 
\cmidrule(lr){2-3}
\cmidrule(lr){4-5}
\cmidrule(lr){6-7}
\cmidrule(lr){8-9}
 & \multicolumn{1}{c}{TPR} &  \multicolumn{1}{c}{FDP}  &  \multicolumn{1}{c}{TPR}  &  \multicolumn{1}{c}{FDP} &  \multicolumn{1}{c}{TPR}  & \multicolumn{1}{c}{FDP}  &  \multicolumn{1}{c}{TPR}  &  \multicolumn{1}{c}{FDP} \\ 
\midrule
Lasso -$\lambda$ & 0.67 (2.4e-3) & 0.01 (1.8e-3) & 0.80 (0.0e-0) & 0.02 (4.2e-3) & 0.92 (6.9e-3) & 0.00 (9.1e-4) & 0.29 (2.4e-3) & 0.00 (0.0e-0) \\ 
  Lasso - $\lambda_k$  & 0.58 (2.7e-3) & 0.04 (2.3e-3) & 0.80 (0.0e-0) & 0.09 (4.7e-3) & 0.66 (6.5e-3) & 0.00 (0.0e-0) & 0.28 (4.3e-3) & 0.00 (2.5e-3) \\ 
  Separate Lasso & 0.42 (3.9e-3) & 0.28 (6.6e-3) & 0.69 (2.7e-3) & 0.13 (7.2e-3) & 0.55 (1.1e-2) & 0.39 (8.8e-3) & 0.00 (0.0e-0) & 0.29 (4.6e-2) \\ 
  Adaptive Lasso & 0.65 (1.7e-3) & 0.08 (2.6e-3) & 0.80 (0.0e-0) & 0.11 (0.0e-0) & 0.84 (5.1e-3) & 0.07 (6.0e-3) & 0.30 (0.0e-0) & 0.01 (4.3e-3) \\ 

\specialrule{1pt}{1pt}{1pt}

\end{tabular}
}}
\vspace{3pt}
\caption{We compare feature recovery for B-RAIL and Lasso-type methods with various model selection methods. Here, we simulate from the block directed graph simulation design with Gaussian responses and report the TPR and FDP, averaged across 200 runs with standard errors in parentheses. We highlight the best overall $TPR*(1-FDP)$ values in bold. Note for stability selection, we initialize $\lambda$ using the $\lambda$ selected by CV.
 } \label{tab:Methods}
\end{table*}

\section{Case Study: Integrative Genomics of Ovarian Cancer}
\label{sec:OVCase}



One promising practical application for our research on multi-view feature selection lies in integrative cancer genomics. Here, scientists seek to integrate data from multiple sources of high-throughput genomic data to more holistically model the genomic systems in cancer cells, leading to a better understanding of disease mechanisms and possible therapies.    

In this case study, we seek to integrate three different types of genomic data to study how epigenetics and short RNAs influence the gene regulatory system in ovarian cancer.  Specifically, we are interested in discovering miRNAs and CpG sites, which affect the gene expression of well-known oncogenes in ovarian cancer and hence can serve as potential drug targets for blocking or decreasing the expression of these oncogenes. Driven by this goal of discovering potential drug targets, we use our proposed B-RAIL method to estimate the integrative ovarian cancer gene regulatory network with the specific intention of identifying miRNAs and CpG sites that are directly linked to known oncogenes of ovarian cancer.

In this investigation, we integrate the following three data sets: (1) count-valued gene expression measured via RNASeq, (2) continuous (Gaussian) miRNA expression, and (3) proportion-valued DNA methylation data from The Cancer Genome Atlas (TCGA) ovarian cancer study, which is publicly available \citep{cancer2011integrated}. The TCGA data originally contained $19,990$ genes, $27,578$ CpG sites, and $799$ miRNAs but only $n = 293$ common patients across all three data sets of interest. We hence reduced the number of features to manageable sizes by first filtering features according to their association with several important clinical outcomes - \emph{survival} via a univariate cox model, \emph{chemo-resistance} via a univariate logistic model, and \emph{recurrence} via a univariate logistic model. In addition, we transformed the RNASeq data using the Kolmogorov-Smirnov Test ($\alpha = 0.262$) to alleviate the problem of very large counts (up to 20,000). This preprocessing yielded $p_1 = 408$ genes, $p_2 = 301$ CpG sites, and $p_3 = 307$ miRNAs in the RNASeq, methylation, and miRNA data sets, respectively. Lastly, per the recommendation of scientists, we included 20 additional highly mutated genes that were experimentally identified as important in ovarian cancer, resulting in $p_1 = 428$ genes in the RNASeq data set. 

To estimate the integrated ovarian cancer network, we fit a Block Directed Markov Random Field (BDMRF) model \citep{yang2014general} using B-RAIL to estimate the neighborhood of each node in the graph. Note that since miRNAs and methylation are both gene regulatory mechanisms, miRNAs and methylation can affect expression levels (measured via RNASeq), but the converse is not possible. To agree with this known physical mechanism, we set the partial ordering of the mixed graph underlying BDMRF as $P[X_1, X_2, X_3] = P[X_1 |X_2,X_3]P[X_2]P[X_3]$, where $P[X_2]$ is a pairwise Ising MRF for the proportion-valued methylation data, $P[X_3]$ is a pairwise Gaussian MRF for the continuous miRNA data, and $P[X_1 |X_2,X_3]$ is a pairwise Poisson CRF for the count-valued RNASeq data. However, we recall that only negative conditional dependencies are permitted in the Poisson MRF and CRF models. Since this constraint is unrealistic for genomics data, we fit a Sub-Linear Poisson CRF, in lieu of the usual Poisson CRF, to allow for both positive and negative conditional dependencies \cite{yang2013poisson}. Under this specified BDMRF model, we employ node-wise neighborhood selection \citep{meinshausen2006high, yang2015graphical} using B-RAIL to learn the edge structure of the integrated network.


%
%

\begin{figure}[h!]
\centering
\begin{minipage}{1\textwidth}
  \centering
  \includegraphics[width=.65\linewidth]{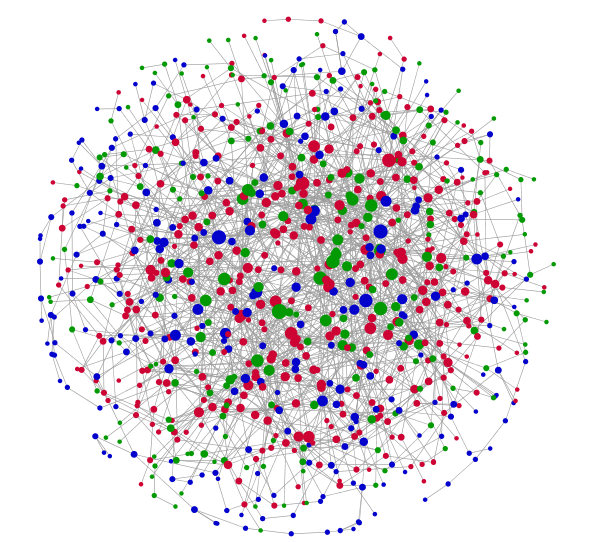}
  \captionof{figure}{We present the integrated ovarian cancer genetic network estimated by the B-RAIL algorithm. The blue nodes denote miRNAs, green nodes denote CpG sites, red nodes denote gene expression via RNASeq, and the size of each node is proportional to the number of connected first neighbors.}
  \label{fig:OVnetwork}
\end{minipage}%
\end{figure}

\begin{figure}[h!]
\centering
\begin{tabular}[c]{cc}
\begin{subfigure}[b]{0.5\linewidth}
\includegraphics[width=\textwidth]{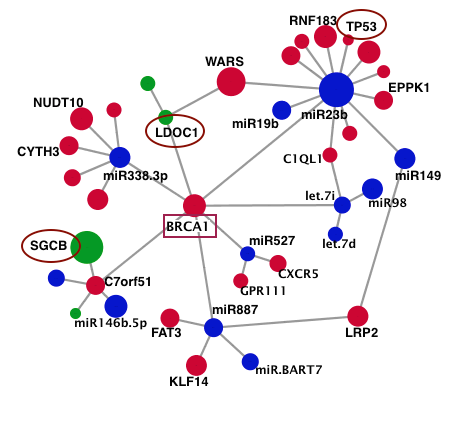}
  \captionof{figure}{Sub-network for BRCA1 gene}
\end{subfigure}&
\hspace*{\fill} 
\begin{subfigure}[b]{0.5\linewidth}
\includegraphics[width=\textwidth]{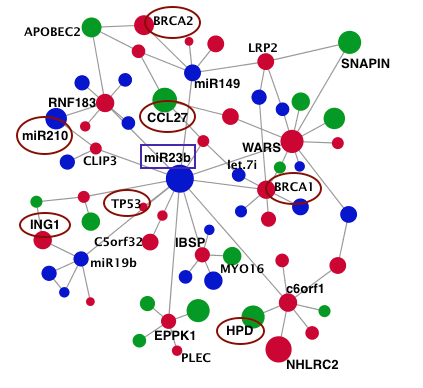}
  \captionof{figure}{ Sub-network for miR23b}
\end{subfigure}\\
  \end{tabular}    
  \caption{We zoom in on the sub-networks for two known biomarkers, which have been previously implicated in ovarian cancer studies. Key mutated cancer biomarkers such as miR23b and BRCA1 are found to have many inter-connections to biomarkers, circled in red, that are consistent with the cancer literature.  \citep{buchholtz2014epigenetic, obermayr2010assessment, giannakakis2008mir, freier2016role, gao2009nk, toyama1999suppression, tong2017down}}
  \label{fig:Zoom}
\end{figure}

Our overall BRAIL-estimated network is presented in Figure~\ref{fig:OVnetwork}, and in Figure~\ref{fig:Zoom}, we more closely examine the relationships between the oncogenes, miRNAs, and CpG sites by zooming in on the sub-networks for the well-known oncogene, BRCA1, and its direct neighbor, miRNA23b. Both BRCA1 and miRNA23b are well-known biomarkers and have been implicated in several ovarian cancer studies \citep{antoniou2003average, king2003breast, brca1994strong, li2014microrna, yan2016mir, geng2012methylation}. Moreover, miRNA23b is known to play a key role in p53 signaling (via TP53) \citep{boren2009micrornas}, agreeing with the estimated edge between the TP53 oncogene and miRNA23b in Figure~\ref{fig:Zoom}(b).

Aside from this link however, the estimated edges between genes, CpG sites, and miRNAs in Figure~\ref{fig:OVnetwork} are largely unexplored and unknown by researchers since B-RAIL is one of the first practical approaches for multi-view feature selection. Nevertheless, we can partially validate our B-RAIL-estimated network by highlighting the many genes with verified connections in the ovarian cancer and cancer proliferation/suppression literatures. In Figure~\ref{fig:Zoom}, we circle this collection of implicated genes, which includes LDOC1, SGCB, and miRNA210 \citep{buchholtz2014epigenetic, obermayr2010assessment, giannakakis2008mir}.

As we have noted, there is substantial evidence in the scientific literature, suggesting that our proposed B-RAIL algorithm successfully identified promising candidates as well as known biomarkers involved in ovarian cancer. By taking into account the relationships between genes, miRNAs, and CpG sites, our integrative analysis via B-RAIL leads to valuable insights beyond a single biomarker type and to novel discoveries of direct connections between miRNAs, CpG sites, and known oncogenes, which may aid the development of targeted drug therapies for ovarian cancer. This is the first integrative analysis of its kind, and future experiments studying the connections between known ovarian cancer oncogenes and candidate miRNAs and CpG sites would be of great value to validate our findings.

\section{Discussion}
\label{sec:Disc}

Though we have primarily focused on applications to integrative genomics in this work, B-RAIL is not limited to this context. B-RAIL can be applied to any field that yields high-dimensional multi-view data, and with the rapid advances in technologies, we expect B-RAIL to have a growing and far-reaching impact in fields such as imaging genetics, national security, climate studies, spatial statistics, Internet data, marketing, and economics. B-RAIL is also a versatile tool that can be used to for any sparse regression or graph selection problem in this multi-view context.

In addition to developing an effective data integration tool for multi-view feature selection, our work addresses the many difficulties of performing multi-view feature selection in practice. These practical challenges were severely under-studied prior to this work, but we partially resolve this gap, identifying four root challenges which interact with one another to impede recovery. Throughout our investigation of these practical challenges, we provide strong empirical evidence of the existence as well as the adverse consequences of such challenges. However, the theoretical underpinnings of these issues are still unknown. Understanding exactly how challenges such as shrinkage noise and the beta-min condition are influenced by varying domains and signals would be of great benefit to the field of data integration as a whole. We also highlight that while the Lasso has been well-studied under Gaussianity and idealized assumptions, the increasing abundance of correlated non-Gaussian data in multi-view settings requires a greater push for theoretical studies on feature selection with heterogeneous data and the GLM Lasso.


Overall, we have demonstrated many challenges posed by multi-view feature selection, and in our investigation of these challenges, we opened up new avenues for future theoretical work to rigorously understand how the heterogeneity of multi-view data complicates feature selection. Driven by these challenges and the ineffectiveness of existing methods, we developed a practical solution to overcome the current challenges. Our method, B-RAIL, is one of the first practical tools for multi-view feature selection and is grounded in deep theoretical foundations. With its versatility and strong empirical performance, B-RAIL facilitates impactful integrative analyses across a broad spectrum of fields.

\section*{Acknowledgements}

Y.B. acknowledges support from the NIH/NCI T32 CA096520 training program in Biostatistics for Cancer Research, grant number 5T32CA09652011. T.T. acknowledges support from the NSF Graduate Research Fellowship Program DGE-1752814 and ARO, grant number W911NF1710005. GA acknowledges support from NSF DMS-1554821, NSF NeuroNex-1707400, and NSF DMS-1264058. We also thank Zhandong Liu, Ying-Wooi Wan, and Matthew L. Anderson at Baylor College of Medicine for thoughtful discussions related to this work.

\bibliographystyle{imsart-nameyear}
\bibliography{paper3bib}

\begin{thebibliography}{53}

\bibitem[\protect\citeauthoryear{Acar, Kolda and Dunlavy}{2011}]{acar2011all}
\begin{barticle}[author]
\bauthor{\bsnm{Acar},~\bfnm{Evrim}\binits{E.}},
  \bauthor{\bsnm{Kolda},~\bfnm{Tamara~G}\binits{T.~G.}} \AND
  \bauthor{\bsnm{Dunlavy},~\bfnm{Daniel~M}\binits{D.~M.}}
(\byear{2011}).
\btitle{All-at-once optimization for coupled matrix and tensor factorizations}.
\bjournal{arXiv preprint arXiv:1105.3422}.
\end{barticle}
\endbibitem

\bibitem[\protect\citeauthoryear{Allen}{1974}]{allen1974cv}
\begin{barticle}[author]
\bauthor{\bsnm{Allen},~\bfnm{David~M}\binits{D.~M.}}
(\byear{1974}).
\btitle{The relationship between variable selection and data agumentation and a
  method for prediction}.
\bjournal{Technometrics}
\bvolume{16}
\bpages{125--127}.
\end{barticle}
\endbibitem

\bibitem[\protect\citeauthoryear{Antoniou et~al.}{2003}]{antoniou2003average}
\begin{barticle}[author]
\bauthor{\bsnm{Antoniou},~\bfnm{Anthony}\binits{A.}},
  \bauthor{\bsnm{Pharoah},~\bfnm{PDP}\binits{P.}},
  \bauthor{\bsnm{Narod},~\bfnm{Steven}\binits{S.}},
  \bauthor{\bsnm{Risch},~\bfnm{Harvey~A}\binits{H.~A.}},
  \bauthor{\bsnm{Eyfjord},~\bfnm{Jorunn~E}\binits{J.~E.}},
  \bauthor{\bsnm{Hopper},~\bfnm{JL}\binits{J.}},
  \bauthor{\bsnm{Loman},~\bfnm{Niklas}\binits{N.}},
  \bauthor{\bsnm{Olsson},~\bfnm{H{\aa}kan}\binits{H.}},
  \bauthor{\bsnm{Johannsson},~\bfnm{O}\binits{O.}},
  \bauthor{\bsnm{Borg},~\bfnm{{\AA}ke}\binits{{\AA}.}} \betal{et~al.}
(\byear{2003}).
\btitle{Average risks of breast and ovarian cancer associated with BRCA1 or
  BRCA2 mutations detected in case series unselected for family history: a
  combined analysis of 22 studies}.
\bjournal{The American Journal of Human Genetics}
\bvolume{72}
\bpages{1117--1130}.
\end{barticle}
\endbibitem

\bibitem[\protect\citeauthoryear{Boren et~al.}{2009}]{boren2009micrornas}
\begin{barticle}[author]
\bauthor{\bsnm{Boren},~\bfnm{Todd}\binits{T.}},
  \bauthor{\bsnm{Xiong},~\bfnm{Yin}\binits{Y.}},
  \bauthor{\bsnm{Hakam},~\bfnm{Ardeshir}\binits{A.}},
  \bauthor{\bsnm{Wenham},~\bfnm{Robert}\binits{R.}},
  \bauthor{\bsnm{Apte},~\bfnm{Sachin}\binits{S.}},
  \bauthor{\bsnm{Chan},~\bfnm{Gina}\binits{G.}},
  \bauthor{\bsnm{Kamath},~\bfnm{Siddharth~G}\binits{S.~G.}},
  \bauthor{\bsnm{Chen},~\bfnm{Dung-Tsa}\binits{D.-T.}},
  \bauthor{\bsnm{Dressman},~\bfnm{Holly}\binits{H.}} \AND
  \bauthor{\bsnm{Lancaster},~\bfnm{Johnathan~M}\binits{J.~M.}}
(\byear{2009}).
\btitle{MicroRNAs and their target messenger RNAs associated with ovarian
  cancer response to chemotherapy}.
\bjournal{Gynecologic oncology}
\bvolume{113}
\bpages{249--255}.
\end{barticle}
\endbibitem

\bibitem[\protect\citeauthoryear{BRCA}{1994}]{brca1994strong}
\begin{barticle}[author]
\bauthor{\bsnm{BRCA},~\bfnm{Susceptibility~Gene}\binits{S.~G.}}
(\byear{1994}).
\btitle{A strong candidate for the breast and ovarian cancer susceptibility
  gene BRCA1}.
\bjournal{Science}
\bvolume{266}
\bpages{7}.
\end{barticle}
\endbibitem

\bibitem[\protect\citeauthoryear{Buchholtz
  et~al.}{2014}]{buchholtz2014epigenetic}
\begin{barticle}[author]
\bauthor{\bsnm{Buchholtz},~\bfnm{Marie-Luise}\binits{M.-L.}},
  \bauthor{\bsnm{Br{\"u}ning},~\bfnm{Ansgar}\binits{A.}},
  \bauthor{\bsnm{Mylonas},~\bfnm{Ioannis}\binits{I.}} \AND
  \bauthor{\bsnm{J{\"u}ckstock},~\bfnm{Julia}\binits{J.}}
(\byear{2014}).
\btitle{Epigenetic silencing of the LDOC1 tumor suppressor gene in ovarian
  cancer cells}.
\bjournal{Archives of gynecology and obstetrics}
\bvolume{290}
\bpages{149--154}.
\end{barticle}
\endbibitem

\bibitem[\protect\citeauthoryear{B{\"u}hlmann
  et~al.}{2013}]{buhlmann2013statistical}
\begin{barticle}[author]
\bauthor{\bsnm{B{\"u}hlmann},~\bfnm{Peter}\binits{P.}} \betal{et~al.}
(\byear{2013}).
\btitle{Statistical significance in high-dimensional linear models}.
\bjournal{Bernoulli}
\bvolume{19}
\bpages{1212--1242}.
\end{barticle}
\endbibitem

\bibitem[\protect\citeauthoryear{Bunea et~al.}{2007}]{bunea2007sparsity}
\begin{barticle}[author]
\bauthor{\bsnm{Bunea},~\bfnm{Florentina}\binits{F.}},
  \bauthor{\bsnm{Tsybakov},~\bfnm{Alexandre}\binits{A.}},
  \bauthor{\bsnm{Wegkamp},~\bfnm{Marten}\binits{M.}} \betal{et~al.}
(\byear{2007}).
\btitle{Sparsity oracle inequalities for the Lasso}.
\bjournal{Electronic Journal of Statistics}
\bvolume{1}
\bpages{169--194}.
\end{barticle}
\endbibitem

\bibitem[\protect\citeauthoryear{Chen and Chen}{2012}]{chen2012extended}
\begin{barticle}[author]
\bauthor{\bsnm{Chen},~\bfnm{Jiahua}\binits{J.}} \AND
  \bauthor{\bsnm{Chen},~\bfnm{Zehua}\binits{Z.}}
(\byear{2012}).
\btitle{Extended BIC for small-n-large-P sparse GLM}.
\bjournal{Statistica Sinica}
\bpages{555--574}.
\end{barticle}
\endbibitem

\bibitem[\protect\citeauthoryear{Chen, Witten and
  Shojaie}{2014}]{chen2014selection}
\begin{barticle}[author]
\bauthor{\bsnm{Chen},~\bfnm{Shizhe}\binits{S.}},
  \bauthor{\bsnm{Witten},~\bfnm{Daniela~M}\binits{D.~M.}} \AND
  \bauthor{\bsnm{Shojaie},~\bfnm{Ali}\binits{A.}}
(\byear{2014}).
\btitle{Selection and estimation for mixed graphical models}.
\bjournal{Biometrika}
\bvolume{102}
\bpages{47--64}.
\end{barticle}
\endbibitem

\bibitem[\protect\citeauthoryear{Cheng et~al.}{2013}]{cheng2013high}
\begin{barticle}[author]
\bauthor{\bsnm{Cheng},~\bfnm{Jie}\binits{J.}},
  \bauthor{\bsnm{Li},~\bfnm{Tianxi}\binits{T.}},
  \bauthor{\bsnm{Levina},~\bfnm{Elizaveta}\binits{E.}} \AND
  \bauthor{\bsnm{Zhu},~\bfnm{Ji}\binits{J.}}
(\byear{2013}).
\btitle{High-dimensional mixed graphical models}.
\bjournal{arXiv preprint arXiv:1304.2810}.
\end{barticle}
\endbibitem

\bibitem[\protect\citeauthoryear{Fan and Li}{2001}]{fan2001variable}
\begin{barticle}[author]
\bauthor{\bsnm{Fan},~\bfnm{Jianqing}\binits{J.}} \AND
  \bauthor{\bsnm{Li},~\bfnm{Runze}\binits{R.}}
(\byear{2001}).
\btitle{Variable selection via nonconcave penalized likelihood and its oracle
  properties}.
\bjournal{Journal of the American statistical Association}
\bvolume{96}
\bpages{1348--1360}.
\end{barticle}
\endbibitem

\bibitem[\protect\citeauthoryear{Freier}{2016}]{freier2016role}
\begin{bphdthesis}[author]
\bauthor{\bsnm{Freier},~\bfnm{Christoph}\binits{C.}}
(\byear{2016}).
\btitle{Role of regulatory T cells and associated chemokines in human
  gynecological tumors},
\btype{PhD thesis},
\bpublisher{lmu}.
\end{bphdthesis}
\endbibitem

\bibitem[\protect\citeauthoryear{Gao et~al.}{2009}]{gao2009nk}
\begin{barticle}[author]
\bauthor{\bsnm{Gao},~\bfnm{Jian-Qing}\binits{J.-Q.}},
  \bauthor{\bsnm{Tsuda},~\bfnm{Yasuhiro}\binits{Y.}},
  \bauthor{\bsnm{Han},~\bfnm{Min}\binits{M.}},
  \bauthor{\bsnm{Xu},~\bfnm{Dong-Hang}\binits{D.-H.}},
  \bauthor{\bsnm{Kanagawa},~\bfnm{Naoko}\binits{N.}},
  \bauthor{\bsnm{Hatanaka},~\bfnm{Yutaka}\binits{Y.}},
  \bauthor{\bsnm{Tani},~\bfnm{Yoichi}\binits{Y.}},
  \bauthor{\bsnm{Mizuguchi},~\bfnm{Hiroyuki}\binits{H.}},
  \bauthor{\bsnm{Tsutsumi},~\bfnm{Yasuo}\binits{Y.}},
  \bauthor{\bsnm{Mayumi},~\bfnm{Tadanori}\binits{T.}} \betal{et~al.}
(\byear{2009}).
\btitle{NK cells are migrated and indispensable in the anti-tumor activity
  induced by CCL27 gene therapy}.
\bjournal{Cancer immunology, immunotherapy}
\bvolume{58}
\bpages{291}.
\end{barticle}
\endbibitem

\bibitem[\protect\citeauthoryear{Geng et~al.}{2012}]{geng2012methylation}
\begin{barticle}[author]
\bauthor{\bsnm{Geng},~\bfnm{Jiong}\binits{J.}},
  \bauthor{\bsnm{Luo},~\bfnm{Hui}\binits{H.}},
  \bauthor{\bsnm{Pu},~\bfnm{Yi}\binits{Y.}},
  \bauthor{\bsnm{Zhou},~\bfnm{Zhimin}\binits{Z.}},
  \bauthor{\bsnm{Wu},~\bfnm{Xiaoming}\binits{X.}},
  \bauthor{\bsnm{Xu},~\bfnm{Wenhui}\binits{W.}} \AND
  \bauthor{\bsnm{Yang},~\bfnm{Zhengxiang}\binits{Z.}}
(\byear{2012}).
\btitle{Methylation mediated silencing of miR-23b expression and its role in
  glioma stem cells}.
\bjournal{Neuroscience letters}
\bvolume{528}
\bpages{185--189}.
\end{barticle}
\endbibitem

\bibitem[\protect\citeauthoryear{Giannakakis et~al.}{2008}]{giannakakis2008mir}
\begin{barticle}[author]
\bauthor{\bsnm{Giannakakis},~\bfnm{Antonis}\binits{A.}},
  \bauthor{\bsnm{Sandaltzopoulos},~\bfnm{Raphael}\binits{R.}},
  \bauthor{\bsnm{Greshock},~\bfnm{Joel}\binits{J.}},
  \bauthor{\bsnm{Liang},~\bfnm{Shun}\binits{S.}},
  \bauthor{\bsnm{Huang},~\bfnm{Jia}\binits{J.}},
  \bauthor{\bsnm{Hasegawa},~\bfnm{Kosei}\binits{K.}},
  \bauthor{\bsnm{Li},~\bfnm{Chunsheng}\binits{C.}},
  \bauthor{\bsnm{O'Brien-Jenkins},~\bfnm{Ann}\binits{A.}},
  \bauthor{\bsnm{Katsaros},~\bfnm{Dionyssios}\binits{D.}},
  \bauthor{\bsnm{Weber},~\bfnm{Barbara~L}\binits{B.~L.}} \betal{et~al.}
(\byear{2008}).
\btitle{miR-210 links hypoxia with cell cycle regulation and is deleted in
  human epithelial ovarian cancer}.
\bjournal{Cancer biology \& therapy}
\bvolume{7}
\bpages{255--264}.
\end{barticle}
\endbibitem

\bibitem[\protect\citeauthoryear{Hall and Llinas}{1997}]{hall1997introduction}
\begin{barticle}[author]
\bauthor{\bsnm{Hall},~\bfnm{David~L}\binits{D.~L.}} \AND
  \bauthor{\bsnm{Llinas},~\bfnm{James}\binits{J.}}
(\byear{1997}).
\btitle{An introduction to multisensor data fusion}.
\bjournal{Proceedings of the IEEE}
\bvolume{85}
\bpages{6--23}.
\end{barticle}
\endbibitem

\bibitem[\protect\citeauthoryear{Haslbeck and Waldorp}{2015}]{haslbeck2015mgm}
\begin{barticle}[author]
\bauthor{\bsnm{Haslbeck},~\bfnm{Jonas}\binits{J.}} \AND
  \bauthor{\bsnm{Waldorp},~\bfnm{Lourens~J}\binits{L.~J.}}
(\byear{2015}).
\btitle{mgm: Structure Estimation for time-varying mixed graphical models in
  high-dimensional data}.
\bjournal{arXiv preprint arXiv:1510.06871}.
\end{barticle}
\endbibitem

\bibitem[\protect\citeauthoryear{Jalali et~al.}{2011}]{jalali2011learning}
\begin{binproceedings}[author]
\bauthor{\bsnm{Jalali},~\bfnm{Ali}\binits{A.}},
  \bauthor{\bsnm{Ravikumar},~\bfnm{Pradeep}\binits{P.}},
  \bauthor{\bsnm{Vasuki},~\bfnm{Vishvas}\binits{V.}} \AND
  \bauthor{\bsnm{Sanghavi},~\bfnm{Sujay}\binits{S.}}
(\byear{2011}).
\btitle{On learning discrete graphical models using group-sparse
  regularization}.
In \bbooktitle{Proceedings of the Fourteenth International Conference on
  Artificial Intelligence and Statistics}
\bpages{378--387}.
\end{binproceedings}
\endbibitem

\bibitem[\protect\citeauthoryear{King et~al.}{2003}]{king2003breast}
\begin{barticle}[author]
\bauthor{\bsnm{King},~\bfnm{Mary-Claire}\binits{M.-C.}},
  \bauthor{\bsnm{Marks},~\bfnm{Joan~H}\binits{J.~H.}},
  \bauthor{\bsnm{Mandell},~\bfnm{Jessica~B}\binits{J.~B.}} \betal{et~al.}
(\byear{2003}).
\btitle{Breast and ovarian cancer risks due to inherited mutations in BRCA1 and
  BRCA2}.
\bjournal{Science}
\bvolume{302}
\bpages{643--646}.
\end{barticle}
\endbibitem

\bibitem[\protect\citeauthoryear{Lee and Hastie}{2013}]{lee2013structure}
\begin{binproceedings}[author]
\bauthor{\bsnm{Lee},~\bfnm{Jason}\binits{J.}} \AND
  \bauthor{\bsnm{Hastie},~\bfnm{Trevor}\binits{T.}}
(\byear{2013}).
\btitle{Structure learning of mixed graphical models}.
In \bbooktitle{Artificial Intelligence and Statistics}
\bpages{388--396}.
\end{binproceedings}
\endbibitem

\bibitem[\protect\citeauthoryear{Li et~al.}{2014}]{li2014microrna}
\begin{barticle}[author]
\bauthor{\bsnm{Li},~\bfnm{Weiping}\binits{W.}},
  \bauthor{\bsnm{Liu},~\bfnm{Zhongyu}\binits{Z.}},
  \bauthor{\bsnm{Chen},~\bfnm{Li}\binits{L.}},
  \bauthor{\bsnm{Zhou},~\bfnm{Li}\binits{L.}} \AND
  \bauthor{\bsnm{Yao},~\bfnm{Yuanqing}\binits{Y.}}
(\byear{2014}).
\btitle{MicroRNA-23b is an independent prognostic marker and suppresses ovarian
  cancer progression by targeting runt-related transcription factor-2}.
\bjournal{FEBS letters}
\bvolume{588}
\bpages{1608--1615}.
\end{barticle}
\endbibitem

\bibitem[\protect\citeauthoryear{Liu, Roeder and
  Wasserman}{2010}]{liu2010stability}
\begin{binproceedings}[author]
\bauthor{\bsnm{Liu},~\bfnm{Han}\binits{H.}},
  \bauthor{\bsnm{Roeder},~\bfnm{Kathryn}\binits{K.}} \AND
  \bauthor{\bsnm{Wasserman},~\bfnm{Larry}\binits{L.}}
(\byear{2010}).
\btitle{Stability approach to regularization selection (stars) for high
  dimensional graphical models}.
In \bbooktitle{Advances in neural information processing systems}
\bpages{1432--1440}.
\end{binproceedings}
\endbibitem

\bibitem[\protect\citeauthoryear{Meinshausen and
  B{\"u}hlmann}{2006}]{meinshausen2006high}
\begin{barticle}[author]
\bauthor{\bsnm{Meinshausen},~\bfnm{Nicolai}\binits{N.}} \AND
  \bauthor{\bsnm{B{\"u}hlmann},~\bfnm{Peter}\binits{P.}}
(\byear{2006}).
\btitle{High-dimensional graphs and variable selection with the lasso}.
\bjournal{The annals of statistics}
\bpages{1436--1462}.
\end{barticle}
\endbibitem

\bibitem[\protect\citeauthoryear{Meinshausen and
  B{\"u}hlmann}{2010}]{meinshausen2010stability}
\begin{barticle}[author]
\bauthor{\bsnm{Meinshausen},~\bfnm{Nicolai}\binits{N.}} \AND
  \bauthor{\bsnm{B{\"u}hlmann},~\bfnm{Peter}\binits{P.}}
(\byear{2010}).
\btitle{Stability selection}.
\bjournal{Journal of the Royal Statistical Society: Series B (Statistical
  Methodology)}
\bvolume{72}
\bpages{417--473}.
\end{barticle}
\endbibitem

\bibitem[\protect\citeauthoryear{Meinshausen and
  Yu}{2009}]{meinshausen2009lasso}
\begin{barticle}[author]
\bauthor{\bsnm{Meinshausen},~\bfnm{Nicolai}\binits{N.}} \AND
  \bauthor{\bsnm{Yu},~\bfnm{Bin}\binits{B.}}
(\byear{2009}).
\btitle{Lasso-type recovery of sparse representations for high-dimensional
  data}.
\bjournal{The Annals of Statistics}
\bpages{246--270}.
\end{barticle}
\endbibitem

\bibitem[\protect\citeauthoryear{Nelsen}{1999}]{nelsen1999introduction}
\begin{bincollection}[author]
\bauthor{\bsnm{Nelsen},~\bfnm{Roger~B}\binits{R.~B.}}
(\byear{1999}).
\btitle{Introduction}.
In \bbooktitle{An Introduction to Copulas}
\bpages{1--4}.
\bpublisher{Springer}.
\end{bincollection}
\endbibitem

\bibitem[\protect\citeauthoryear{{The Cancer Genome Atlas Research
  Network}}{2011}]{cancer2011integrated}
\begin{barticle}[author]
\bauthor{\bsnm{{The Cancer Genome Atlas Research Network}}}
(\byear{2011}).
\btitle{Integrated genomic analyses of ovarian carcinoma}.
\bjournal{Nature}
\bvolume{474}
\bpages{609}.
\end{barticle}
\endbibitem

\bibitem[\protect\citeauthoryear{Obermayr
  et~al.}{2010}]{obermayr2010assessment}
\begin{barticle}[author]
\bauthor{\bsnm{Obermayr},~\bfnm{Eva}\binits{E.}},
  \bauthor{\bsnm{Sanchez-Cabo},~\bfnm{Fatima}\binits{F.}},
  \bauthor{\bsnm{Tea},~\bfnm{Muy-Kheng~M}\binits{M.-K.~M.}},
  \bauthor{\bsnm{Singer},~\bfnm{Christian~F}\binits{C.~F.}},
  \bauthor{\bsnm{Krainer},~\bfnm{Michael}\binits{M.}},
  \bauthor{\bsnm{Fischer},~\bfnm{Michael~B}\binits{M.~B.}},
  \bauthor{\bsnm{Sehouli},~\bfnm{Jalid}\binits{J.}},
  \bauthor{\bsnm{Reinthaller},~\bfnm{Alexander}\binits{A.}},
  \bauthor{\bsnm{Horvat},~\bfnm{Reinhard}\binits{R.}},
  \bauthor{\bsnm{Heinze},~\bfnm{Georg}\binits{G.}} \betal{et~al.}
(\byear{2010}).
\btitle{Assessment of a six gene panel for the molecular detection of
  circulating tumor cells in the blood of female cancer patients}.
\bjournal{BMC cancer}
\bvolume{10}
\bpages{666}.
\end{barticle}
\endbibitem

\bibitem[\protect\citeauthoryear{Ravikumar et~al.}{2010}]{ravikumar2010high}
\begin{barticle}[author]
\bauthor{\bsnm{Ravikumar},~\bfnm{Pradeep}\binits{P.}},
  \bauthor{\bsnm{Wainwright},~\bfnm{Martin~J}\binits{M.~J.}},
  \bauthor{\bsnm{Lafferty},~\bfnm{John~D}\binits{J.~D.}} \betal{et~al.}
(\byear{2010}).
\btitle{High-dimensional Ising model selection using $\ell_1$-regularized
  logistic regression}.
\bjournal{The Annals of Statistics}
\bvolume{38}
\bpages{1287--1319}.
\end{barticle}
\endbibitem

\bibitem[\protect\citeauthoryear{Shao}{1993}]{shao1993linear}
\begin{barticle}[author]
\bauthor{\bsnm{Shao},~\bfnm{Jun}\binits{J.}}
(\byear{1993}).
\btitle{Linear model selection by cross-validation}.
\bjournal{Journal of the American statistical Association}
\bvolume{88}
\bpages{486--494}.
\end{barticle}
\endbibitem

\bibitem[\protect\citeauthoryear{Shen, Olshen and
  Ladanyi}{2009}]{shen2009integrative}
\begin{barticle}[author]
\bauthor{\bsnm{Shen},~\bfnm{Ronglai}\binits{R.}},
  \bauthor{\bsnm{Olshen},~\bfnm{Adam~B}\binits{A.~B.}} \AND
  \bauthor{\bsnm{Ladanyi},~\bfnm{Marc}\binits{M.}}
(\byear{2009}).
\btitle{Integrative clustering of multiple genomic data types using a joint
  latent variable model with application to breast and lung cancer subtype
  analysis}.
\bjournal{Bioinformatics}
\bvolume{25}
\bpages{2906--2912}.
\end{barticle}
\endbibitem

\bibitem[\protect\citeauthoryear{Stone}{1974}]{stone1974cv}
\begin{barticle}[author]
\bauthor{\bsnm{Stone},~\bfnm{Mervyn}\binits{M.}}
(\byear{1974}).
\btitle{Cross-validatory choice and assessment of statistical predictions}.
\bjournal{Journal of the Royal Statistical Society: Series B (Methodological)}
\bvolume{36}
\bpages{111--133}.
\end{barticle}
\endbibitem

\bibitem[\protect\citeauthoryear{Su, Bogdan and Candes}{2015}]{su2015false}
\begin{barticle}[author]
\bauthor{\bsnm{Su},~\bfnm{Weijie}\binits{W.}},
  \bauthor{\bsnm{Bogdan},~\bfnm{Malgorzata}\binits{M.}} \AND
  \bauthor{\bsnm{Candes},~\bfnm{Emmanuel}\binits{E.}}
(\byear{2015}).
\btitle{False discoveries occur early on the lasso path}.
\bjournal{arXiv preprint arXiv:1511.01957}.
\end{barticle}
\endbibitem

\bibitem[\protect\citeauthoryear{Tibshirani}{1996}]{tibshirani1996regression}
\begin{barticle}[author]
\bauthor{\bsnm{Tibshirani},~\bfnm{Robert}\binits{R.}}
(\byear{1996}).
\btitle{Regression shrinkage and selection via the lasso}.
\bjournal{Journal of the Royal Statistical Society. Series B (Methodological)}
\bpages{267--288}.
\end{barticle}
\endbibitem

\bibitem[\protect\citeauthoryear{Tibshirani et~al.}{2013}]{tibshirani2013lasso}
\begin{barticle}[author]
\bauthor{\bsnm{Tibshirani},~\bfnm{Ryan~J}\binits{R.~J.}} \betal{et~al.}
(\byear{2013}).
\btitle{The lasso problem and uniqueness}.
\bjournal{Electronic Journal of Statistics}
\bvolume{7}
\bpages{1456--1490}.
\end{barticle}
\endbibitem

\bibitem[\protect\citeauthoryear{Tong et~al.}{2017}]{tong2017down}
\begin{bmisc}[author]
\bauthor{\bsnm{Tong},~\bfnm{Man}\binits{M.}},
  \bauthor{\bsnm{Wong},~\bfnm{Tin~Lok}\binits{T.~L.}},
  \bauthor{\bsnm{Luk},~\bfnm{Steve Tin-Chi}\binits{S.~T.-C.}},
  \bauthor{\bsnm{Che},~\bfnm{No{\'e}lia}\binits{N.}},
  \bauthor{\bsnm{Wong},~\bfnm{Kai~Yau}\binits{K.~Y.}},
  \bauthor{\bsnm{Fung},~\bfnm{Tsun~Ming}\binits{T.~M.}},
  \bauthor{\bsnm{Guan},~\bfnm{Xin-Yuan}\binits{X.-Y.}},
  \bauthor{\bsnm{Lee},~\bfnm{Nikki~P}\binits{N.~P.}},
  \bauthor{\bsnm{Yuan},~\bfnm{Yun-Fei}\binits{Y.-F.}},
  \bauthor{\bsnm{Lee},~\bfnm{Terence~K}\binits{T.~K.}} \betal{et~al.}
(\byear{2017}).
\btitle{Down-regulation of 4-hydroxyphenylpyruvate dioxygenate (HPD)
  contributes to the pathogenesis of hepatocellular carcinoma (HCC) through
  ERK/BCL-2 signalling activation}.
\end{bmisc}
\endbibitem

\bibitem[\protect\citeauthoryear{Toyama et~al.}{1999}]{toyama1999suppression}
\begin{barticle}[author]
\bauthor{\bsnm{Toyama},~\bfnm{Tatsuya}\binits{T.}},
  \bauthor{\bsnm{Iwase},~\bfnm{Hirotaka}\binits{H.}},
  \bauthor{\bsnm{Watson},~\bfnm{Peter}\binits{P.}},
  \bauthor{\bsnm{Muzik},~\bfnm{Huong}\binits{H.}},
  \bauthor{\bsnm{Saettler},~\bfnm{Elizabeth}\binits{E.}},
  \bauthor{\bsnm{Magliocco},~\bfnm{Anthony}\binits{A.}},
  \bauthor{\bsnm{DiFrancesco},~\bfnm{Lisa}\binits{L.}},
  \bauthor{\bsnm{Forsyth},~\bfnm{Peter}\binits{P.}},
  \bauthor{\bsnm{Garkavtsev},~\bfnm{Igor}\binits{I.}},
  \bauthor{\bsnm{Kobayashi},~\bfnm{Shunzo}\binits{S.}} \betal{et~al.}
(\byear{1999}).
\btitle{Suppression of ING1 expression in sporadic breast cancer.}
\bjournal{Oncogene}
\bvolume{18}.
\end{barticle}
\endbibitem

\bibitem[\protect\citeauthoryear{Tseng}{2001}]{tseng2001convergence}
\begin{barticle}[author]
\bauthor{\bsnm{Tseng},~\bfnm{Paul}\binits{P.}}
(\byear{2001}).
\btitle{Convergence of a block coordinate descent method for nondifferentiable
  minimization}.
\bjournal{Journal of optimization theory and applications}
\bvolume{109}
\bpages{475--494}.
\end{barticle}
\endbibitem

\bibitem[\protect\citeauthoryear{Wainwright}{2009}]{wainwright2009sharp}
\begin{barticle}[author]
\bauthor{\bsnm{Wainwright},~\bfnm{Martin~J}\binits{M.~J.}}
(\byear{2009}).
\btitle{Sharp thresholds for High-Dimensional and noisy sparsity recovery using
  $\ell_1$ - Constrained Quadratic Programming (Lasso)}.
\bjournal{IEEE transactions on information theory}
\bvolume{55}
\bpages{2183--2202}.
\end{barticle}
\endbibitem

\bibitem[\protect\citeauthoryear{Wang, Wainwright and
  Ramchandran}{2010}]{wang2010information}
\begin{barticle}[author]
\bauthor{\bsnm{Wang},~\bfnm{Wei}\binits{W.}},
  \bauthor{\bsnm{Wainwright},~\bfnm{Martin~J}\binits{M.~J.}} \AND
  \bauthor{\bsnm{Ramchandran},~\bfnm{Kannan}\binits{K.}}
(\byear{2010}).
\btitle{Information-theoretic limits on sparse signal recovery: Dense versus
  sparse measurement matrices}.
\bjournal{IEEE Transactions on Information Theory}
\bvolume{56}
\bpages{2967--2979}.
\end{barticle}
\endbibitem

\bibitem[\protect\citeauthoryear{Yan et~al.}{2016}]{yan2016mir}
\begin{barticle}[author]
\bauthor{\bsnm{Yan},~\bfnm{Jing}\binits{J.}},
  \bauthor{\bsnm{Jiang},~\bfnm{Jing-yi}\binits{J.-y.}},
  \bauthor{\bsnm{Meng},~\bfnm{Xiao-Na}\binits{X.-N.}},
  \bauthor{\bsnm{Xiu},~\bfnm{Yin-Ling}\binits{Y.-L.}} \AND
  \bauthor{\bsnm{Zong},~\bfnm{Zhi-Hong}\binits{Z.-H.}}
(\byear{2016}).
\btitle{MiR-23b targets cyclin G1 and suppresses ovarian cancer tumorigenesis
  and progression}.
\bjournal{Journal of Experimental \& Clinical Cancer Research}
\bvolume{35}
\bpages{31}.
\end{barticle}
\endbibitem

\bibitem[\protect\citeauthoryear{Yang et~al.}{2012}]{yang2012graphical}
\begin{binproceedings}[author]
\bauthor{\bsnm{Yang},~\bfnm{Eunho}\binits{E.}},
  \bauthor{\bsnm{Allen},~\bfnm{Genevera}\binits{G.}},
  \bauthor{\bsnm{Liu},~\bfnm{Zhandong}\binits{Z.}} \AND
  \bauthor{\bsnm{Ravikumar},~\bfnm{Pradeep~K}\binits{P.~K.}}
(\byear{2012}).
\btitle{Graphical models via generalized linear models}.
In \bbooktitle{Advances in Neural Information Processing Systems}
\bpages{1358--1366}.
\end{binproceedings}
\endbibitem

\bibitem[\protect\citeauthoryear{Yang et~al.}{2013}]{yang2013poisson}
\begin{binproceedings}[author]
\bauthor{\bsnm{Yang},~\bfnm{Eunho}\binits{E.}},
  \bauthor{\bsnm{Ravikumar},~\bfnm{Pradeep~K}\binits{P.~K.}},
  \bauthor{\bsnm{Allen},~\bfnm{Genevera~I}\binits{G.~I.}} \AND
  \bauthor{\bsnm{Liu},~\bfnm{Zhandong}\binits{Z.}}
(\byear{2013}).
\btitle{On Poisson graphical models}.
In \bbooktitle{Advances in Neural Information Processing Systems}
\bpages{1718--1726}.
\end{binproceedings}
\endbibitem

\bibitem[\protect\citeauthoryear{Yang et~al.}{2014a}]{yang2014general}
\begin{barticle}[author]
\bauthor{\bsnm{Yang},~\bfnm{Eunho}\binits{E.}},
  \bauthor{\bsnm{Ravikumar},~\bfnm{Pradeep}\binits{P.}},
  \bauthor{\bsnm{Allen},~\bfnm{Genevera~I}\binits{G.~I.}},
  \bauthor{\bsnm{Baker},~\bfnm{Yulia}\binits{Y.}},
  \bauthor{\bsnm{Wan},~\bfnm{Ying-Wooi}\binits{Y.-W.}} \AND
  \bauthor{\bsnm{Liu},~\bfnm{Zhandong}\binits{Z.}}
(\byear{2014}a).
\btitle{A General Framework for Mixed Graphical Models}.
\bjournal{arXiv preprint arXiv:1411.0288}.
\end{barticle}
\endbibitem

\bibitem[\protect\citeauthoryear{Yang et~al.}{2014b}]{yang2014mixed}
\begin{binproceedings}[author]
\bauthor{\bsnm{Yang},~\bfnm{Eunho}\binits{E.}},
  \bauthor{\bsnm{Baker},~\bfnm{Yulia}\binits{Y.}},
  \bauthor{\bsnm{Ravikumar},~\bfnm{Pradeep}\binits{P.}},
  \bauthor{\bsnm{Allen},~\bfnm{Genevera}\binits{G.}} \AND
  \bauthor{\bsnm{Liu},~\bfnm{Zhandong}\binits{Z.}}
(\byear{2014}b).
\btitle{Mixed graphical models via exponential families}.
In \bbooktitle{Artificial Intelligence and Statistics}
\bpages{1042--1050}.
\end{binproceedings}
\endbibitem

\bibitem[\protect\citeauthoryear{Yang et~al.}{2015}]{yang2015graphical}
\begin{barticle}[author]
\bauthor{\bsnm{Yang},~\bfnm{Eunho}\binits{E.}},
  \bauthor{\bsnm{Ravikumar},~\bfnm{Pradeep}\binits{P.}},
  \bauthor{\bsnm{Allen},~\bfnm{Genevera~I}\binits{G.~I.}} \AND
  \bauthor{\bsnm{Liu},~\bfnm{Zhandong}\binits{Z.}}
(\byear{2015}).
\btitle{Graphical models via univariate exponential family distributions.}
\bjournal{Journal of Machine Learning Research}
\bvolume{16}
\bpages{3813--3847}.
\end{barticle}
\endbibitem

\bibitem[\protect\citeauthoryear{Yu}{2013}]{yu2013stability}
\begin{barticle}[author]
\bauthor{\bsnm{Yu},~\bfnm{Bin}\binits{B.}}
(\byear{2013}).
\btitle{Stability}.
\bjournal{Bernoulli}
\bvolume{19}
\bpages{1484--1500}.
\bdoi{10.3150/13-BEJSP14}
\end{barticle}
\endbibitem

\bibitem[\protect\citeauthoryear{Yuan and Lin}{2007}]{yuan2007model}
\begin{barticle}[author]
\bauthor{\bsnm{Yuan},~\bfnm{Ming}\binits{M.}} \AND
  \bauthor{\bsnm{Lin},~\bfnm{Yi}\binits{Y.}}
(\byear{2007}).
\btitle{Model selection and estimation in the Gaussian graphical model}.
\bjournal{Biometrika}
\bvolume{94}
\bpages{19--35}.
\end{barticle}
\endbibitem

\bibitem[\protect\citeauthoryear{Zhang et~al.}{2010}]{zhang2010nearly}
\begin{barticle}[author]
\bauthor{\bsnm{Zhang},~\bfnm{Cun-Hui}\binits{C.-H.}} \betal{et~al.}
(\byear{2010}).
\btitle{Nearly unbiased variable selection under minimax concave penalty}.
\bjournal{The Annals of statistics}
\bvolume{38}
\bpages{894--942}.
\end{barticle}
\endbibitem

\bibitem[\protect\citeauthoryear{Zhang and Huang}{2008}]{zhang2008sparsity}
\begin{barticle}[author]
\bauthor{\bsnm{Zhang},~\bfnm{Cun-Hui}\binits{C.-H.}} \AND
  \bauthor{\bsnm{Huang},~\bfnm{Jian}\binits{J.}}
(\byear{2008}).
\btitle{The sparsity and bias of the lasso selection in high-dimensional linear
  regression}.
\bjournal{The Annals of Statistics}
\bpages{1567--1594}.
\end{barticle}
\endbibitem

\bibitem[\protect\citeauthoryear{Zhao and Yu}{2006}]{zhao2006model}
\begin{barticle}[author]
\bauthor{\bsnm{Zhao},~\bfnm{Peng}\binits{P.}} \AND
  \bauthor{\bsnm{Yu},~\bfnm{Bin}\binits{B.}}
(\byear{2006}).
\btitle{On model selection consistency of Lasso}.
\bjournal{Journal of Machine learning research}
\bvolume{7}
\bpages{2541--2563}.
\end{barticle}
\endbibitem

\bibitem[\protect\citeauthoryear{Zou}{2006}]{zou2006adaptive}
\begin{barticle}[author]
\bauthor{\bsnm{Zou},~\bfnm{Hui}\binits{H.}}
(\byear{2006}).
\btitle{The adaptive lasso and its oracle properties}.
\bjournal{Journal of the American statistical association}
\bvolume{101}
\bpages{1418--1429}.
\end{barticle}
\endbibitem

\end{thebibliography}

\newpage

\appendix

\section{B-RAIL Convergence} \label{app:converge}

\begin{figure}[h!]
\centering
\begin{tabular}[c]{c}
\begin{subfigure}[b]{0.7\linewidth}
\includegraphics[width=\textwidth]{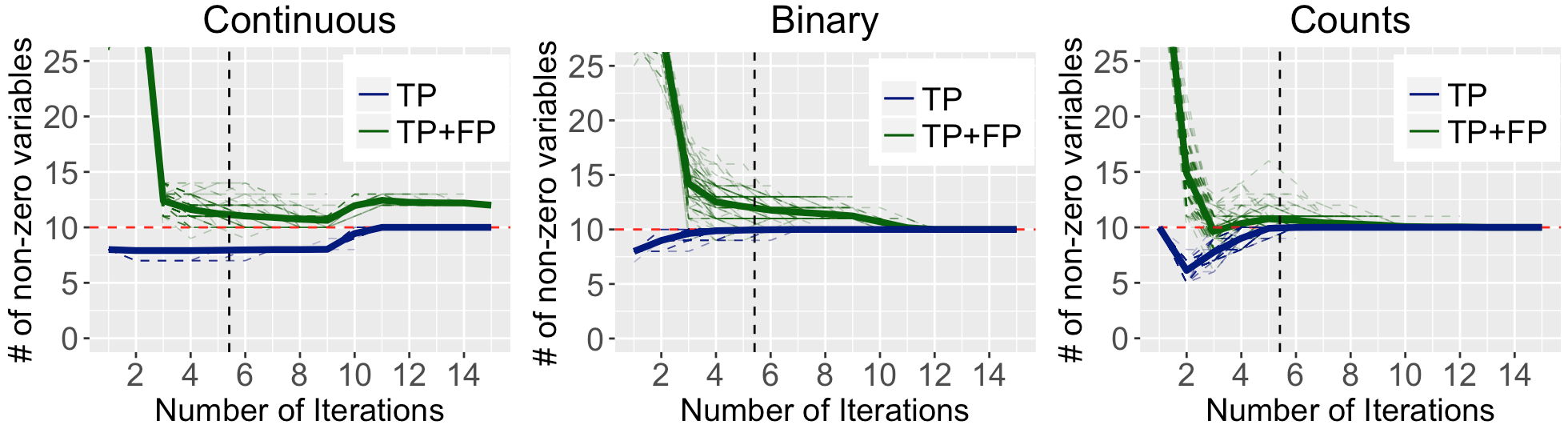}
  \captionof{figure}{Gaussian response, OV data}
\end{subfigure}\\
\begin{subfigure}[b]{0.7\linewidth}
\includegraphics[width=\textwidth]{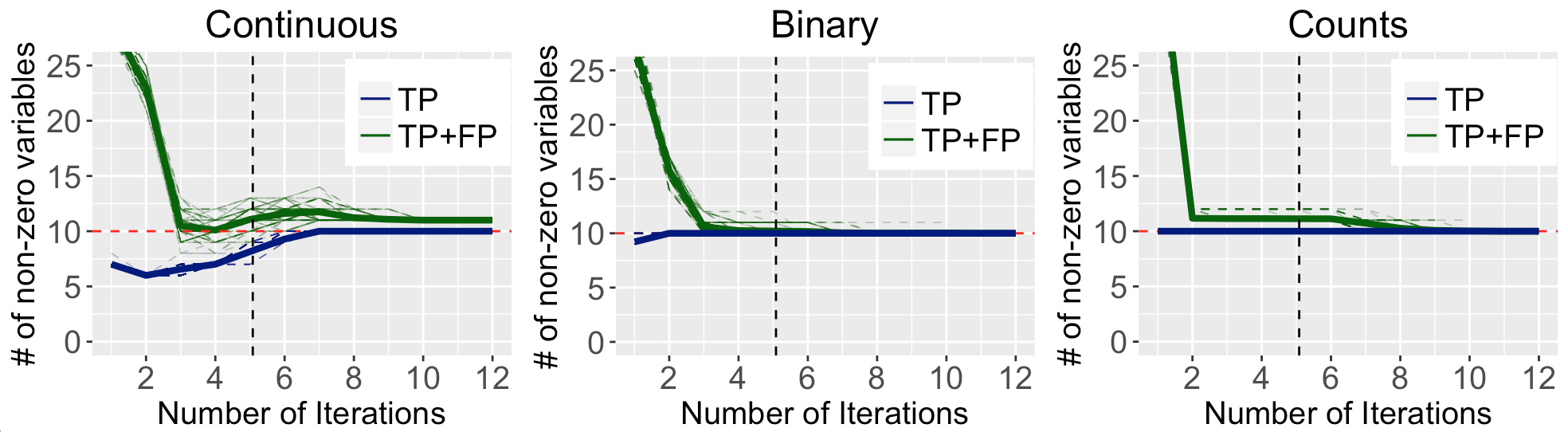}
  \captionof{figure}{Binary response, OV data}
\end{subfigure} \\
\begin{subfigure}[b]{0.7\linewidth}
\includegraphics[width=\textwidth]{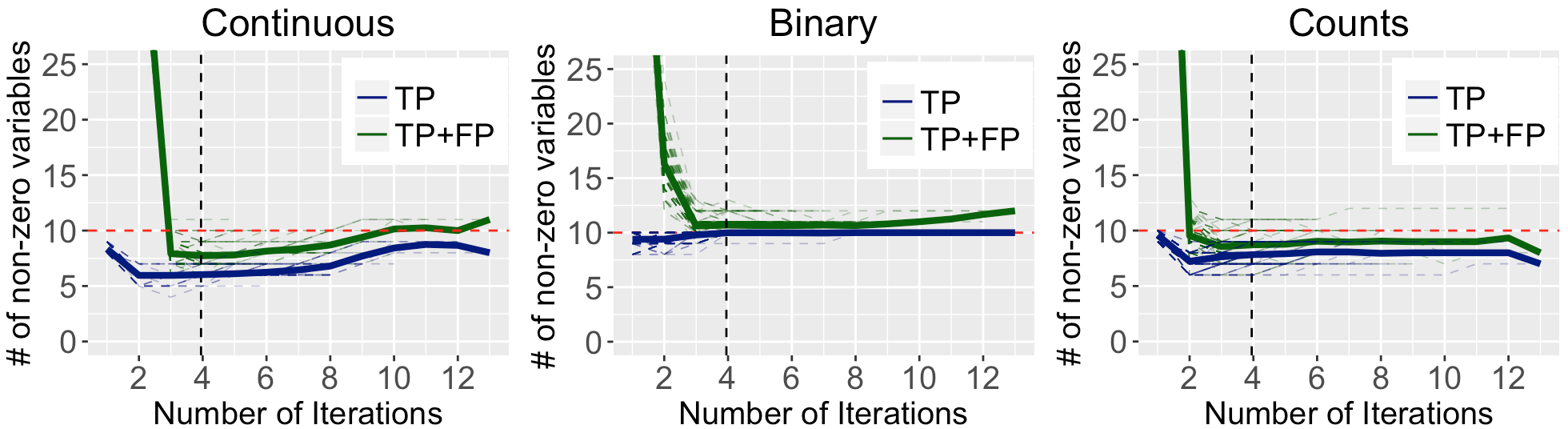}
  \captionof{figure}{Poisson response, OV data}
\end{subfigure}\\
  \end{tabular}    
  \caption{For each block, we report signal recovery across iterations of the B-RAIL algorithm. The solid lines indicate the average feature recovery over 100 runs, and the faint dashed lines in the background represent individual runs. We represent the total number of selected features in green and the number of true positives in blue. The red horizontal line indicates the number of true features in each block, and the vertical black line represents the average number of iterations until B-RAIL converges over the 100 runs.}
  \label{fig:Conv}
\end{figure}

We investigate the convergence of the B-RAIL algorithm. Our empirical analysis indicates that B-RAIL has quick support convergence for all simulation scenarios. We demonstrate this convergence for one type of simulation in Figure~\ref{fig:Conv}. Here, we simulate data using predictors from the TCGA ovarian cancer data (see Section~\ref{sec:OVCase}) for three types of responses: (a) continuous, (b) binary, and (c) counts. We report the number of iterations over 100 runs of the B-RAIL algorithm and denote the average number of iterations until convergence by the dashed vertical black line. We see that the average number of iterations is between 4 and 5 with the maximum number of iterations reaching 15. These ranges were similar for all simulation designs, confirming relatively fast convergence of B-RAIL. Furthermore, we also show the true positive rates and the total number of selected features in Figure~\ref{fig:Conv} to highlight B-RAIL's convergence to a relatively accurate solution.

\begin{algorithm}[H]
{\fontsize{8}{8}\selectfont
\caption{Blockwise Lasso Algorithm}
\begin{flushleft}
\textbf{Given} fixed sequence of regularization parameters $\lambda_{k, j}$ \\
\textbf{Initialize} $t=0$ and $\hat{\bbeta}^{(0)}_{k}$ to have a fixed proportion of sparsity for $k = 1, \ldots K$ \\
\vspace{1mm}
\textbf{Do} until $\Supp(\hat{\bbeta}^{(t)})$ stops changing:  \\ 
\begin{itemize} 
\setlength{\itemindent}{-1em}
\item Set $t = t+1$.
\item For $k = 1, \dots, K$, estimate $\hat{\bbeta}^{(t)}_k$ blockwise, holding $\hat{\bbeta}^{(t)}_l$ ($l < k$) and $\hat{\bbeta}^{(t-1)}_l$ ($l > k$) fixed:
\begin{align*}
 \hat{\bbeta}_k^{(t)} = \argmin_{\alpha, \bbeta} -\frac{1}{n} \ell \left(\by; \alpha + \bX_k \bbeta + \bPhi^{(t)}_k \right) + \sum_{j=1}^{p_k} \lambda_{k,j} |\beta_{j}|
\end{align*}
where $\bPhi^{(t)}_k = \sum_{l <  k} \bX_l \hat{\bbeta}^{(t)}_l   + \sum_{l >  k} \bX_l \hat{\bbeta}^{(t-1)}_l $.
\end{itemize}
\textbf{Output} $\hat{\bbeta}_1 , \ldots \hat{\bbeta}_{K}$.
\end{flushleft}
\label{alg:block_lasso}
}
\end{algorithm}

We provide the blockwise Lasso algorithm in Algorithm~\ref{alg:block_lasso} and prove its convergence below.

\begin{prop}
Suppose that the objective function in \eqref{eq:block_lasso} is bounded below. Then the blockwise (GLM) Lasso converges to a global minimum of\eqref{eq:block_lasso}.
\end{prop}

\begin{proof}
Let $f(z_1, \dots, z_k)$ denote the objective function in \eqref{eq:block_lasso}, where $z_k = (\alpha_k, \bbeta_k)$ for each $k = 1, \dots, K$. Since the domain of the GLM log-likelihood $\ell$ is open and Gateaux-differentiable on its domain, then $f$ is regular at each point in the domain of $f$ by Lemma 3.1 in \citet{tseng2001convergence}. Note also that since the GLM negative log-likelihood is convex and the $\ell_1$ penalty term in \eqref{eq:block_lasso} is convex and separable, then $f(z_1, \dots, z_K)$ is convex with respect to each block $z_k$ for $k = 1, \dots, K$. Because $f$ is regular and convex with respect to each block $z_k$ ($k = 1, \dots, K$), it follows that the blockwise (GLM) Lasso always converges to a stationary point of $f$ by Theorem 4.1 in \citet{tseng2001convergence}. By convexity of $f$, this implies that the blockwise (GLM) Lasso converges to a global minimum.
\end{proof}

\section{Additional B-RAIL Simulations}\label{app:sims}
%
%

We provide the following figures and tables to supplement our simulations and to further support the strong empirical performance of B-RAIL.

\begin{figure}[H]
\begin{tabular}[c]{cc}
\captionsetup{type=figure}\addtocounter{figure}{-1}
\begin{subfigure}[b]{0.44\textwidth}
\includegraphics[width=.85\linewidth]{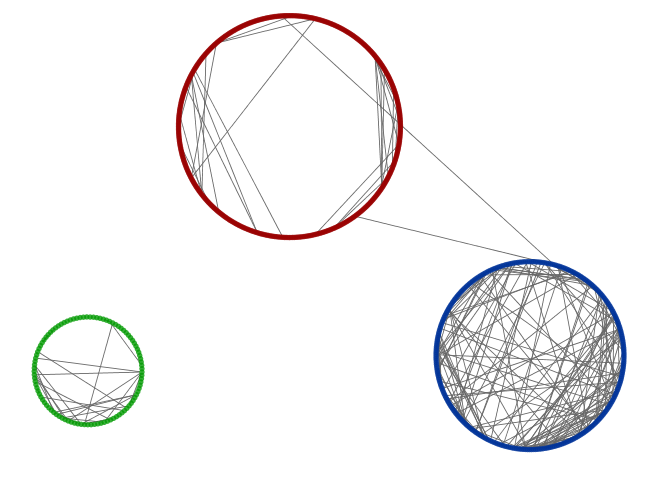}
  \captionof{figure}{MGM with EBIC}
\end{subfigure}& 
\hspace{.5em}
\begin{subfigure}[b]{0.44\textwidth}
\includegraphics[width=.85\linewidth]{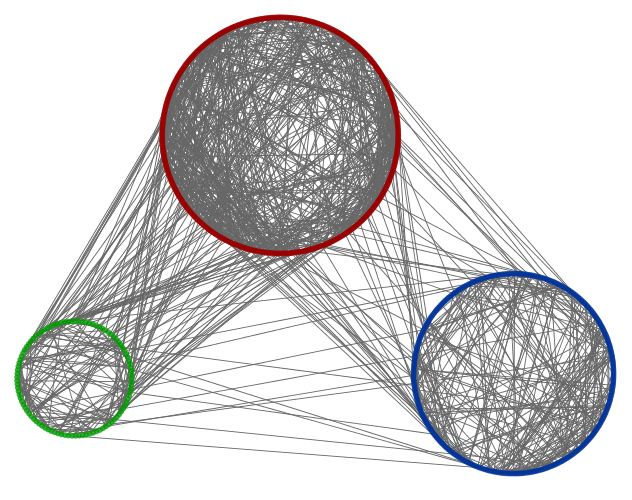}
  \captionof{figure}{MGM with CV}
\end{subfigure}\\
\multicolumn{2}{c}{
\hspace{.5em}
\begin{subfigure}[b]{0.44\textwidth}
\includegraphics[width=.85\linewidth]{BRAIL_AND.png}
  \captionof{figure}{The B-RAIL algorithm}
\end{subfigure} }\\
  \end{tabular}   
  \captionof{figure}{We compare three different graph selection methods (B-RAIL and two methods from the mgm R package) when applied to real ovarian cancer genomics data. The data is comprised of three blocks: RNASeq (red), miRNA (blue), and methylation (green), with $n = 293$ and overall $p = 836$.
}
    \label{fig:MGM}
\end{figure}

To augment the motivating example in Figure~\ref{fig:MotivPlot}, we provide comparisons of B-RAIL to two other mixed graphical selection methods from the mgm R package \citep{haslbeck2015mgm} in Figure~\ref{fig:MGM}. mgm takes a node-wise neighborhood estimation approach, and for each node, mgm selects the Lasso regularization parameter using either the extended BIC or CV, fits a penalized GLM model, and applies additional thresholding to the estimated coefficients to remove noise. Here, we used the ``AND'' rule to combine estimated neighborhoods for all three graphs. (Note that we had to convert the proportion-valued methylation values into 0-1 binary values in order to comply with mgm package restrictions.)

From Figure~\ref{fig:MGM}, we see that mgm with the extended BIC selection criteria tends to under-select features while mgm with CV often over-selects. This agrees with our simulations and discussion of the Lasso-type model selection biases in Section~\ref{sec:Sims}. We also observe that like the Lasso-type methods in Figure~\ref{fig:MotivPlot}, mgm with CV and extended BIC can result in imbalanced selection between the blocks.

\begin{figure}[h!]
\centering
\begin{tabular}[c]{cc}
\begin{subfigure}[b]{0.45\linewidth}
\includegraphics[width=\textwidth]{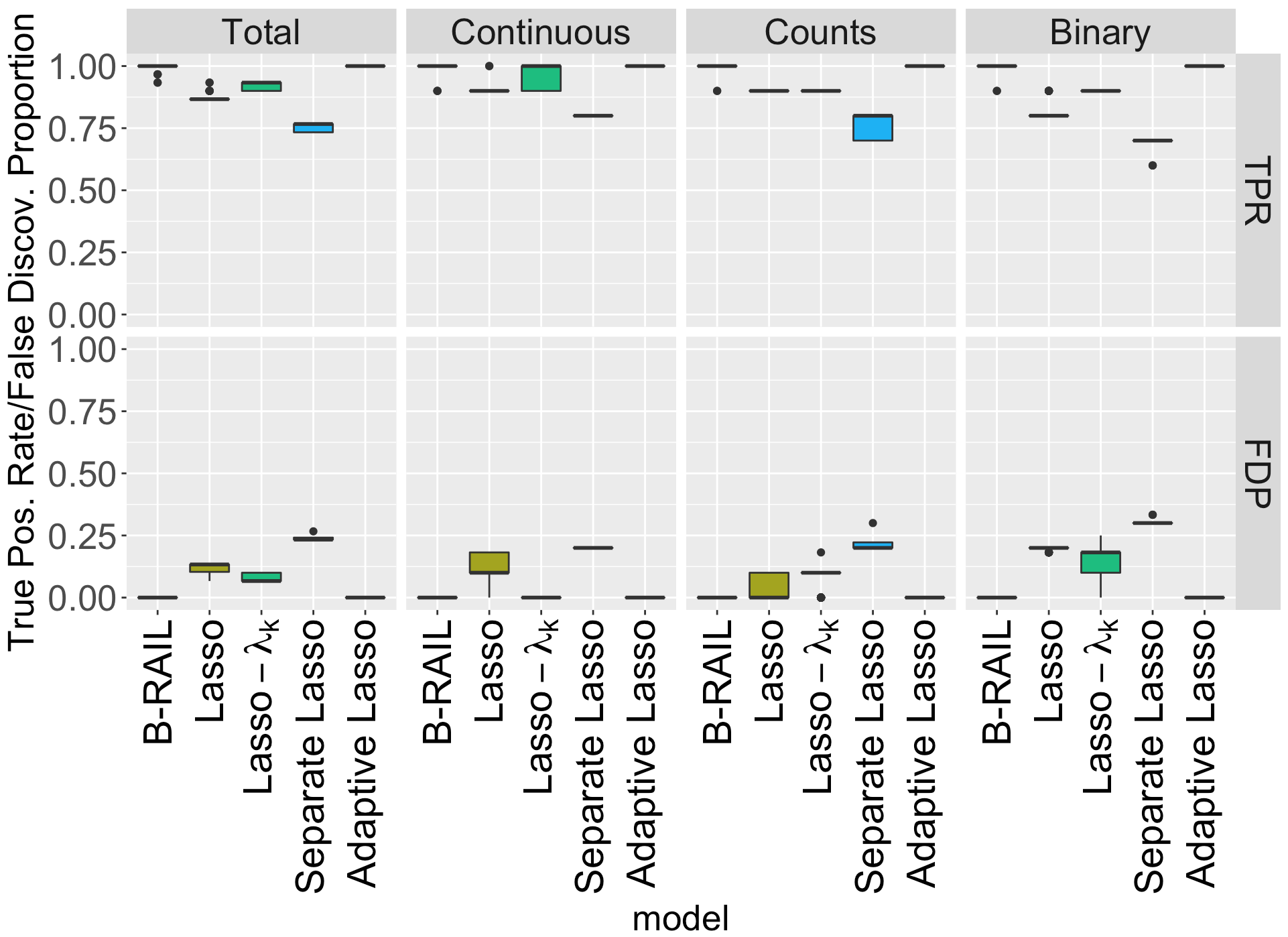}
  \captionof{figure}{iid case, $p = 300$}
\end{subfigure}&
\hspace*{\fill} 
\begin{subfigure}[b]{0.45\linewidth}
\includegraphics[width=\textwidth]{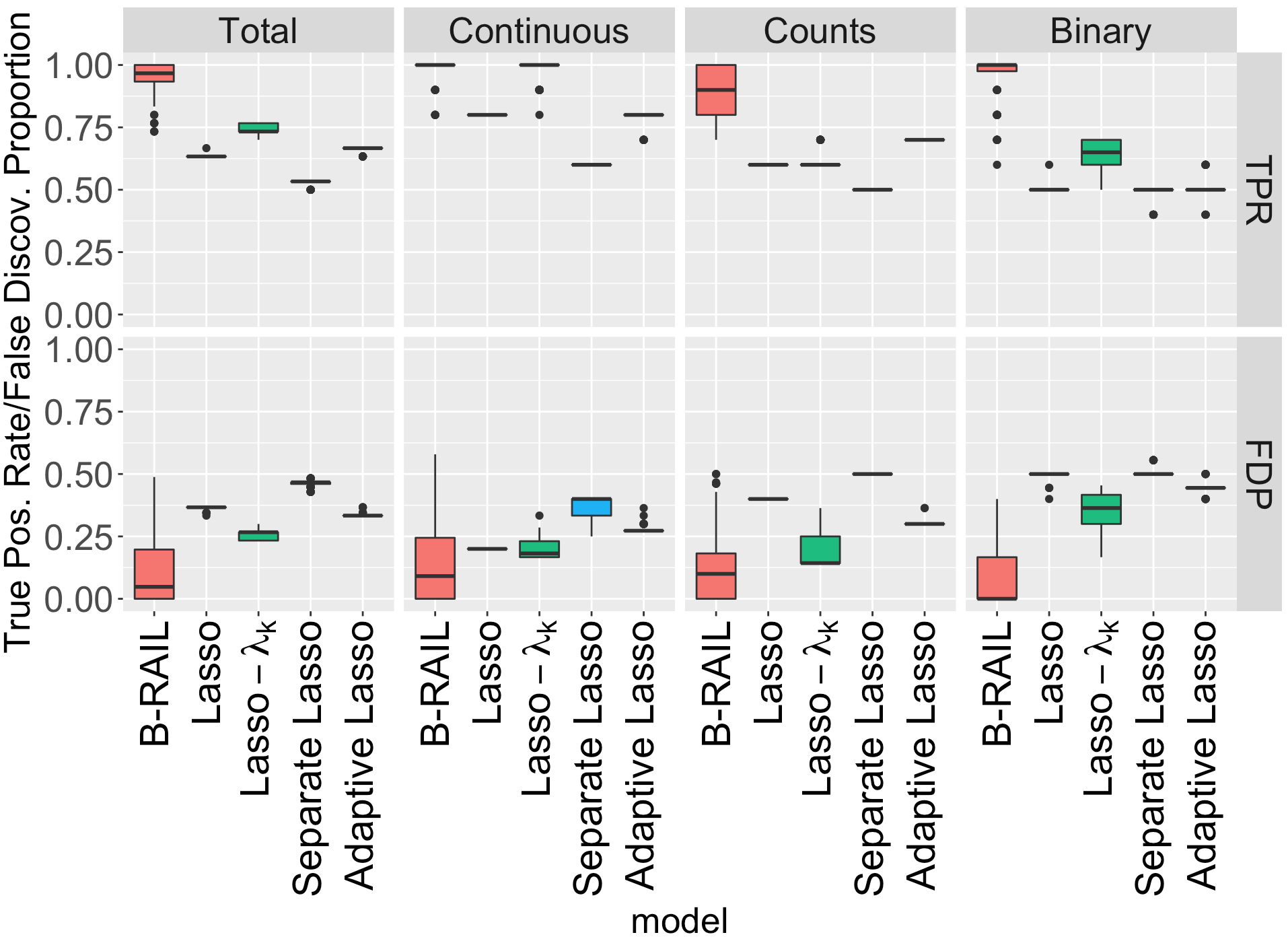}
  \captionof{figure}{iid case, $p = 900$}
\end{subfigure}\\
\hspace*{\fill} 
\begin{subfigure}[b]{0.45\linewidth}
\includegraphics[width=\textwidth]{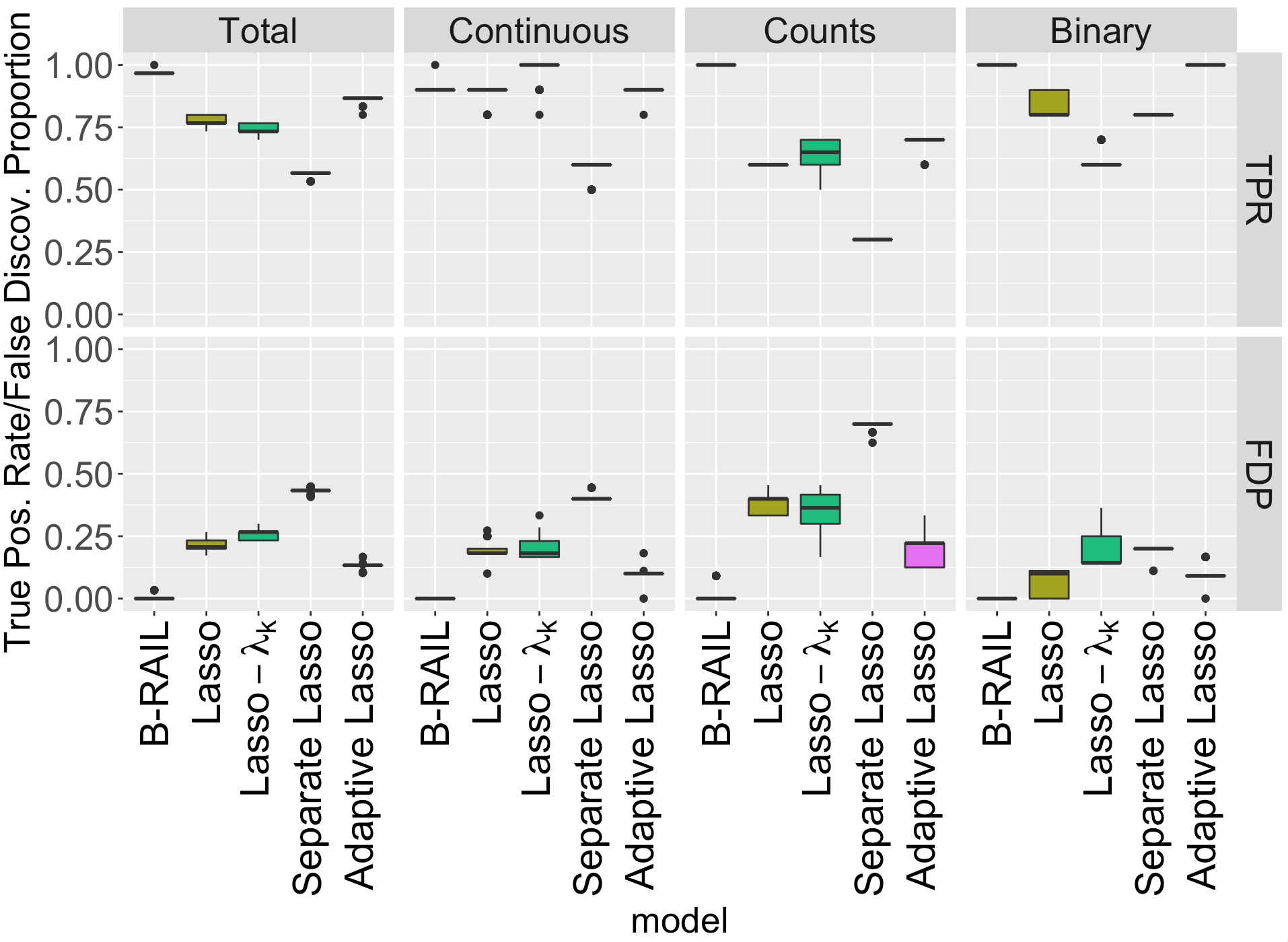}
  \captionof{figure}{Non-constant variance}
\end{subfigure}&
\hspace*{\fill} 
\begin{subfigure}[b]{0.45\linewidth}
\includegraphics[width=\textwidth]{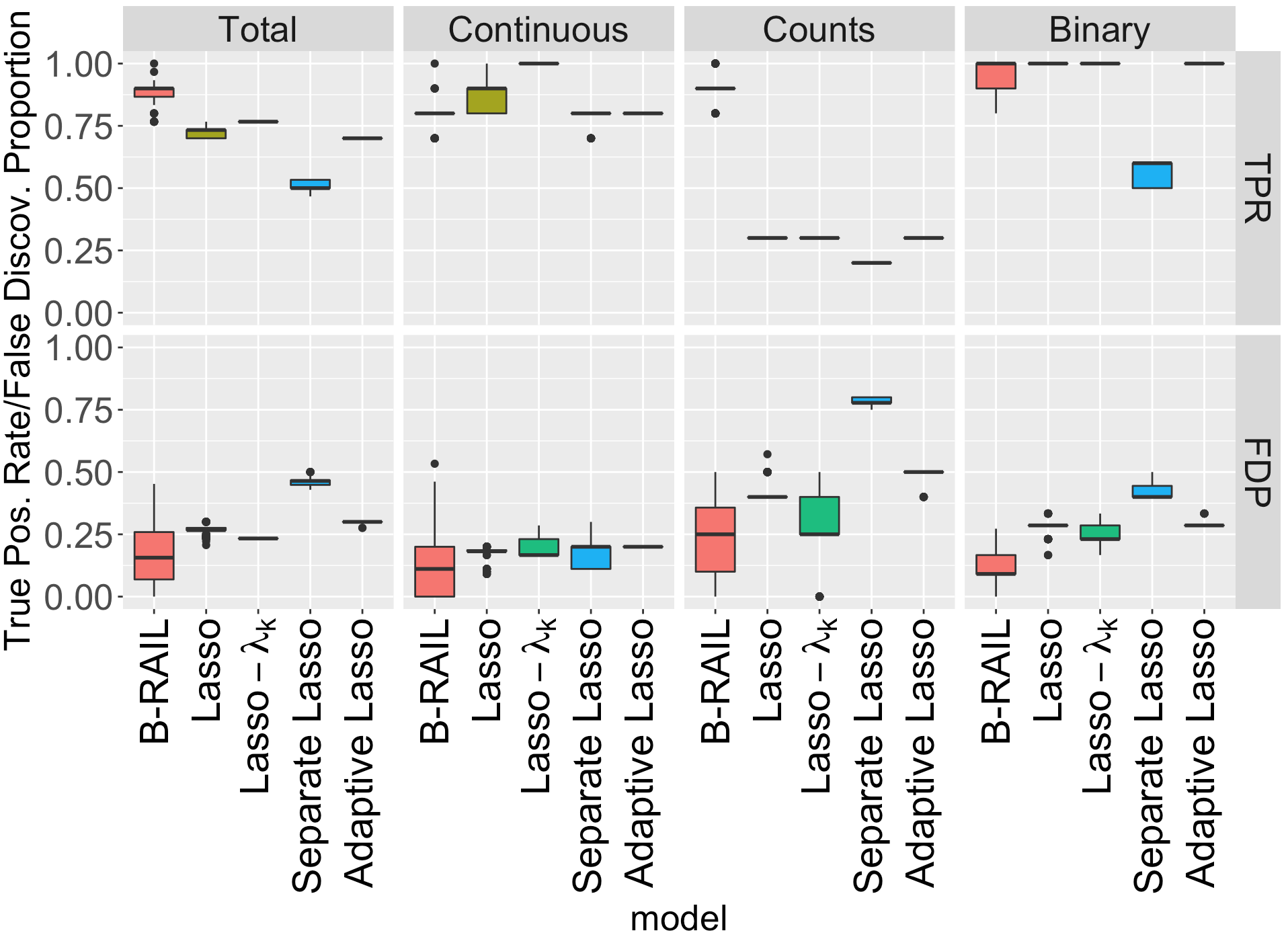}
  \captionof{figure}{Block directed graph structure}
\end{subfigure}\\
\multicolumn{2}{c}{
\hspace{.5em}
\begin{subfigure}[b]{0.45\textwidth}
\includegraphics[width=\textwidth]{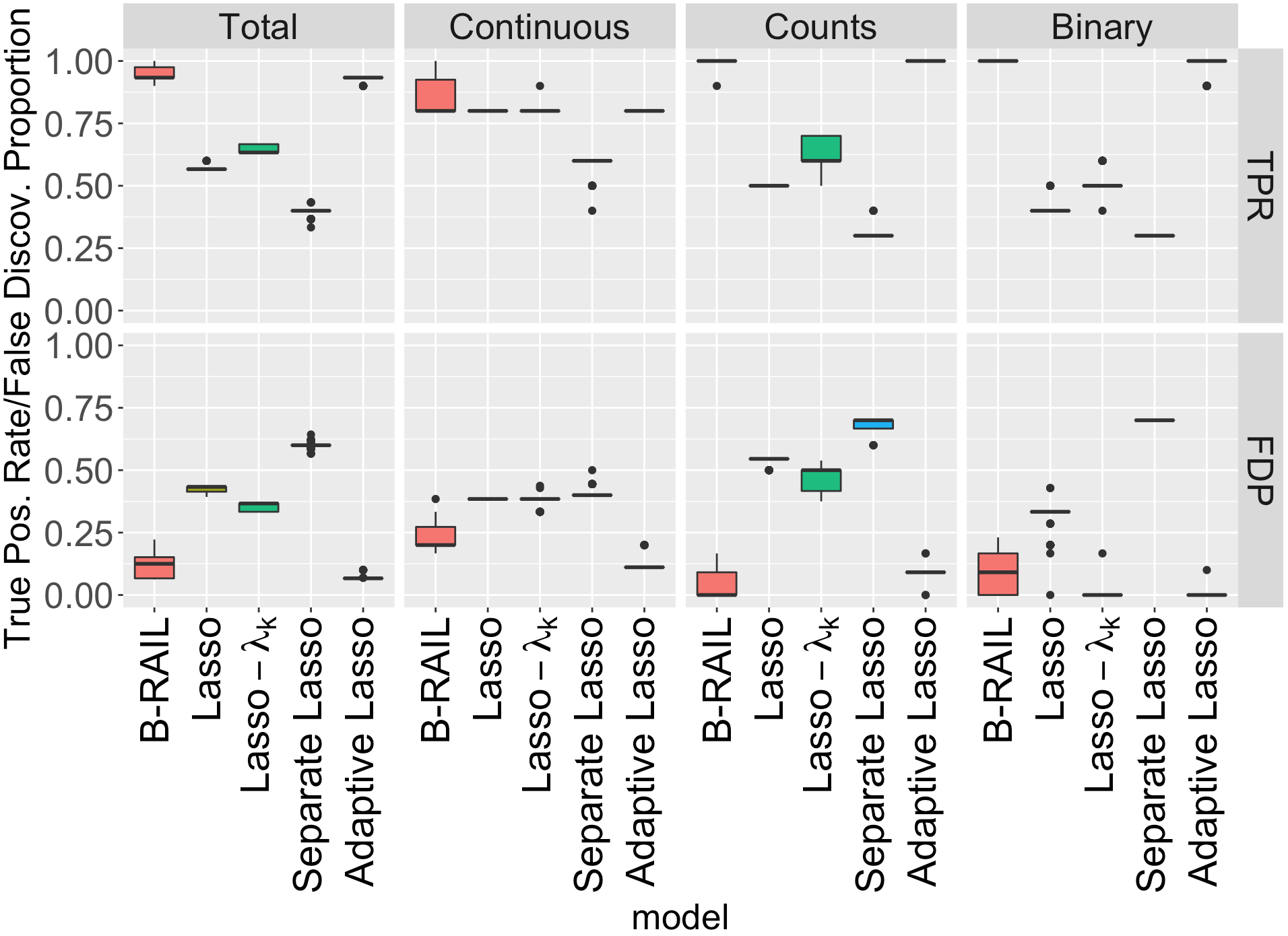}
  \captionof{figure}{OV data}
\end{subfigure} }
  \end{tabular}    
  \caption{We compare various selection methods under 5 different simulation scenarios. For each scenario, we simulate $\X$ with three blocks (continuous, binary, counts) and a Gaussian response $\y$. We report the TPR and FDP for overall feature recovery and individual block recoveries across 200 runs. Note that we used oracle information for the Lasso-type methods.}
  \label{fig:fig3}
\end{figure}

\newpage

Figure~\ref{fig:fig3} duplicates the information in Table~\ref{tab:Gaus} but using boxplots for better visualization. We recall that these simulations compared B-RAIL and various Lasso-type methods (using oracle information) under four simulation designs with Gaussian responses (see Section~\ref{sec:Sims} for further details). In almost all of these simulations, B-RAIL is able to achieve a higher TPR while maintaining a low FDP.

\begin{figure}[h!]
\centering
\begin{tabular}[c]{cc}
\begin{subfigure}[b]{0.45\linewidth}
\includegraphics[width=\textwidth]{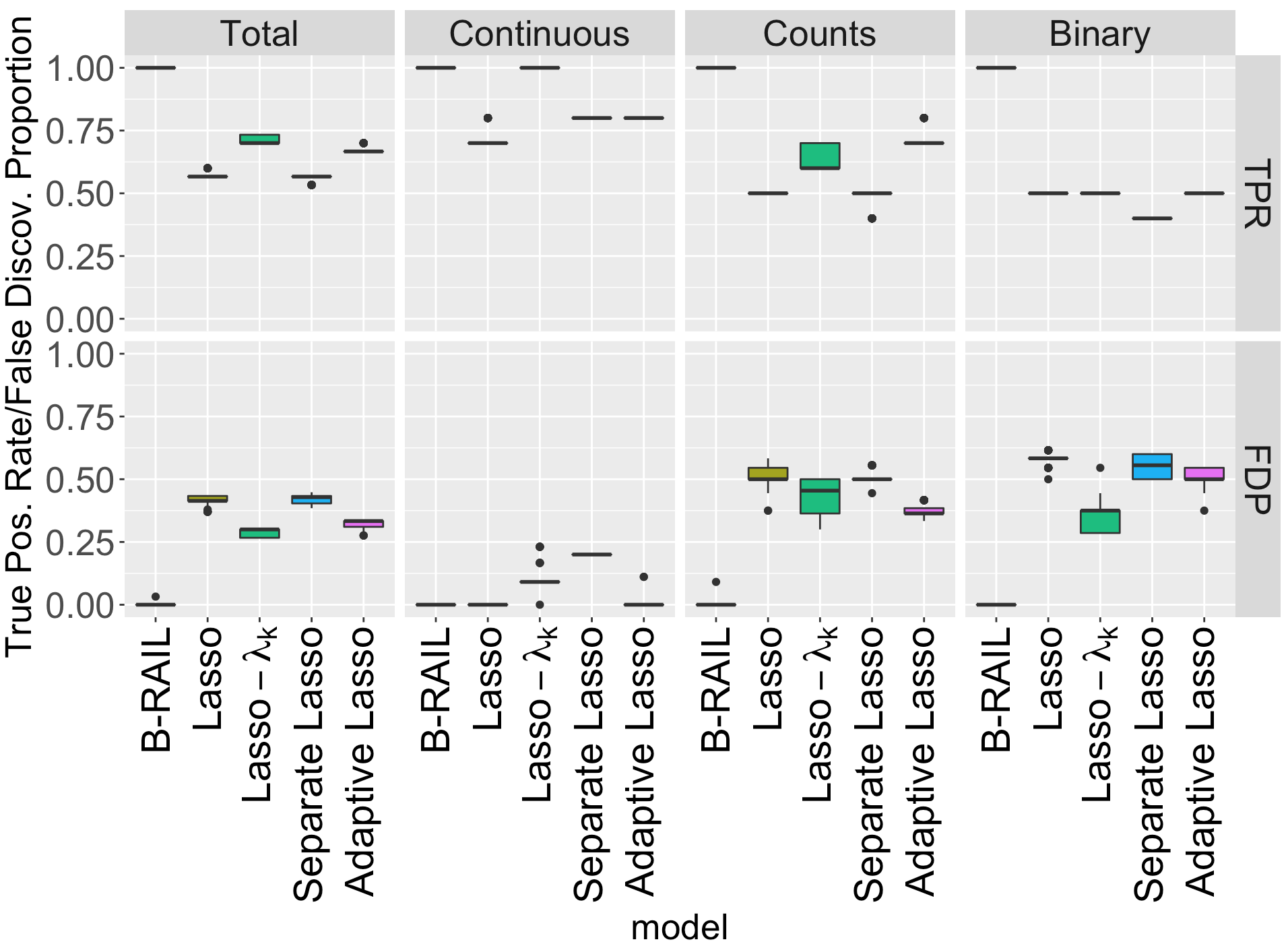}
  \captionof{figure}{$p_{1}=50$ $p_{2}=250$, $p_{3}=500$, $||\beta||_0=10$ in each block}
\end{subfigure}&
\hspace*{\fill} 
\begin{subfigure}[b]{0.45\linewidth}
\includegraphics[width=\textwidth]{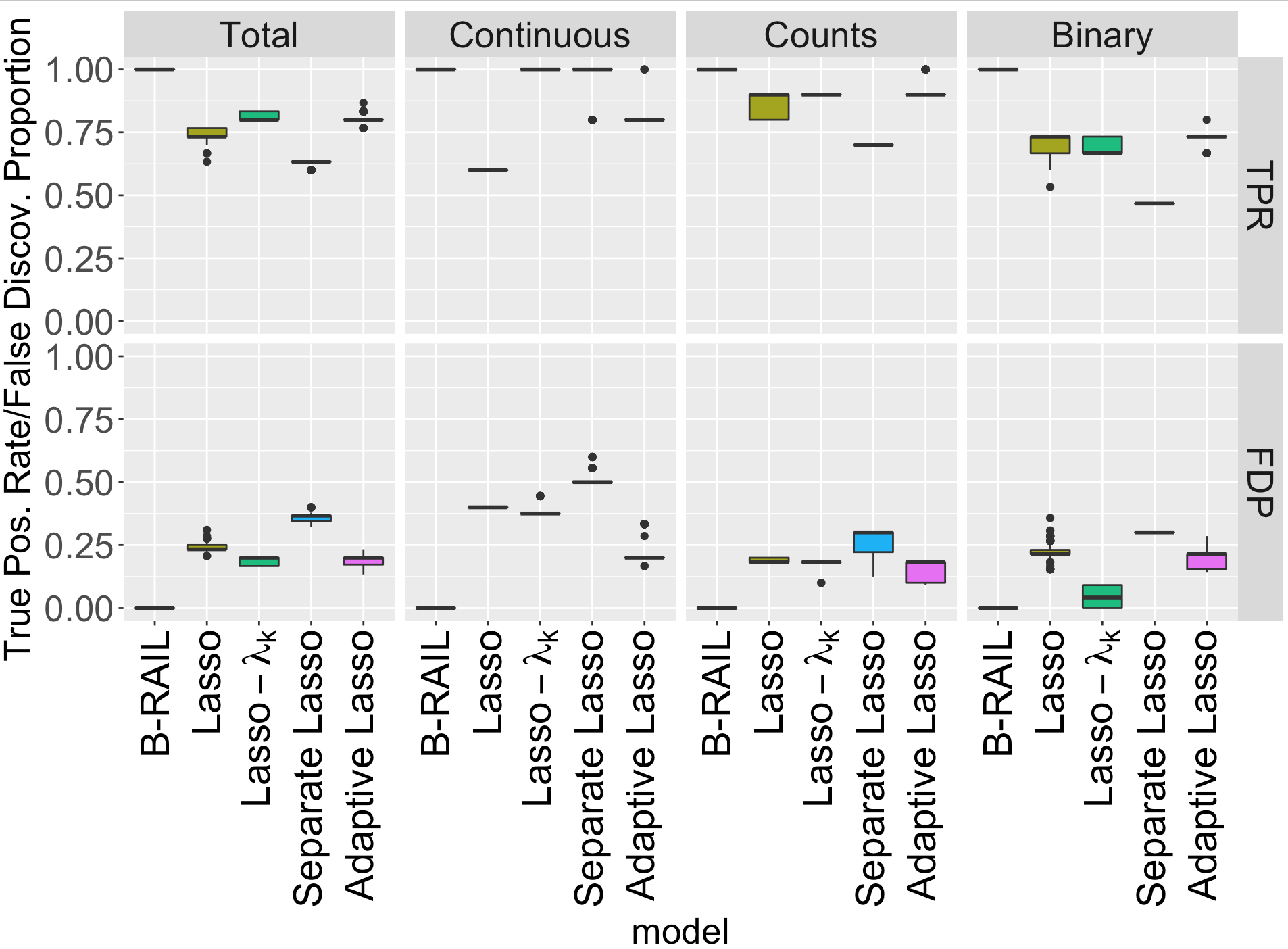}
  \captionof{figure}{$p_{1}=50$ $p_{2}=250$, $p_{3}=500$, $||\beta_1||_0=5$, $||\beta_2||_0=10$, $||\beta_3||_0=15$}
\end{subfigure}\\
  \end{tabular}    
  \caption{We compare various selection methods for non-uniform choices of $p_k$ and $||\bbeta_k||_0$ under the iid simulation design. Here, we simulated $\X$ with three blocks (continuous, binary, counts) and a Gaussian response $\y$. Note that there were 200 runs and that we used oracle information for the Lasso-type methods.}
  \label{fig:figDiffpk}
\end{figure}

We next verify that the above simulation results are not heavily dependent on our choice of $p_k$ and $||\beta_k||_0$. In Figure~\ref{fig:figDiffpk}, we ran the iid simulation design with Gaussian responses for different non-uniform values of $p_k$ and $||\beta_k||_0$. These results show that B-RAIL can successfully recover features from unequally sized blocks with different amounts of sparsity while other methods may struggle to account for biases introduced by the different $p_k$'s and $||\beta_k||_0$'s.


\newpage

\begin{table*}
\makebox[\linewidth]{
\scalebox{0.7}{
\begin{tabular}{l  ll  ll  ll  ll}

\addlinespace
\multicolumn{9}{c}{ \begin{normalsize}  iid Case, Poisson Response \end{normalsize} } \\
\midrule

  & \multicolumn{2}{c}{Total} &  \multicolumn{2}{c}{Continuous} & \multicolumn{2}{c}{ Binary} & \multicolumn{2}{c}{Counts}  \\ 
 \cmidrule(lr){2-3}
\cmidrule(lr){4-5}
\cmidrule(lr){6-7}
\cmidrule(lr){8-9}

 & \multicolumn{1}{c}{TPR} &  \multicolumn{1}{c}{FDP}  &  \multicolumn{1}{c}{TPR}  &  \multicolumn{1}{c}{FDP} &  \multicolumn{1}{c}{TPR}  & \multicolumn{1}{c}{FDP}  &  \multicolumn{1}{c}{TPR}  &  \multicolumn{1}{c}{FDP} \\ 
\midrule
B-RAIL & \bft{0.86 (6.3e-3)} & \bft{0.10 (4.7e-3)} & 0.92 (3.8e-3) & 0.15 (8.2e-3) & 0.90 (6.2e-3) & 0.00 (0.0e-0) & 0.77 (1.1e-2) & 0.13 (6.4e-3) \\ 
  Lasso - $\lambda$ (oracle) & 0.67 (3.3e-4) & 0.33 (5.4e-4) & 0.90 (1.0e-3) & 0.18 (1.8e-4) & 0.50 (0.0e-0) & 0.28 (1.2e-3) & 0.60 (0.0e-0) & 0.50 (9.0e-4) \\ 
  Lasso - $\lambda$ (oracle) & 0.74 (1.6e-3) & 0.26 (1.6e-3) & 0.93 (4.6e-3) & 0.21 (3.5e-3) & 0.70 (4.9e-3) & 0.29 (5.9e-3) & 0.60 (1.4e-3) & 0.26 (3.1e-3) \\ 
  Separate Lasso (oracle) & 0.55 (1.6e-3) & 0.43 (2.2e-3) & 0.70 (0.0e-0) & 0.30 (7.8e-4) & 0.44 (4.8e-3) & 0.51 (5.9e-3) & 0.50 (0.0e-0) & 0.48 (2.7e-3) \\ 
  Adaptive Lasso (oracle) & 0.65 (2.2e-3) & 0.34 (1.8e-3) & 0.77 (4.5e-3) & 0.29 (2.5e-3) & 0.47 (4.6e-3) & 0.29 (6.1e-3) & 0.70 (0.0e-0) & 0.41 (3.6e-3) \\ 
\addlinespace
\multicolumn{9}{c}{ \begin{normalsize} Block Directed Graph Structure, Poisson Response \end{normalsize}} \\
\midrule

& \multicolumn{2}{c}{Total} &  \multicolumn{2}{c}{Continuous} & \multicolumn{2}{c}{ Binary} & \multicolumn{2}{c}{Counts}  \\ 
\cmidrule(lr){2-3}
\cmidrule(lr){4-5}
\cmidrule(lr){6-7}
\cmidrule(lr){8-9}
 & \multicolumn{1}{c}{TPR} &  \multicolumn{1}{c}{FDP}  &  \multicolumn{1}{c}{TPR}  &  \multicolumn{1}{c}{FDP} &  \multicolumn{1}{c}{TPR}  & \multicolumn{1}{c}{FDP}  &  \multicolumn{1}{c}{TPR}  &  \multicolumn{1}{c}{FDP} \\ 
\midrule
B-RAIL & \bft{0.81 (8.6e-3)} & \bft{0.07 (5.2e-3)} & 0.80 (0.0e-0) & 0.04 (5.4e-3) & 0.97 (5.7e-3) & 0.09 (6.2e-3) & 0.68 (2.1e-2) & 0.07 (1.2e-2) \\ 
  Lasso - $\lambda$ (oracle) & 0.70 (0.0e-0) & 0.29 (1.3e-3) & 0.90 (0.0e-0) & 0.23 (4.1e-3) & 0.90 (0.0e-0) & 0.37 (2.0e-3) & 0.30 (0.0e-0) & 0.18 (1.3e-2) \\ 
  Lasso - $\lambda$ (oracle) & 0.70 (8.5e-4) & 0.30 (8.5e-4) & 0.89 (2.6e-3) & 0.21 (4.3e-3) & 0.90 (0.0e-0) & 0.27 (3.6e-3) & 0.30 (0.0e-0) & 0.52 (7.4e-3) \\ 
  Separate Lasso (oracle) & 0.50 (1.3e-3) & 0.46 (2.3e-3) & 0.80 (0.0e-0) & 0.20 (0.0e-0) & 0.51 (3.9e-3) & 0.41 (7.0e-3) & 0.20 (0.0e-0) & 0.79 (1.7e-3) \\ 
  Adaptive Lasso (oracle) & 0.70 (9.6e-4) & 0.29 (1.4e-3) & 0.81 (2.9e-3) & 0.20 (1.3e-3) & 1.00 (0.0e-0) & 0.36 (2.5e-3) & 0.30 (0.0e-0) & 0.27 (5.7e-3) \\ 
    
\addlinespace
\multicolumn{9}{c}{ \begin{normalsize} OV Data, Poisson Response \end{normalsize}} \\
\midrule

& \multicolumn{2}{c}{Total} &  \multicolumn{2}{c}{Continuous} & \multicolumn{2}{c}{Proportion} & \multicolumn{2}{c}{Counts}  \\ 
\cmidrule(lr){2-3}
\cmidrule(lr){4-5}
\cmidrule(lr){6-7}
\cmidrule(lr){8-9}
 & \multicolumn{1}{c}{TPR} &  \multicolumn{1}{c}{FDP}  &  \multicolumn{1}{c}{TPR}  &  \multicolumn{1}{c}{FDP} &  \multicolumn{1}{c}{TPR}  & \multicolumn{1}{c}{FDP}  &  \multicolumn{1}{c}{TPR}  &  \multicolumn{1}{c}{FDP} \\ 
\midrule
B-RAIL & \bft{0.99 (2.1e-3)} & \bft{0.05 (3.3e-3)} & 0.98 (6.4e-3) & 0.11 (5.3e-3) & 1.00 (0.0e-0) & 0.01 (2.2e-3) & 1.00 (0.0e-0) & 0.03 (4.3e-3) \\ 
  Lasso - $\lambda$  (oracle) & 0.51 (1.9e-3) & 0.46 (1.9e-3) & 0.54 (5.6e-3) & 0.45 (3.9e-3) & 0.60 (0.0e-0) & 0.31 (3.8e-3) & 0.40 (0.0e-0) & 0.61 (2.0e-3) \\ 
  Lasso - $\lambda_k$  (oracle) & 0.60 (1.1e-3) & 0.40 (1.1e-3) & 0.83 (4.8e-3) & 0.46 (2.0e-3) & 0.61 (4.9e-3) & 0.21 (8.0e-3) & 0.36 (7.9e-3) & 0.43 (1.5e-2) \\ 
  Separate Lasso (oracle) & 0.38 (2.2e-3) & 0.60 (1.9e-3) & 0.43 (4.6e-3) & 0.53 (4.4e-3) & 0.32 (3.8e-3) & 0.65 (3.2e-3) & 0.40 (1.0e-3) & 0.60 (6.7e-4) \\ 
  Adaptive Lasso (oracle) & 0.93 (1.5e-3) & 0.07 (1.2e-3) & 0.80 (0.0e-0) & 0.12 (3.2e-3) & 1.00 (1.0e-3) & 0.00 (0.0e-0) & 0.98 (4.2e-3) & 0.09 (3.8e-4) \\

\specialrule{1pt}{1pt}{1pt}
\end{tabular}
}}
\vspace{3pt}
\caption{We compare various selection methods under 3 different simulation scenarios. For each scenario, we simulate $\X$ with three blocks (continuous, binary, counts) and a Poisson response $\y$. We report the TPR and FDP for overall feature recovery and individual block recoveries, averaged across 200 runs with standard errors in parentheses. We bold the best overall $TPR * (1 - FDP)$ values for each simulation scenario.} \label{tab:Pois}
\end{table*}

\begin{table*}
\captionsetup{justification=raggedright, singlelinecheck=false}
\makebox[\linewidth]{
\scalebox{0.7}{
\begin{tabular}{l  ll  ll  ll  ll}

\addlinespace
\multicolumn{9}{c}{ \begin{normalsize}  iid Case, Binary Response \end{normalsize} } \\
\midrule

  & \multicolumn{2}{c}{Total} &  \multicolumn{2}{c}{Continuous} & \multicolumn{2}{c}{ Binary} & \multicolumn{2}{c}{Counts}  \\ 
 \cmidrule(lr){2-3}
\cmidrule(lr){4-5}
\cmidrule(lr){6-7}
\cmidrule(lr){8-9}

 & \multicolumn{1}{c}{TPR} &  \multicolumn{1}{c}{FDP}  &  \multicolumn{1}{c}{TPR}  &  \multicolumn{1}{c}{FDP} &  \multicolumn{1}{c}{TPR}  & \multicolumn{1}{c}{FDP}  &  \multicolumn{1}{c}{TPR}  &  \multicolumn{1}{c}{FDP} \\ 
\midrule
B-RAIL & \bft{0.83 (9.4e-3)} & \bft{0.11 (1.0e-2)} & 0.92 (6.9e-3) & 0.09 (1.0e-2) & 0.73 (1.5e-2) & 0.12 (1.3e-2) & 0.82 (1.0e-2) & 0.11 (1.2e-2) \\ 
  Lasso - $\lambda$ (oracle) & 0.65 (2.9e-3) & 0.34 (3.0e-3) & 0.80 (2.0e-3) & 0.15 (4.9e-3) & 0.54 (8.3e-3) & 0.42 (6.9e-3) & 0.60 (0.0e-0) & 0.43 (3.9e-3) \\ 
  Lasso - $\lambda_k$ (oracle) & 0.72 (3.7e-3) & 0.31 (2.8e-3) & 0.87 (5.6e-3) & 0.18 (8.0e-3) & 0.69 (9.8e-3) & 0.37 (7.6e-3) & 0.61 (2.6e-3) & 0.37 (6.0e-3) \\ 
  Separate Lasso (oracle) & 0.51 (1.5e-3) & 0.47 (2.3e-3) & 0.60 (0.0e-0) & 0.39 (2.8e-3) & 0.43 (4.4e-3) & 0.55 (5.0e-3) & 0.50 (0.0e-0) & 0.47 (4.3e-3) \\ 
  Adaptive Lasso (oracle) & 0.60 (3.0e-3) & 0.39 (2.7e-3) & 0.91 (3.0e-3) & 0.50 (3.7e-3) & 0.39 (7.8e-3) & 0.05 (9.6e-3) & 0.49 (9.9e-3) & 0.31 (6.2e-3) \\ 
    
\addlinespace
\multicolumn{9}{c}{ \begin{normalsize} Block Directed Graph Structure, Binary Response \end{normalsize}} \\
\midrule

& \multicolumn{2}{c}{Total} &  \multicolumn{2}{c}{Continuous} & \multicolumn{2}{c}{ Binary} & \multicolumn{2}{c}{Counts}  \\ 
\cmidrule(lr){2-3}
\cmidrule(lr){4-5}
\cmidrule(lr){6-7}
\cmidrule(lr){8-9}
 & \multicolumn{1}{c}{TPR} &  \multicolumn{1}{c}{FDP}  &  \multicolumn{1}{c}{TPR}  &  \multicolumn{1}{c}{FDP} &  \multicolumn{1}{c}{TPR}  & \multicolumn{1}{c}{FDP}  &  \multicolumn{1}{c}{TPR}  &  \multicolumn{1}{c}{FDP} \\ 
\midrule
B-RAIL & \bft{0.75 (3.5e-3)} & \bft{0.15 (5.0e-3)} & 0.90 (0.0e-0) & 0.01 (3.6e-3) & 0.83 (8.9e-3) & 0.23 (7.2e-3) & 0.51 (3.8e-3) & 0.19 (1.1e-2) \\ 
  Lasso - $\lambda$ (oracle) & 0.68 (1.9e-3) & 0.30 (1.8e-3) & 0.80 (0.0e-0) & 0.20 (5.7e-3) & 0.95 (5.6e-3) & 0.30 (3.9e-3) & 0.30 (0.0e-0) & 0.48 (5.7e-3) \\ 
  Lasso - $\lambda_k$ (oracle) & 0.69 (1.8e-3) & 0.31 (1.9e-3) & 0.80 (1.4e-3) & 0.12 (3.8e-3) & 0.96 (5.2e-3) & 0.22 (6.1e-3) & 0.30 (0.0e-0) & 0.64 (6.6e-3) \\ 
  Separate Lasso (oracle) & 0.49 (2.9e-3) & 0.49 (2.6e-3) & 0.74 (5.9e-3) & 0.20 (5.7e-3) & 0.51 (6.8e-3) & 0.45 (6.3e-3) & 0.21 (3.3e-3) & 0.78 (3.2e-3) \\ 
  Adaptive Lasso (oracle) & 0.68 (2.3e-3) & 0.31 (2.7e-3) & 0.96 (5.2e-3) & 0.38 (4.5e-3) & 0.98 (3.7e-3) & 0.21 (6.4e-3) & 0.09 (2.9e-3) & 0.24 (2.9e-2) \\ 
    
\addlinespace
\multicolumn{9}{c}{ \begin{normalsize} OV Data, Binary Response \end{normalsize}} \\
\midrule

& \multicolumn{2}{c}{Total} &  \multicolumn{2}{c}{Continuous} & \multicolumn{2}{c}{Proportion} & \multicolumn{2}{c}{Counts}  \\ 
\cmidrule(lr){2-3}
\cmidrule(lr){4-5}
\cmidrule(lr){6-7}
\cmidrule(lr){8-9}
 & \multicolumn{1}{c}{TPR} &  \multicolumn{1}{c}{FDP}  &  \multicolumn{1}{c}{TPR}  &  \multicolumn{1}{c}{FDP} &  \multicolumn{1}{c}{TPR}  & \multicolumn{1}{c}{FDP}  &  \multicolumn{1}{c}{TPR}  &  \multicolumn{1}{c}{FDP} \\ 
\midrule
B-RAIL & \bft{0.81 (3.5e-3)} & \bft{0.11 (4.9e-3)} & 0.64 (7.8e-3) & 0.21 (8.5e-3) & 1.00 (0.0e-0) & 0.05 (5.9e-3) & 0.80 (6.1e-3) & 0.09 (8.5e-3) \\ 
  Lasso - $\lambda$ (oracle) & 0.59 (2.9e-3) & 0.39 (2.8e-3) & 0.80 (0.0e-0) & 0.35 (2.8e-3) & 0.49 (7.2e-3) & 0.33 (7.6e-3) & 0.49 (5.7e-3) & 0.49 (5.0e-3) \\ 
  Lasso - $\lambda_k$ (oracle) & 0.63 (2.2e-3) & 0.37 (2.2e-3) & 0.80 (0.0e-0) & 0.39 (3.8e-3) & 0.71 (7.1e-3) & 0.40 (5.3e-3) & 0.36 (5.6e-3) & 0.19 (1.6e-2) \\ 
  Separate Lasso (oracle) & 0.44 (2.1e-3) & 0.54 (1.9e-3) & 0.59 (3.1e-3) & 0.40 (3.4e-3) & 0.34 (5.0e-3) & 0.63 (4.3e-3) & 0.40 (0.0e-0) & 0.60 (1.3e-3) \\ 
  Adaptive Lasso (oracle) & 0.83 (2.9e-3) & 0.16 (2.6e-3) & 0.79 (3.8e-3) & 0.22 (3.8e-3) & 0.99 (3.0e-3) & 0.01 (2.6e-3) & 0.72 (6.3e-3) & 0.24 (6.9e-3) \\ 

\specialrule{1pt}{1pt}{1pt}
\end{tabular}
}}
\vspace{3pt}
\caption{We compare various selection methods under 3 different simulation scenarios. For each scenario, we simulate $\X$ with three blocks (continuous, binary, counts) and a binary response $\y$. We report the TPR and FDP for overall feature recovery and individual block recoveries, averaged across 200 runs with standard errors in parentheses. We bold the best overall $TPR * (1 - FDP)$ values for each simulation scenario.} \label{tab:Bin}
\end{table*}

For Poisson and binary responses, Tables~\ref{tab:Pois} and \ref{tab:Bin} provide the results from additional simulation designs to supplement Table~\ref{tab:BinPois}. Here, the response $\y$ and predictors $\X$ were simulated according to the description in Section~\ref{sec:Sims} with $n = 200$, $p_1 = p_2 = p_3 = 300$, and $||\bbeta_k||_0 = 10$ for each block. 

For easier visualizations, we duplicate the results of Tables~\ref{tab:Pois} and \ref{tab:Bin} using boxplots in Figure~\ref{fig:fig4}. As shown by the plots, B-RAIL is able to achieve higher TPR and maintain low FDP across various simulations for both binary and Poisson responses. 

\begin{figure}[H]
\centering
\begin{tabular}[c]{cc}
\hspace*{\fill} 
\begin{subfigure}[b]{0.45\linewidth}
\includegraphics[width=\textwidth]{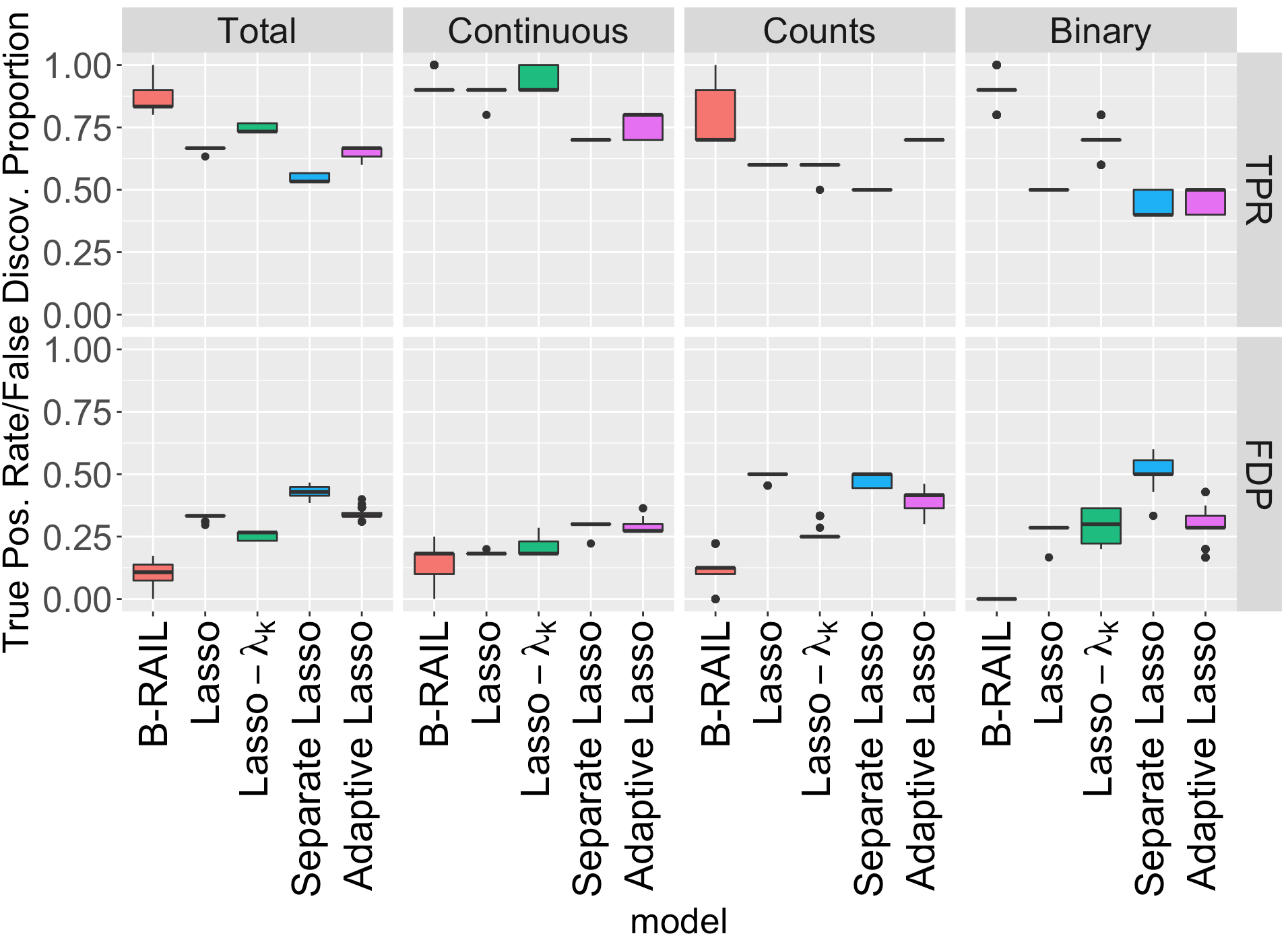}
  \captionof{figure}{Binary response, iid case}
\end{subfigure}&
\hspace*{\fill} 
\begin{subfigure}[b]{0.45\linewidth}
\includegraphics[width=\textwidth]{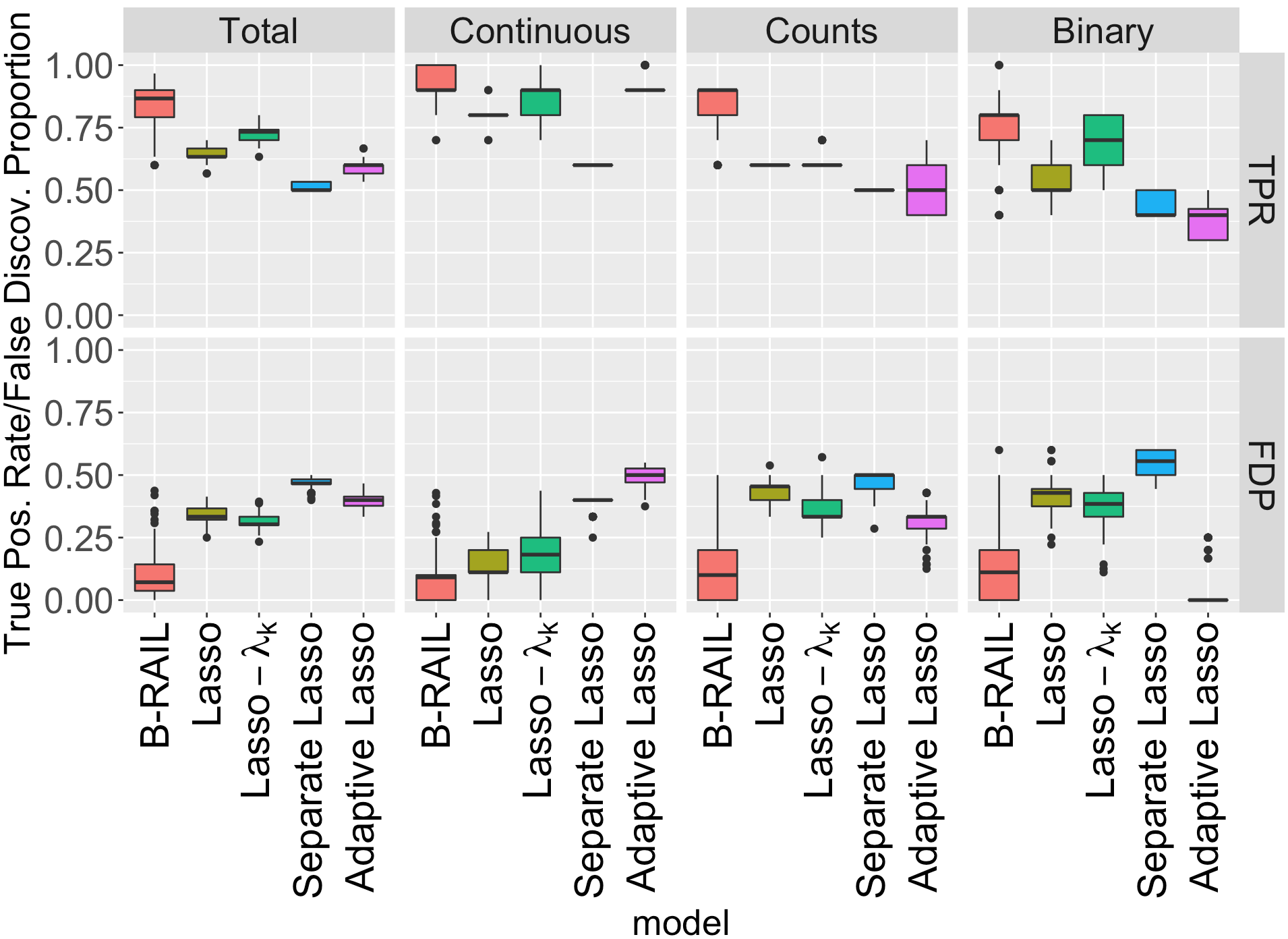}
  \captionof{figure}{Poisson response, iid case}
\end{subfigure}\\
\hspace*{\fill} 
\begin{subfigure}[b]{0.45\linewidth}
\includegraphics[width=\textwidth]{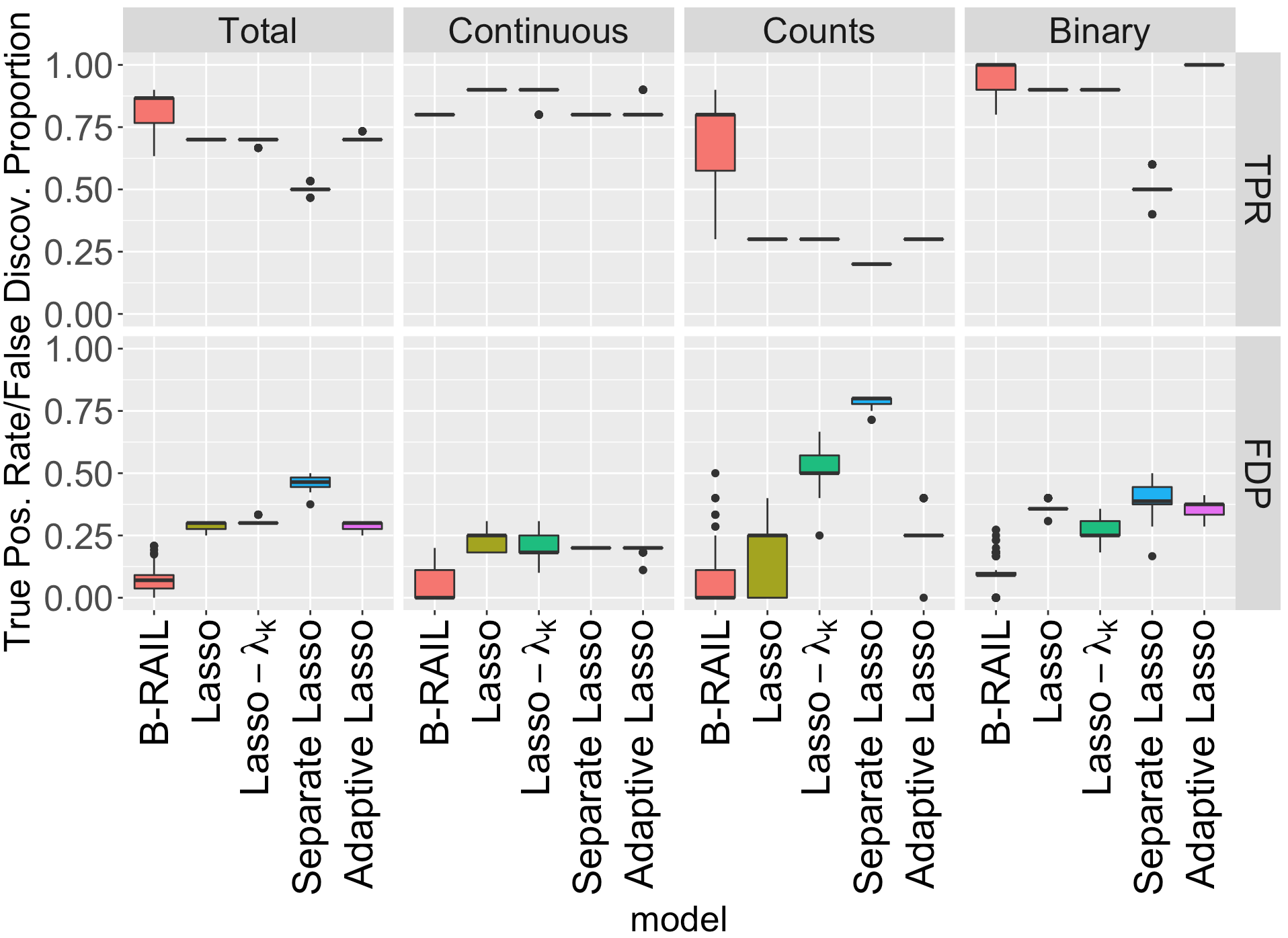}
  \captionof{figure}{Binary response, Block directed graph structure}
\end{subfigure}&
\hspace*{\fill} 
\begin{subfigure}[b]{0.45\linewidth}
\includegraphics[width=\textwidth]{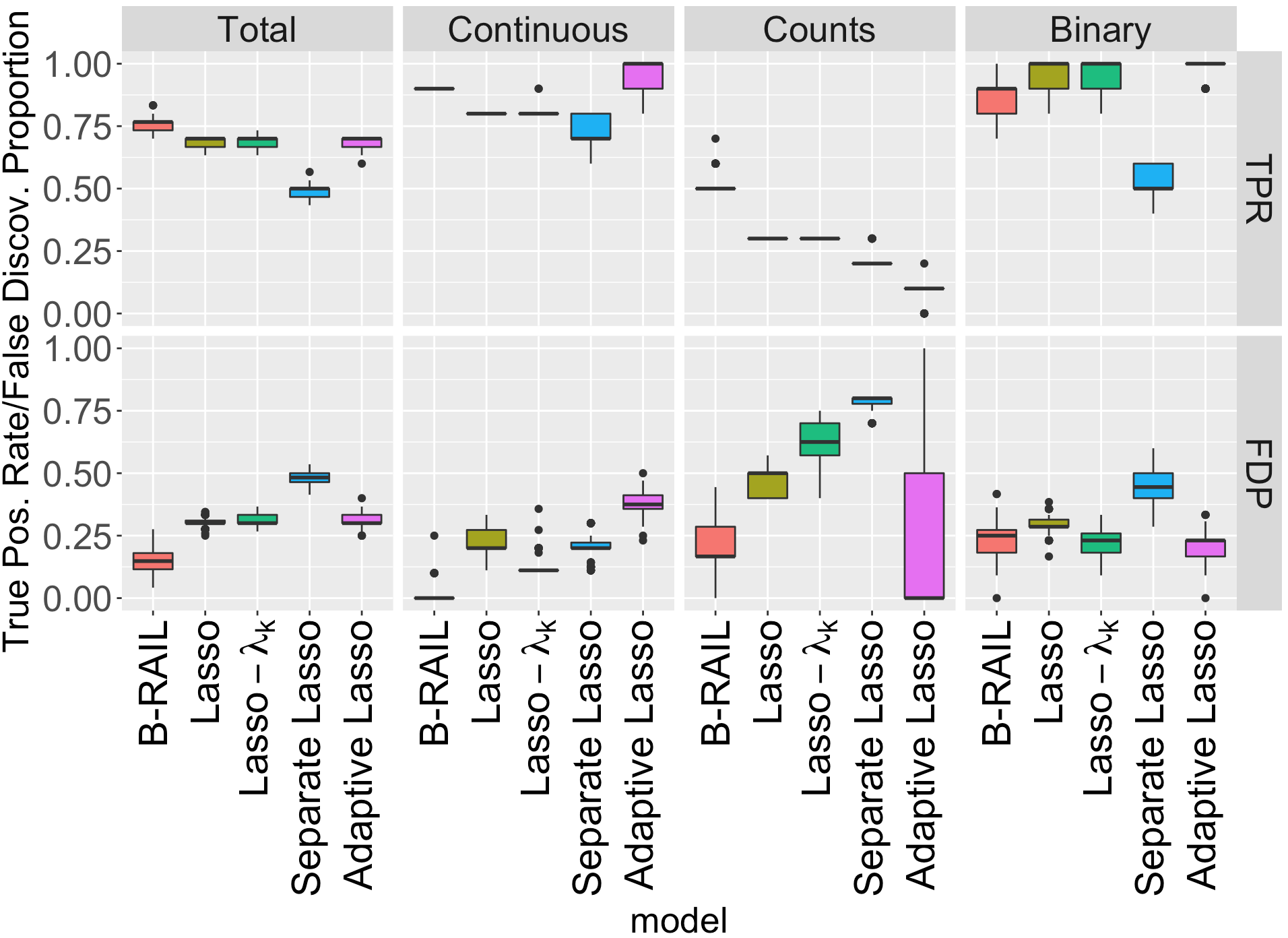}
  \captionof{figure}{Poisson response, Block directed graph structure}
\end{subfigure}\\
\hspace*{\fill} 
\begin{subfigure}[b]{0.45\linewidth}
\includegraphics[width=\textwidth]{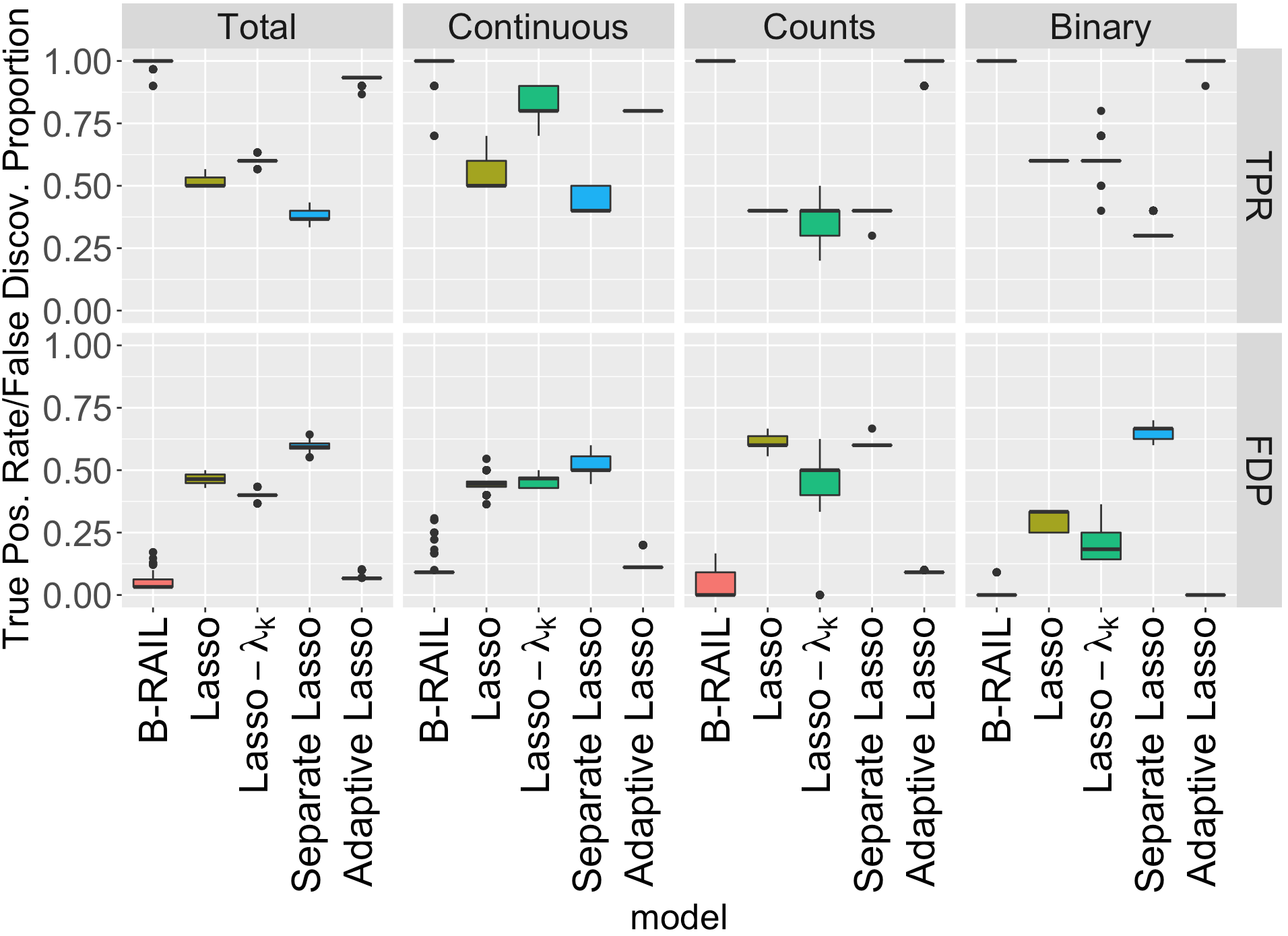}
  \captionof{figure}{Binary response, OV data}
\end{subfigure}&
\hspace*{\fill} 
\begin{subfigure}[b]{0.45\linewidth}
\includegraphics[width=\textwidth]{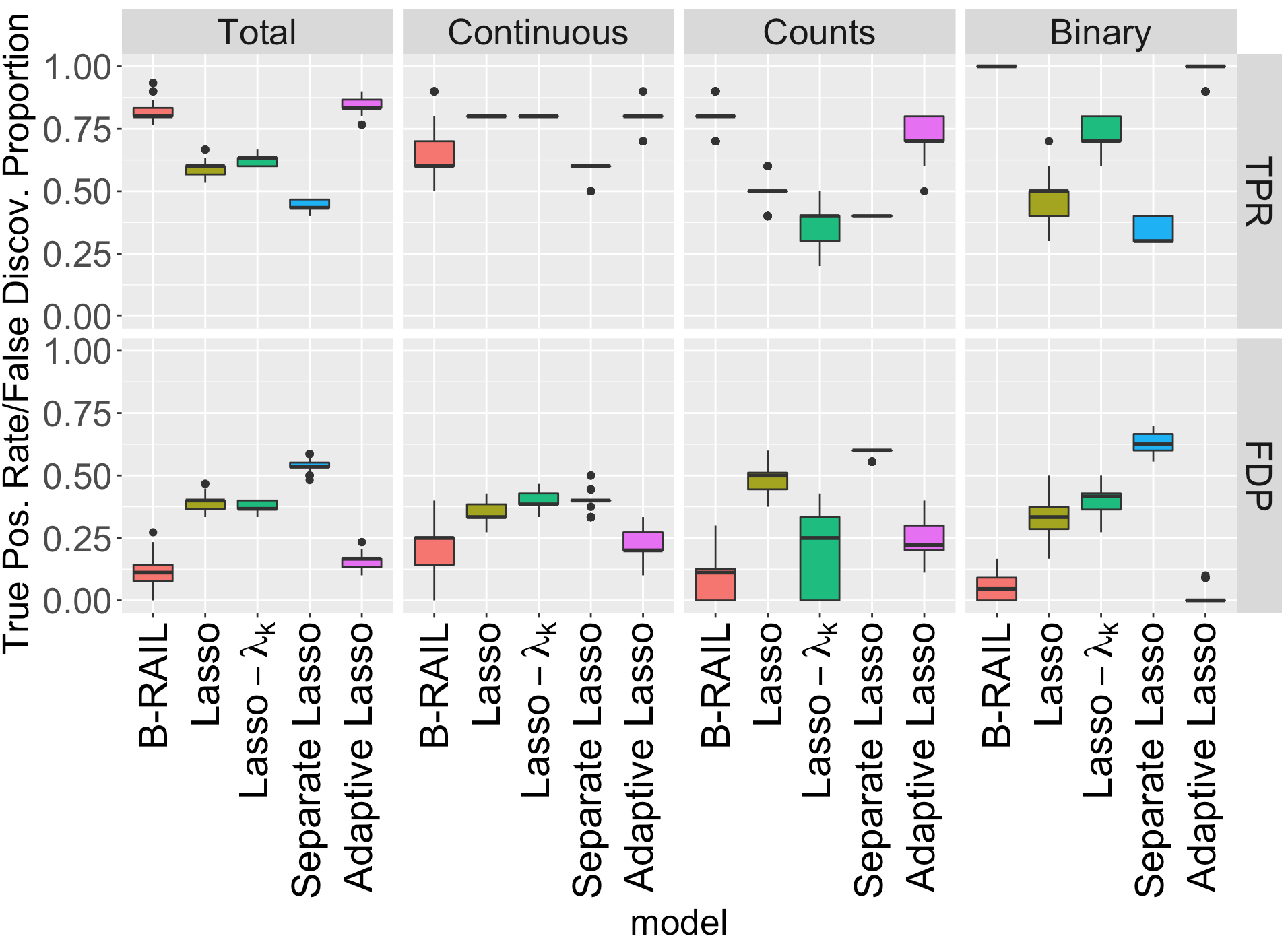}
  \captionof{figure}{Poisson response, OV data}
\end{subfigure}\\
  \end{tabular}    
  \caption{We compare various selection methods under 3 different simulation scenarios. For each scenario, we simulate $\X$ with three blocks (continuous, binary, counts) and either a binary response (left column of subplots) or Poisson response (right column of subplots). We report the TPR and FDP for overall feature recovery and individual block recoveries across 200 runs. Note that we used oracle information for the Lasso-type methods.}
  \label{fig:fig4}
\end{figure}

\begin{table*}[h]
\makebox[\linewidth]{
\scalebox{0.7}{
\begin{tabular}{l  ll  ll  ll  ll}


\addlinespace
\multicolumn{9}{c}{ \begin{normalsize}  Gaussian response \end{normalsize} } \\
\midrule


  & \multicolumn{2}{c}{Total} &  \multicolumn{2}{c}{Continuous} & \multicolumn{2}{c}{ Binary} & \multicolumn{2}{c}{Counts}  \\ 
 \cmidrule(lr){2-3}
\cmidrule(lr){4-5}
\cmidrule(lr){6-7}
\cmidrule(lr){8-9}

 & \multicolumn{1}{c}{TPR} &  \multicolumn{1}{c}{FDP}  &  \multicolumn{1}{c}{TPR}  &  \multicolumn{1}{c}{FDP} &  \multicolumn{1}{c}{TPR}  & \multicolumn{1}{c}{FDP}  &  \multicolumn{1}{c}{TPR}  &  \multicolumn{1}{c}{FDP} \\ 
\midrule

MCP (oracle) & \bft{1.00 (5.2e-4)} & \bft{0.00 (5.2e-4) }& 1.00 (5.0e-4) & 0.00 (0.0e-0) & 1.00 (0.0e-0) & 0.00 (4.5e-4) & 1.00 (1.3e-3) & 0.00 (1.3e-3) \\ 
SCAD (oracle) & 0.77 (5.2e-3) & 0.23 (5.2e-3) & 0.92 (2.8e-3) & 0.18 (4.7e-3) & 1.00 (0.0e-0) & 0.17 (5.1e-3) & 0.38 (1.4e-2) & 0.43 (1.1e-2) \\ 
B-RAIL & 0.88 (4.8e-3) & 0.16 (1.1e-2) & 0.79 (4.8e-3) & 0.12 (1.4e-2) & 0.96 (6.8e-3) & 0.12 (7.0e-3) & 0.89 (5.2e-3) & 0.22 (1.4e-2) \\ 

\addlinespace
\multicolumn{9}{c}{ \begin{normalsize}  Binary response \end{normalsize} } \\
\midrule


  & \multicolumn{2}{c}{Total} &  \multicolumn{2}{c}{Continuous} & \multicolumn{2}{c}{ Binary} & \multicolumn{2}{c}{Counts}  \\ 
 \cmidrule(lr){2-3}
\cmidrule(lr){4-5}
\cmidrule(lr){6-7}
\cmidrule(lr){8-9}

 & \multicolumn{1}{c}{TPR} &  \multicolumn{1}{c}{FDP}  &  \multicolumn{1}{c}{TPR}  &  \multicolumn{1}{c}{FDP} &  \multicolumn{1}{c}{TPR}  & \multicolumn{1}{c}{FDP}  &  \multicolumn{1}{c}{TPR}  &  \multicolumn{1}{c}{FDP} \\ 
\midrule

  MCP (oracle) & 0.49 (2.7e-3) & 0.47 (2.7e-3) & 0.65 (4.1e-3) & 0.27 (9.7e-3) & 0.56 (5.2e-3) & 0.55 (5.6e-3) & 0.26 (3.7e-3) & 0.52 (1.6e-2) \\ 
   SCAD (oracle) & 0.50 (6.7e-3) & 0.50 (6.8e-3) & 0.58 (8.1e-3) & 0.49 (4.4e-3) & 0.61 (1.3e-2) & 0.53 (9.5e-3) & 0.30 (2.7e-3) & 0.40 (1.4e-2) \\ 
  B-RAIL & \bft{0.75 (3.5e-3)} & \bft{0.15 (5.0e-3)} & 0.90 (0.0e-0) & 0.01 (3.6e-3) & 0.83 (8.9e-3) & 0.23 (7.2e-3) & 0.51 (3.8e-3) & 0.19 (1.1e-2) \\ 

\addlinespace
\multicolumn{9}{c}{ \begin{normalsize}  Poisson response \end{normalsize} } \\
\midrule


  & \multicolumn{2}{c}{Total} &  \multicolumn{2}{c}{Continuous} & \multicolumn{2}{c}{ Binary} & \multicolumn{2}{c}{Counts}  \\ 
 \cmidrule(lr){2-3}
\cmidrule(lr){4-5}
\cmidrule(lr){6-7}
\cmidrule(lr){8-9}

 & \multicolumn{1}{c}{TPR} &  \multicolumn{1}{c}{FDP}  &  \multicolumn{1}{c}{TPR}  &  \multicolumn{1}{c}{FDP} &  \multicolumn{1}{c}{TPR}  & \multicolumn{1}{c}{FDP}  &  \multicolumn{1}{c}{TPR}  &  \multicolumn{1}{c}{FDP} \\ 
\midrule

  MCP (oracle) & 0.86 (1.3e-3) & 0.13 (1.5e-3) & 0.90 (1.6e-3) & 0.08 (5.3e-3) & 1.00 (0.0e-0) & 0.18 (4.1e-3) & 0.70 (3.6e-3) & 0.07 (6.1e-3) \\ 
  SCAD (oracle) & 0.65 (1.9e-3) & 0.34 (2.0e-3) & 0.81 (1.8e-3) & 0.28 (3.7e-3) & 0.84 (4.5e-3) & 0.25 (4.1e-3) & 0.30 (1.5e-3) & 0.56 (5.0e-3) \\ 
  B-RAIL & \bft{0.81 (8.6e-3)} & \bft{0.07 (5.2e-3)} & 0.80 (0.0e-0) & 0.04 (5.4e-3) & 0.97 (5.7e-3) & 0.09 (6.2e-3) & 0.68 (2.1e-2) & 0.07 (1.2e-2) \\ 

\specialrule{1pt}{1pt}{1pt}
\end{tabular}
}}
\vspace{3pt}
\caption{We compare B-RAIL to MCP and SCAD for three types of responses and three types of covariates (continuous $\X_1$, counts $\X_2$, and binary $\X_3$) under the block directed graph simulation design. Here $n=200$, $p_{1}=p_{2}=p_{3}=300$, and $||\beta||_0=10$ in each block. Oracle model selection was used for MCP and SCAD penalties. For B-RAIL, we used stability selection, outlined in Section \ref{sec:BRAIL}, thresholded at $0.8$. We report the TPR and FDP for overall recovery and for each block separately across 200 runs. For the binary and Poisson responses, MCP and SCAD do not perform as well as B-RAIL.} \label{tab:SCAD_MCP}
\end{table*}

In Table \ref{tab:SCAD_MCP}, we compare B-RAIL to the non-convex penalties, MCP and SCAD, under the block directed graph simulation design with three different types of responses. We see that MCP performs well with Gaussian responses, and in fact, the MCP penalty can be used instead of a Lasso penalty in the B-RAIL algorithm for Gaussian responses. However, the B-RAIL algorithm outperforms MCP and SCAD for non-Gaussian responses. We thus chose to use a Lasso penalty when introducing the B-RAIL algorithm for consistency.  

\begin{figure}[H]
\centering
\begin{tabular}[c]{cc}
\begin{subfigure}[b]{0.45\linewidth}
\includegraphics[width=\textwidth]{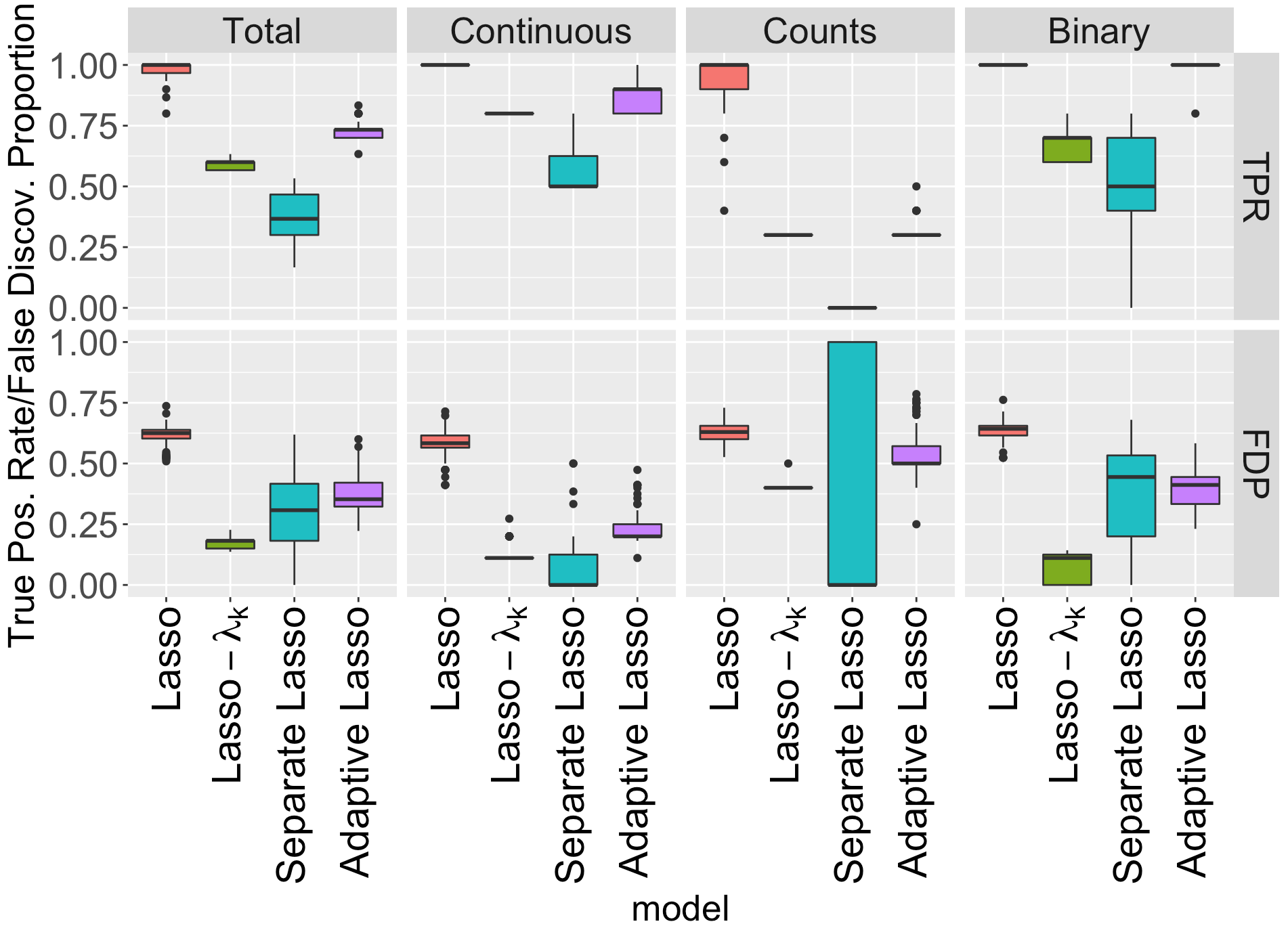}
  \captionof{figure}{5 fold cross validation }
\end{subfigure}&
\hspace*{\fill} 
\begin{subfigure}[b]{0.45\linewidth}
\includegraphics[width=\textwidth]{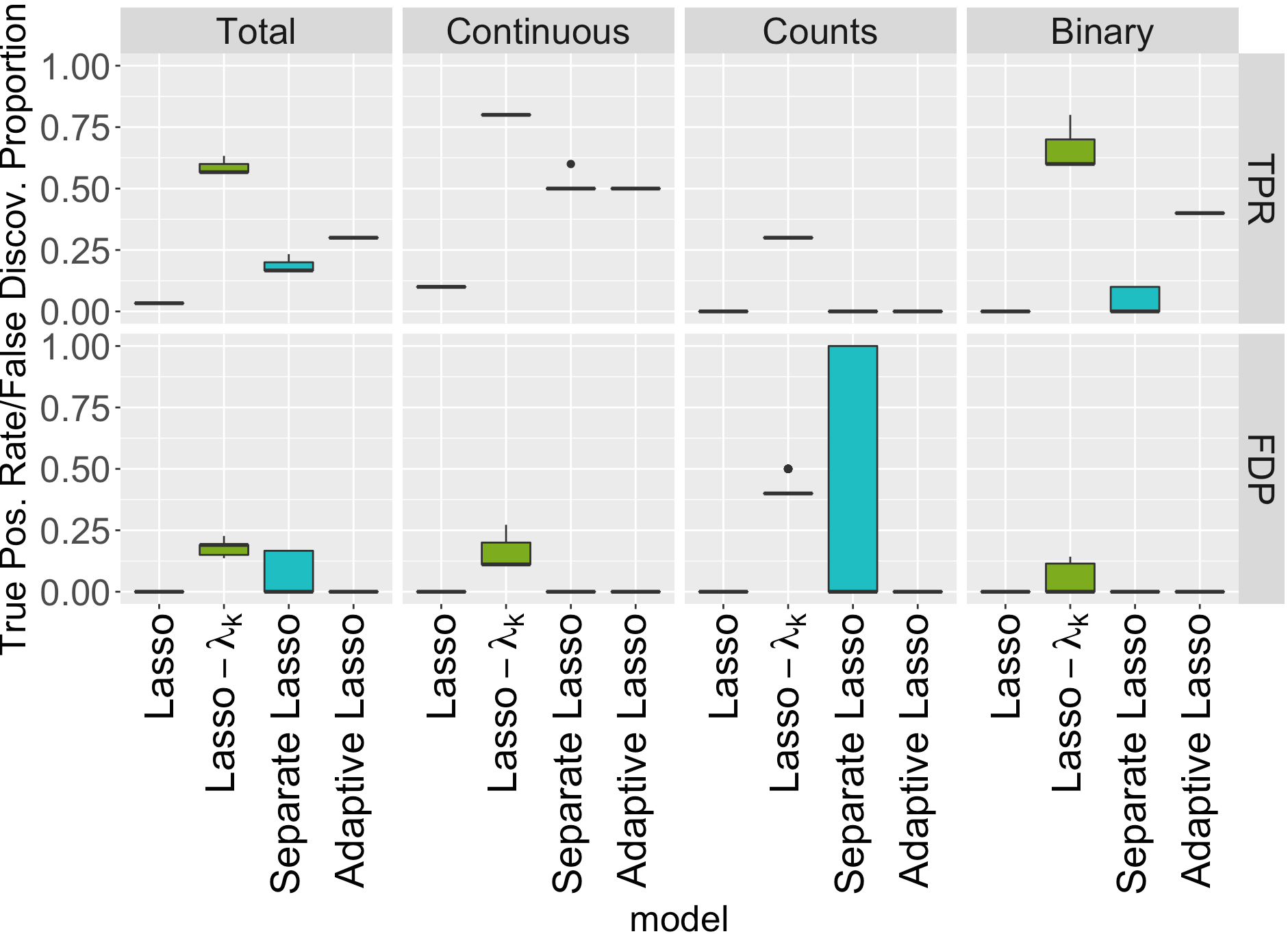}
  \captionof{figure}{Extended BIC}
\end{subfigure}\\
\hspace*{\fill} 
\begin{subfigure}[b]{0.45\linewidth}
\includegraphics[width=\textwidth]{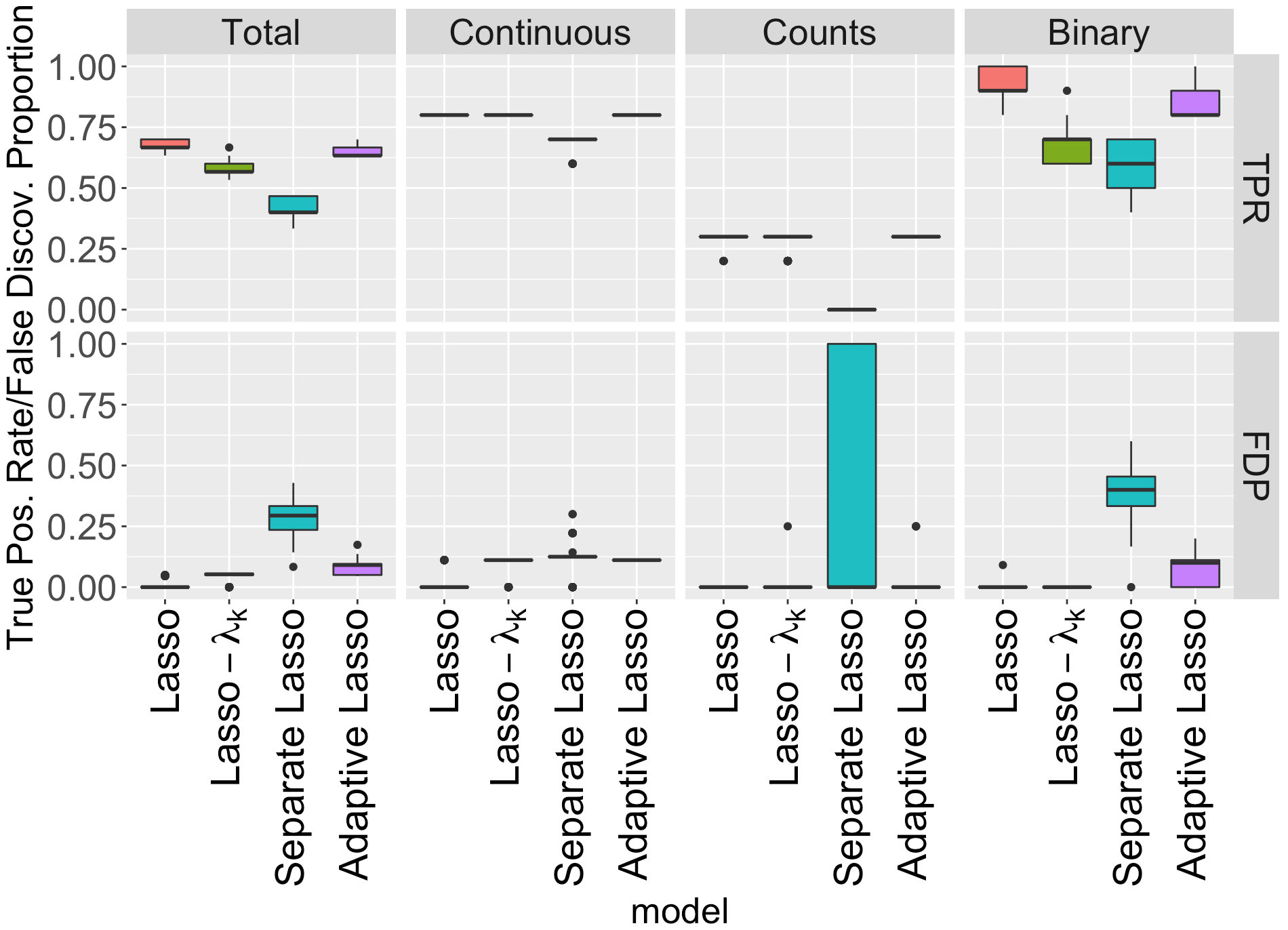}
  \captionof{figure}{Stability Selection ($\tau=.75$)}
\end{subfigure}&
\hspace*{\fill} 
\begin{subfigure}[b]{0.45\linewidth}
\includegraphics[width=\textwidth]{GausMGM.png}
  \captionof{figure}{Oracle selection}
\end{subfigure}\\
  \end{tabular}    
  \caption{We compare feature recovery for B-RAIL and Lasso-type selection methods using oracle information, 5-fold CV, extended BIC, and stability selection to select the regularization parameters. Here, we simulate from the block directed graph simulation design with Gaussian responses. We report the TPR and FDP for overall feature recovery and individual block recoveries across 200 runs.}\label{fig:CompBoxPots}
\end{figure}

Figure~\ref{fig:CompBoxPots} provides the same information as Table~\ref{tab:Methods} but using boxplots for easier visualization. While the model selection techniques (i.e. CV, extended BIC, and stability selection) give lower values of $TPR * (1 - FDP)$ than their oracle selection counterparts, B-RAIL outperforms even the oracle selection methods and yields the highest $TPR * (1 - FDP)$.

\end{document}